\documentclass[11pt,a4paper,reqno]{amsart}%
\usepackage{amsthm,amsmath,amsfonts,amssymb,xcolor,amsxtra,appendix,bookmark,dsfont,bm,mathrsfs}
\usepackage[latin1]{inputenc}
\usepackage{color} 
\usepackage{amssymb}\usepackage{graphicx}
\usepackage{tikz}
\usepackage{hyperref}
\theoremstyle{plain}
\newtheorem{theorem}{Theorem}
\newtheorem{lemma}[theorem]{Lemma}

\newtheorem{proposition}[theorem]{Proposition}

\newtheorem{definition}{Definition}

\theoremstyle{remark}
\newtheorem{remark}[theorem]{Remark}

\allowdisplaybreaks

\title[Condensation and non-condensation times]{Condensation and non-condensation times for   4-wave kinetic equations}

\author[G. Staffilani]{Gigliola Staffilani
}
\address{Department of Mathematics, Massachusetts Institute of Technology, Cambridge, MA 02139, USA}
\email{gigliola@math.mit.edu} 
\thanks{G.S. is  funded in part by  the NSF grants DMS-2052651, DMS-2306378 and the Simons Foundation through the Simons Collaboration on Wave Turbulence.}

\author[M.-B. Tran]{Minh-Binh Tran}
\address{Department of Mathematics, Texas A\&M University, College Station, TX 77843, USA}
\email{minhbinh@tamu.edu} 
\thanks{M.-B. T is  funded in part by  a   Humboldt Fellowship,   NSF CAREER  DMS-2303146, and NSF Grants DMS-2204795, DMS-2305523,  DMS-2306379.}

\begin{document}
\date{\today}

\begin{abstract} Inspired by the pioneering work of Escobedo and Velazquez \cite{EscobedoVelazquez:2015:FTB,EscobedoVelazquez:2015:OTT}, we prove that solutions of 4-wave kinetic equations, under very general forms of  dispersion relations, develop   condensation at the origin in finite time,  under weaker conditions on the initial data  than the ones considered in \cite{EscobedoVelazquez:2015:FTB,EscobedoVelazquez:2015:OTT}. We also provide some estimates on the non-condensation times of the solutions. 

\end{abstract}

\maketitle

 \tableofcontents

\section{Introduction}\label{intro} 
The wave turbulence theory, that describe  the dynamics of   nonlinear waves  out of  thermal equilibrium, 
  has its  origin in the works of  Peierls \cite{Peierls:1993:BRK,Peierls:1960:QTS}, Brout-Prigogine \cite{brout1956statistical},  Zaslavskii-Sagdeev \cite{zaslavskii1967limits}, Hasselmann \cite{hasselmann1962non,hasselmann1974spectral},  Benney-Saffman-Newell \cite{benney1969random,benney1966nonlinear},  Zakharov \cite{zakharov2012kolmogorov}. Having a vast range of physical applications, from  Alfv\'en wave turbulence in solar wind to planetary Rossby waves, from water surface gravity and capillary waves to inertial waves due to rotation and internal waves on density stratifications, from Bose-Einstein condensates (BECs) to nonlinear optics, this theory has  a common mathematical framework that   describes the dynamics of spectral energy transfer through  probability densities associated with weakly non-linear  wave systems. These probability densities satisfy the so-called  wave kinetic equations. We refer to books of  Zakharov, Lvov, Falkovich \cite{zakharov2012kolmogorov} and Nazarenko
\cite{Nazarenko:2011:WT}, and the review paper of Newell and Rumpf \cite{newell2011wave} for more discussion on the topic.

The focus of our work is the so-called   4-wave kinetic equation 
\begin{eqnarray}\label{4wave}
		\partial_t f &\ = \ & 	\mathfrak C \left[ f\right],\ \ \ f(0,k)=f_0(k),\ \ \ k\in\mathbb{R}^3,\ \ \ t\in [0,\infty),\\ 
		\mathfrak C \left[ f\right]  
		&\ =\ & \iiint_{\mathbb{R}^{9}}\mathrm{d}k_1\,\mathrm{d}k_2\,\mathrm{d}k_3  \delta(k+k_1-k_2-k_3)\delta(\omega + \omega_1 -\omega_2 - \omega_3)[f_2f_3(f_1+f)-ff_1(f_2+f_3)] ,\nonumber
\end{eqnarray}
where $\omega,\omega_1,\omega_2,\omega_3$  is the shorthand notation for  $\omega(k), \omega(k_1), \omega(k_2), \omega(k_3)$, and $f,f_1,f_2,f_3$ is the shorthand notation for $f(k), f(k_1), f(k_2), f(k_3)$. 

This equation has been rigorously studied in the pioneering and ground breaking works of Escobedo-Velazquez \cite{EscobedoVelazquez:2015:FTB,EscobedoVelazquez:2015:OTT}  in which several important and fundamental results have been obtained for the case $\omega(k)=|k|^2$. A local wellposedness theory of the equation has later been studied in \cite{germain2017optimal}. The stability and cascades for the Kolmogorov-Zakharov spectrum of the equation have been obtained in \cite{collot2024stability}. The near equilibrium stability and instability has recently been investigated in \cite{menegaki20222,escobedo2024instability}. The local well-posedness for the MMT model as well as the   convergence rates of discrete solutions have been studied in \cite{germain2023local} and \cite{dolce2024convergence}. 

The current work is the second paper in our series of 2 papers that are devoted to the study  of equation \eqref{4wave}, in which the dispersion relation $\omega$ is allowed to take a very general form. In the first work \cite{StaffilaniTranCascade1}, we have provided a local well-posedness result of mild solutions for the equation as well as a proof of the transfer of energies to higher wave numbers. In the current work, we prove that there is a finite time condensation (or equivalently, the appearance of a delta function at the origin) for a more general class of initial data than the one considered in  the previous work \cite{EscobedoVelazquez:2015:FTB,EscobedoVelazquez:2015:OTT}, and with   more general dispersion relations. We refer to Remark  \ref{Remark2}  for the comparisons between the results. In brief, the dispersion relation considered in \cite{EscobedoVelazquez:2015:FTB,EscobedoVelazquez:2015:OTT}  corresponds to the special case $\alpha=2$ (see \eqref{Settings2}-\eqref{Settings3}  and Remark \ref{remarkalpha} below  for the definition of $\alpha$) and the initial conditions considered in \cite{EscobedoVelazquez:2015:FTB,EscobedoVelazquez:2015:OTT} require  a concentration around the origin (see Figure \ref{Fig2}), while our initial conditions do not need this strong concerntration (see Figure \ref{Fig1} and Remark \ref{Remark2}).  The proof of our main theorem is based on a domain decomposion technique, using a numerical ``divide and conquer"  idea   (see \cite{halpern2009nonlinear,Lions:1989:OSA,toselli2004domain}), in which   the complex operators  are decomposed into    independent operators on smaller subdomains as well as some growth lemmas used in the theory of non-local differential equations. In addition, the use of domain decomposition allows us to consider more general types of dispersion relations with a wider range for the parameter $\alpha$: $\alpha\in(1,2]$. Note that the special case $\alpha=1$ is not expected to have a finite time condensation since the dispersion relation in this case is a straight line. The solutions in the case $\alpha>2$ will have slightly different behaviors that will be reported  in a separate work.  We are also able to provide some estimates on the length  of the time, at which there is no concentration of solutions around the origin.

Besides the 4-wave kinetic equations, also the 3-wave kinetic equations  play  an important role in the theory of weak turbulence, and they have  been studied    in \cite{GambaSmithBinh} for stratified  flows in the ocean, in \cite{nguyen2017quantum,soffer2020energy,walton2022deep,walton2023numerical,walton2024numerical} for capillary waves, in \cite{cortes2020system,EPV,escobedo2023linearized1,escobedo2023linearized,ToanBinh,nguyen2017quantum,soffer2018dynamics} for Bose-Einstein Condensates, in \cite{AlonsoGambaBinh,CraciunBinh,EscobedoBinh,GambaSmithBinh,tran2020reaction,PomeauBinh} for the phonon interactions in anharmonic crystal lattices and in \cite{rumpf2021wave} for beam waves. In rigorously deriving wave kinetic equations,  significant progresses have been made and we refer to   \cite{buckmaster2019onthe,buckmaster2019onset,collot2020derivation,collot2019derivation,deng2019derivation,deng2021propagation,deng2023long,deng2022wave,dymov2019formal,dymov2019formal2,dymov2020zakharov,dymov2021large,germain2024universality,grande2024rigorous,hani2023inhomogeneous,hannani2022wave,LukkarinenSpohn:WNS:2011,staffilani2021wave,ma2022almost} and the references therein.

{\bf Acknowledgment} We would like to express our gratitude to Prof. J. J.-L. Velazquez, Prof. H. Spohn, Prof. J. Lukkarinen, Prof. M. Escobedo for fruitful discussions on the topics.
\section{The Settings}
We suppose that for $k\in\mathbb{R}^3$, $\omega(k)=\omega(|k|)$ is convex in $|k|$. Moreover,
\begin{equation}
	\label{Settings3}\begin{aligned}
		&\omega(|k|) \ge C_\omega|k|^\alpha, \forall k\in\mathbb{R}^d,\\
		&\omega(|k|) \le C_\omega'|k|^{\alpha'}, \forall k\in\mathbb{R}^d,|k|<1, 
	\end{aligned}
\end{equation}  
	for  constants $2\ge\alpha >1,\alpha\ge \alpha'\ge 1,C_\omega,C_\omega'>0$. 
We set
\begin{equation}
	\label{Settings1}\mho=\frac{|k|}{\omega'(|k|)}.
\end{equation}  
We assume that $\mho$ is a function from $[0,\infty)$ to $[0,\infty)$ of the variables $\omega$ and $|k|$. Moreover,   $\mho$ is assume to be increasing in $|k|$ and $\omega$ 
	and for  sufficiently small $ 1>\omega\ge 0$
	\begin{equation}\
		\label{Settings2}\begin{aligned}& \ \check{C}_\mho\omega^\beta \ \le \ \mho(\omega),\end{aligned}\end{equation}
		where $\check{C}_\mho\ge0,\frac{2-\alpha}{\alpha}\ge\beta\ge 0$ are constants independent of $k$. Moreover, \begin{equation}
		\label{Settings2a} \mho\le {C}_\mho^1 |k|^\iota, \mbox{ for all } k, \end{equation}   where ${C}_\mho^1\ge0,1\ge \iota\ge 0$ are  constants independent of $k$.


 We also define
 \begin{equation}
 	\label{FDefinition} F(t,k)=f(t,k)|k|\mho(k),
 \end{equation}
	\begin{equation}\label{MassEnergy}
	\begin{aligned}
		\mathscr M  \ =	\ &	\int_{\mathbb{R}^3}\mathrm{d}k  f(0,k),
	\  \ \mathscr E\ =	\ &	\int_{\mathbb{R}^3}\mathrm{d}k  f(0,k)  \omega,
	\end{aligned}
\end{equation}
 and the non-condensation time set 
  \begin{equation}	\label{NoCondensateTime}
 \digamma \ := \ {\Big\{}t\in[0,\infty)~~~ \large|~~~ \int_{\{0\}}\mathrm{d}k f(t,k)= 0 {\Big\}}.
  \end{equation}

 Next, we define the following subset of the non-condensation time set
 \begin{equation}\label{NoCondensateTime1}
 	\begin{aligned}
 		\digamma^*_{\mathscr N,\theta,R}\ := & \ \Big\{t\in\digamma~~~ \big|~~~ \forall i\in\{0,\cdots,\mathscr N-3\},\\ & \int_{[iR/\mathscr N,(i+3)R/\mathscr N)}\mathrm{d}\omega F(t) < (1-\theta)\int_{[0,R)}\mathrm{d}\omega F(t)\Big\},
 	\end{aligned}
 \end{equation}
for $1/100>\theta>0,R>0,\mathscr{N}\in{\mathbb{N}}, \mathscr{N}\ge 1000$.
\begin{remark}\label{Remark1}
{\it	  $[0,\infty)\backslash\digamma$ is the set of all times $t$ such that $f$ is concentrated around the origin, namely
$$	\int_{\{0\}}\mathrm{d}k f(t,k)  \ne 0.$$

	The set $\digamma^*_{\mathscr N,\theta,R}$ describes the times $t\in\digamma$ such that $F$ is spread out to all   intervals $[iR/\mathscr N,(i+3)R/\mathscr N)$ of $[0,R)$ instead of concentrating in one small interval.}
\end{remark}
 For 3 numbers $x,y,z$ in $\mathbb{R}$, we will also use the notations
\begin{equation}\label{Mid}
\mathrm{mid}\left\{ x,y,z\right\}:=t\in\left\{ x,y,z\right\} \backslash\{\max\left\{ x,y,z\right\},\min\left\{ x,y,z\right\}\}.
\end{equation} 
Let us define  
\begin{equation}\label{T0}
	T_0 \ := \ \sup\Big\{T ~~ \Big| ~~~ \int_{\{0\}}\mathrm{d}k f(t,k) = 0, ~~~~~~\forall t\in[0,T)\Big\}
\end{equation}
which, in the case that $T_0<\infty$, can be seen as the first condensation time.

\begin{remark}\label{remarkalpha}
	There are many physical examples of $\omega$ that satisfies the above conditions:		The  dispersion relation $\omega(k)=|k|^\alpha$ with $1<\alpha\le 2$. In this case $\mho=\frac1\alpha |k|^{2-\alpha}=\frac1\alpha\omega^{\frac{2-\alpha}{\alpha}}$. Note that the dispersion relation considered in the fundamental works \cite{EscobedoVelazquez:2015:FTB,EscobedoVelazquez:2015:OTT} corresponds to the special case $\alpha=2,$ $\beta=0$. 
	
\end{remark}

We recall the following definition of \cite{StaffilaniTranCascade1}. 
\begin{definition}\cite[Definition 1]{StaffilaniTranCascade1}
	We say that $f(t,k)=f(t,|k|)\in C^1([0,\infty), L^1(\mathbb{R}^3))$ is a global mild radial solution of \eqref{4wave} with a radial initial condition $f_0(k)=f_0(|k|) \ge0$ if $f(t,k)\ge0$ and
	for all $\phi\in C_c^2([0,\infty))$, we have
	\begin{equation}\label{4wavemild}
		\begin{aligned}
			\int_{\mathbb{R}^3}\mathrm{d}k f(t,k)\phi(|k|)\ = \ & 	\int_{\mathbb{R}^3}\mathrm{d}k f(0,k)\phi(|k|) \  + \ \int_0^t\int_{\mathbb{R}^3}\mathrm{d}k	\mathfrak C \left[ f\right]\phi(|k|),
		\end{aligned}
	\end{equation}
	for all $t\in\mathbb{R}_+$. 
\end{definition}

The main theorem of \cite{StaffilaniTranCascade1} states as follows.
\begin{theorem}
	\label{Theorem1} Suppose that $\omega$ satisfies the above assumptions (Assumption \eqref{Settings2} is not needed). 
	Let $f_0(k)=f_0(|k|)\ge 0$ be an initial condition satisfying 
	\begin{equation}
		\label{Theorem1:1} \int_{\mathbb{R}^3}\mathrm{d}k f_0(k) \ = \ \mathscr{M}, \ \ \ \ \ \int_{\mathbb{R}^3}\mathrm{d}k f_0(k) \omega(k) \ = \ \mathscr{E}. 
	\end{equation}
	There exists at least a global mild radial solution $f(t,k)$ of \eqref{4wave} in the sense of \eqref{4wavemild} such that
	\begin{equation}
		\label{Theorem1:2} \int_{\mathbb{R}^3}\mathrm{d}k f(t,k) \ = \ \mathscr{M},   \ \ \ \ \ \int_{\mathbb{R}^3}\mathrm{d}k f(t,k) \omega(k) \ = \ \mathscr{E},
	\end{equation}
	for all $t\ge 0$. Suppose further that 
	\begin{itemize}
		\item [(I)] When $\omega(k)=|k|^2$, the support of $f_0$ has a non-empty interior;
		\item [(II)] When $\omega(k)\ne|k|^2$,   the origin belongs to the interior of the support of $f_0$.
	\end{itemize}
	For any $\mathfrak R>0$, we have the following energy cascade phenomenon  
	\begin{equation}\label{Theorem1:3}\begin{aligned}
			&\lim_{t\to\infty}\int_{\mathbb{R}^3}^\infty	  d\omega f\left(t,k\right)   \omega(k) \chi_{B(O,\mathfrak R)}(k)
			\ = \ 0.\end{aligned}
	\end{equation}

\end{theorem}
Our main theorem states as follows.
\begin{theorem}
	\label{Theorem2} Suppose that $\omega$ satisfies the above assumptions. 
	Let $f_0(k)=f_0(|k|)\ge 0$ be an initial condition satisfying 
	\begin{equation}
		\label{Theorem2:1} \int_{\mathbb{R}^3}\mathrm{d}k f_0(k) \ = \ \mathscr{M}, \ \ \ \ \ \int_{\mathbb{R}^3}\mathrm{d}k f_0(k) \omega(k) \ = \ \mathscr{E}, \ \ \ \ \ \int_{\{0\}}\mathrm{d}k f_0(k) \ = \ 0. 
	\end{equation}
Suppose that there exists   a global mild radial solution $f(t,k)$ of \eqref{4wave} in the sense of \eqref{4wavemild} such that
	\begin{equation}
		\label{Theorem2:2} \int_{\mathbb{R}^3}\mathrm{d}k f(t,k) \ = \ \mathscr{M},   \ \ \ \ \ \int_{\mathbb{R}^3}\mathrm{d}k f(t,k) \omega(k) \ = \ \mathscr{E},
	\end{equation}
	for all $t\ge 0$.  
	
\begin{itemize}
	\item[(i)]
	 Suppose that there exist  a  sufficiently large number $N_0$ and a small  constant  $\mathscr{C}_F^*>0$ such that for $F$ defined in \eqref{FDefinition} 
	\begin{equation}	\label{Theorem2:4}
		\begin{aligned}
			&    \int_{0}^{2^{-N_0-1}}\mathrm{d}\omega F(0,\omega) \ge \  \mathscr{C}_F^*. 
	\end{aligned}\end{equation} 
	Moreover, there exist   constants $C_{ini}>0$, $c_{ini}\ge 0$,  such that
	\begin{equation}\label{Theorem2:3}
		\begin{aligned}
			\int_{0}^r\mathrm d\omega F(0,\omega)
			\geq  C_{ini}  {r^{c_{ini}}}, \end{aligned}
	\end{equation}
	for all $2^{-N_0-1}=r_0>r>0$ and 
	Then the first condensation time $T_0$ defined in \eqref{T0} is finite and can be bounded by constants that depends only on $N_0,\mathscr M, \mathscr E$ and $\alpha$. 
	
\item[(ii)] We have

	\begin{equation}\label{Theorem2:5}
		\begin{aligned}
			\frac{{C}_0(m\Re	)^{2+\frac2\alpha}(	\mathscr M +\mathscr E)}{(\Re m/\mathscr N)^{2} \mho(\Re m/\mathscr N) \theta^4\left( \frac{(m-1)   \mathscr{M}_o}{m}\right)^3}
			\ \ge\ &   \mathcal M(\digamma^*_{\mathscr N,\theta, m\Re}),
		\end{aligned}
	\end{equation}
for all $m,\mathscr{N}\in\mathbb{N}$, $\mathscr{N}\ge 1000$, $0<\theta<1/100$, where $\mathcal M$ denotes the Lebesgue measure and
$$\mathscr{M}_o=\int_{0}^{\Re	}\mathrm{d}\omega F(0,\omega).$$

\end{itemize}
\end{theorem}

\begin{remark}\label{Remark2}{\it
	The above theorem says that if   $F(0,\omega)$ is concentrated in a neighborhood near the origin but can be allowed to vanish at the origin, then in finite time, the solution will develop a non-zero measure at the origin in the sense that 
	   $T_0<\infty$ and $T_0$ can be bounded explicitly. Figure \ref{Fig1} illustrates one of such initial conditions. This initial condition vanishes at $0$ and grows smoothly from $0$ to $r_0$, such that \eqref{Theorem2:3} is valid. Moreover, the mass of the solution is sufficiently concentrated in $[0,2^{-N_0-1}]$, that is \eqref{Theorem2:4}. We believe that the domain decomposition technique allows us to consider  initial conditions that are more general than those considered in \cite{EscobedoVelazquez:2015:FTB,EscobedoVelazquez:2015:OTT}.\\
	 In \cite{EscobedoVelazquez:2015:FTB,EscobedoVelazquez:2015:OTT}, it has been proved that for the boundary case $\omega(k)=|k|^2$ (and $\alpha=2,\beta=0$) the finite time condensation phenomenon happens when the initial condition is sufficiently concentrated at  the origin (see Figure \ref{Fig2} and \cite{EscobedoVelazquez:2015:FTB,EscobedoVelazquez:2015:OTT}). This condition can be expressed in  analytic form as (see \cite[Theorem 3.18 and Theorem 3.19]{EscobedoVelazquez:2015:OTT}), 
	 \begin{equation}\label{EVConcen1}
	 \exists r>0, K^*>0, \theta^*>0 \mbox{ such that }  \int_0^R\mathrm{d}\omega F(0,\omega) \gtrsim   R^\frac32 \mbox{ for } 0<R\le r, \int_0^r\mathrm d\omega F(0,\omega) \ge K^*r^{\theta^*}.
	 \end{equation} 
   Note that such assumption on the initial conditions  $F(0,\omega)$ means that $F(0,\omega)\ge   \omega^{1/2}$ near the origin, while for us, we only need $F(0,\omega)\approx   \omega^{c_{ini}-1}$, for any $c_{ini}\ge0$. The condition considered in \cite{EscobedoVelazquez:2015:FTB,EscobedoVelazquez:2015:OTT} corresponds to the case $c_{ini}=3/2$.}
	\begin{figure}
		\centering
		\includegraphics[width=.60\linewidth]{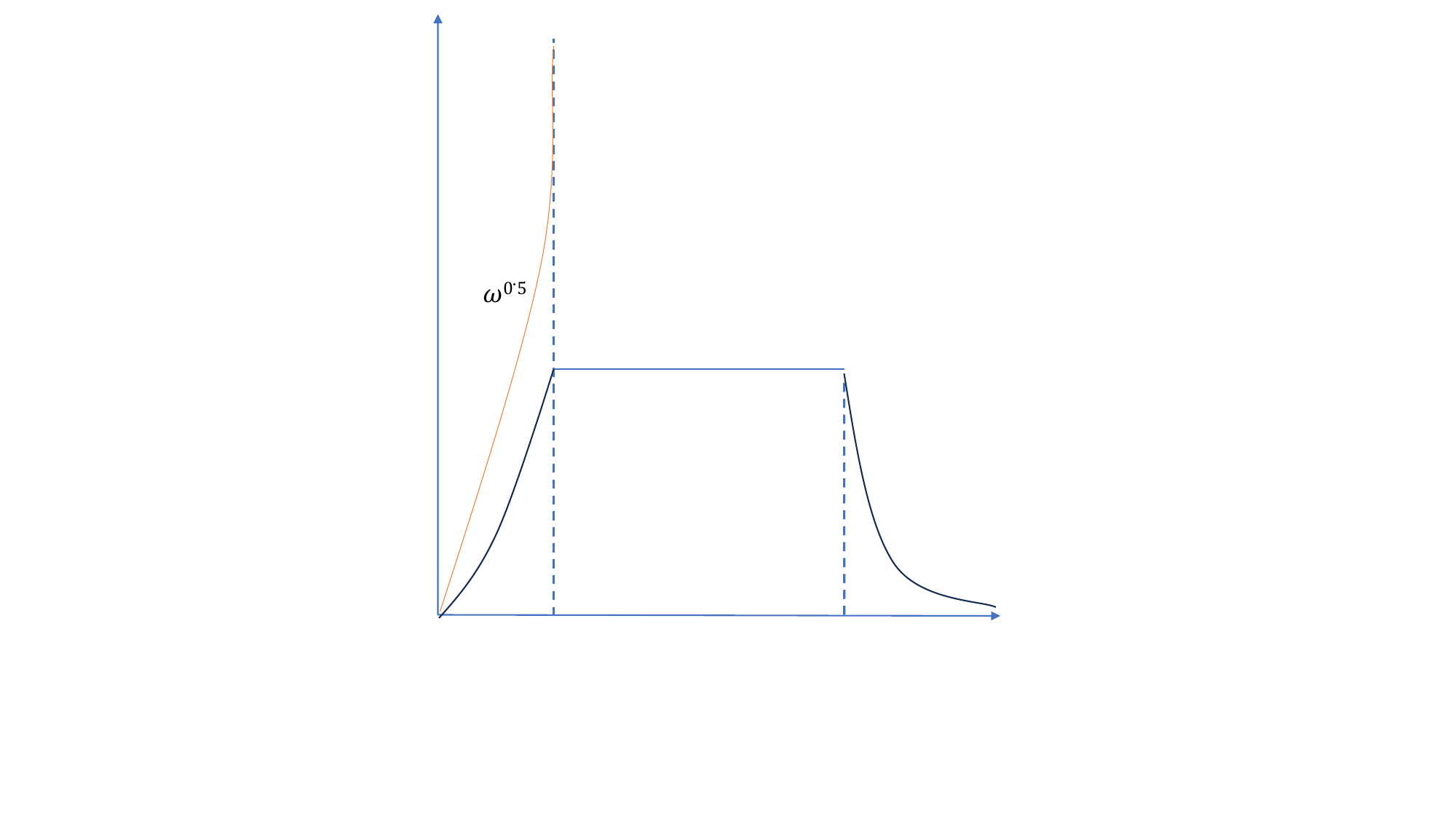}
		\caption{An example of an initial condition that satisfies \eqref{Theorem2:3}-\eqref{Theorem2:4} of Proposition \ref{Propo:FiniteTimeCondensation}. Our initial condition does not require to be above the line   $\omega^{1/2}$. } 
		\label{Fig1}
	\end{figure}
	
	\begin{figure}
		\centering
		\includegraphics[width=.60\linewidth]{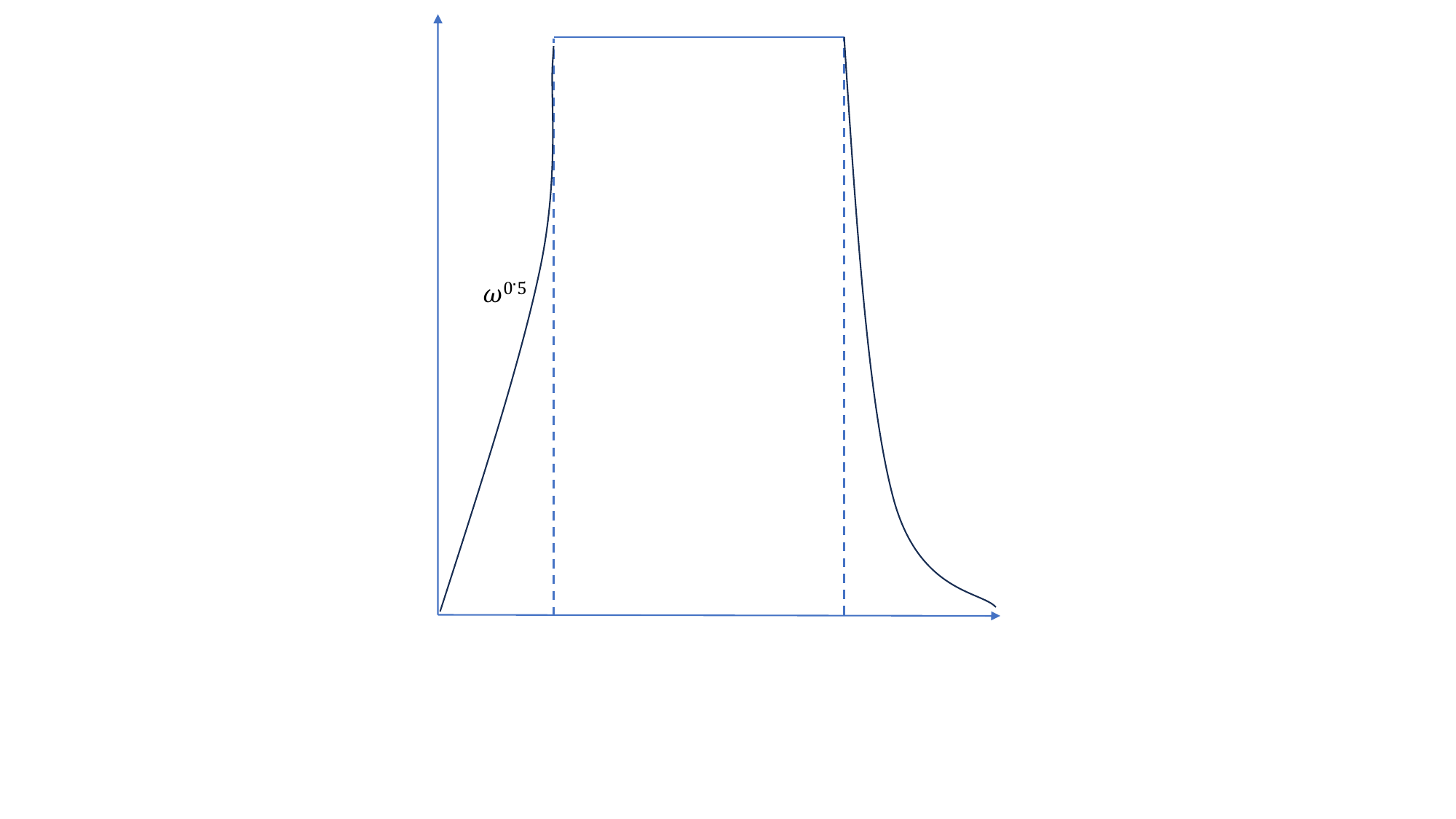}
		\caption{An example of an initial condition that has to be above the line $\omega^{1/2}$ - See   \cite{EscobedoVelazquez:2015:FTB,EscobedoVelazquez:2015:OTT}).}
		\label{Fig2}
	\end{figure}
\end{remark}

\begin{remark}\label{Remark3}{\it
	Inequality \eqref{Theorem2:5} implies that, if we consider a time $t$ such that there is no condensation at the origin, namely  $t\in \digamma$, then the set of times where the solution is spread out on  any   subinterval $[im\Re/\mathscr N,(i+3)m\Re/\mathscr N)$ of $[0,m\Re]$ always has a finite measure. Then one of the following possibilities has to be true:
	\begin{itemize}
	\item[(a)] Either the set of times $t\in \digamma$ such that $F(t)$ is concentrated in one of the  intervals  $[im\Re/\mathscr N,(i+3)m\Re/\mathscr N)$, $i>0$,
has infinite measure;
		\item[(b)] Or the set of condensation times ($F(t)$ is concentrated around the origin)
	$$
	[0,\infty)\backslash	\digamma \  = \ {\Big\{}t\in[0,\infty)~~~ \large|~~~ \int_{\{0\}}\mathrm{d}k f(t,k) =\int_{\{0\}}\mathrm{d}\omega F(t,\omega) \ne 0 {\Big\}},
	$$
	and  	has infinite measure. 
\end{itemize}   }
\end{remark}
\begin{remark}
{\it Our analysis can be extended to the case
	\begin{equation}\label{4wavesingular}
		\begin{aligned}
			\partial_t f\ = \ & 	\mathfrak C \left[ f\right],\ \ \ f(0,k)=f_0(k),\ \ \ k\in\mathbb{R}^3,\ \ \ t\in [0,\infty),\\ 
			\mathfrak C \left[ f\right]  
			\ =\ & \iiint_{\mathbb{R}^{9}}\mathrm{d}k_1\,\mathrm{d}k_2\,\mathrm{d}k_3  \delta(k+k_1-k_2-k_3)\delta(\omega + \omega_1 -\omega_2 - \omega_3)\\
			&\times |\omega\omega_1\omega_2\omega_3|^{c_o}[f_2f_3(f_1+f)-ff_1(f_2+f_3)],
		\end{aligned}
	\end{equation}
where the kernel is   $|\omega\omega_1\omega_2\omega_3|^{c_o}$ with $c_o\in\mathbb{R}$ under some constraints on $|c_o|$. However, we only focus on \eqref{4wave} for the sake of simplicity.}
	\end{remark}
\section{Non-condensation times}
\subsection{A Domain Decomposition Method}\label{Sec:DDM}
In this subsection, we propose an approach based on a Domain Decomposition Method (see \cite{halpern2009nonlinear,Lions:1989:OSA,toselli2004domain}). The method is based on a numerical ``divide and conquer  idea'', in which    complex   operators, such as our collision operators, are decomposed into    independent operators/integrals on smaller subdomains. 
For  a fixed small domain size $0<h<< R$,
we define  sequences of nonoverlapping  subdomains $\left\{ \Delta^{h, R}_{i}\right\}
_{i=0}^{\mathscr N_{h, R}}$  and overlapping subdomains $\left\{ \Xi^{h, R}_{i}\right\}
_{i=0}^{\mathscr N_{h,R}}$ lying inside the domain $[0,R)$ satisfying

%
%

\begin{equation}\label{Sec:FV:2}
	\begin{aligned}
		&\mbox{(A) Number of Subdomains:  } 	\mathscr N_{h,R}  = \left\lfloor \frac{R}{h}\right\rfloor, \\ \ 
		& \mbox{(B) Nonoverlapping Subdomains:  } 	\\
		&\Delta_{i}^{h,R} =\left[ ih,(1+i)h\right) \ \ ,\ \
		i=0, \cdots, \mathscr N_{h,R}-2, \ \ 	\Delta_{\mathscr N_{h,R}-1}^{h,R} =\left[ (\mathscr N_{h,R}-1)h,R\right),  \\
		& \mbox{(C) Overlapping Subdomains:  } 	\Xi_{i}^{h,R} =\left[ (i-1)h,(2+i)h\right) \ \ ,\ \
		i=1, \cdots, \mathscr N_{h,R}-3, \\
		&	\Xi_{0}^{h,R} =\left[0, 2h\right),  \ \ \Xi_{\mathscr N_{h,R}-2}^{h,R} =\left[ (\mathscr N_{h,R}-3)h,R\right),  \ \ \ \	\Xi_{\mathscr N_{h,R}-1}^{h,R} =\left[ (\mathscr N_{h,R}-2)h,R\right).	\end{aligned}
\end{equation}

We  set
\begin{equation}\label{Sec:FV:3}\begin{aligned}
	&	\Delta_{i,j,l}^{h,R} =\Delta_{i}^{h,R} \times%
\Delta_{j}^{h,R} \times\Delta_{l}^{h,R} ,\ \ i,j,l=0,\cdots,\mathscr N_{h,R}-1,\end{aligned}
\end{equation}
and

\begin{equation}\label{Sec:FV:3:1}\begin{aligned}
	&I_{i}^{h,R}  =\{ i-1,i,i+1\} \ \ ,\ \
		i=1, \cdots, \mathscr N_{h,R}-2, \\
		&	I_{0}^{h,R} =\{0,1\},  \ \ I_{\mathscr N_{h,R}-1}^{h,R} =\{\mathscr N_{h,R}-2,\mathscr N_{h,R}-1\}.	\end{aligned}
\end{equation}

and define the sets
\begin{equation}\label{Sec:FV:4}
	{E}_{R,h}=\left\{ \left( \omega,\omega_{1},\omega_{2}\right) \in\left[ 0,R\right) ^{3}:\left\vert \omega_{Mid}-	\omega_{Min}\right\vert \ge 2h\right\}, 
\end{equation}
and
\begin{equation}\label{Sec:FV:4}
	{E}_{R,h}'=\left\{ \left( \omega,\omega_{1},\omega_{2}\right) \in\left[ 0,R\right) ^{3}:\left\vert \omega_{Mid}-	\omega_{Min}\right\vert \ge h\right\}, 
\end{equation}
in which
	\begin{align}\label{Sec:FV:5}
	\omega_{Max}\left( \omega,\omega_{1},\omega_{2}\right) & =\max\left\{\omega,\omega_{1},\omega_{2}\right\} , \ 
	\omega_{Min}\left( \omega,\omega_{1},\omega_{2}\right)   =\min\left\{ \omega,\omega_{1},\omega_{2}\right\} , \\
	\omega_{Mid}\left( \omega,\omega_{1},\omega_{2}\right) & =\mathrm{mid}\left\{ \omega,\omega_{1},\omega_{2}\right\},
\end{align}
where we have used the notations of \eqref{Mid}. For a constant $\nu>0$, we denote 
\begin{equation}\label{TimePump}
	\begin{aligned}
		\digamma^2\ := & \ \Big\{t\in\digamma~~~ \big|~~~ \forall i\in\{0,\cdots,\mathscr N_{h,R}-1\},\int_{\Xi_{i}^{h,R} }\mathrm{d}\omega F(t) < (1-\nu)\int_{[0,R)}\mathrm{d}\omega F(t)\Big\},\\ \ \ \digamma^1 \ := \ & \digamma\backslash \digamma^2.
	\end{aligned}
\end{equation}
%

\subsection{A uniform in time estimate of the nonlinearity}

\begin{lemma}
	\label{lemma:Concave} 
	Let $f$ be a radial solution in the sense of \eqref{4wavemild} of the wave kinetic equation \eqref{4wave}, then the following estimate holds true for all $T>0$
		\begin{equation}\label{lemma:Concave:1}
		\begin{aligned}
		&	\int_{\mathbb{R}^3}\mathrm{d}k  f(0,k)  \ln(\omega+1)
			\ \ge\  C\int_0^T\mathrm{d}t\iiint_{\mathbb{R}_+^{3}}\mathrm{d}\omega_1\,\mathrm{d}\omega_2\,\mathrm{d}\omega f_1f_2f\mathbf{1}_{\omega+\omega_1-\omega_2\ge 0}\\
			&\times \mho(\omega_{Max})\mho(\omega_{Min})\mho(\omega_{Mid})\mho(\omega_{Max}-\omega_{Min}+\omega_{Mid})	|k_{Min}| \frac{(\omega_{Mid}-\omega_{Min})^2}{( 2\omega_{Mid}-\omega_{Min})^2+1},
		\end{aligned}
	\end{equation}
where $k_{Min}$ is associated to $\omega_{Min}$, $C$ is a universal constant,	and  we have used  the shorthand notation $f=f(k)=f(|k|)=f(\omega)$, $f_1=f(k_1)=f(|k_1|)=f(\omega_1)$, $f_2=f(k_2)=f(|k_2|)=f(\omega_2)$, $\omega=\omega(k)=\omega(|k|)$, $\omega_1=\omega(k_1)=\omega(|k_1|)$, $\omega_2=\omega(k_2)=\omega(|k_2|)$,  and we also have used the same notations as in  \eqref{Sec:FV:5}.
\end{lemma}

\begin{proof}

	We first write  the collision operator in  weak form. This computation has been done in  \cite{StaffilaniTranCascade1} but we repeat it here for the sake of completeness. We write
	
	\begin{equation}\label{Lemma:Concave:E1}
		\begin{aligned}
			\mathfrak C \left[ f\right]  
			\ =\ & \iiint_{\mathbb{R}^{3\times3}}\mathrm{d}k_1\,\mathrm{d}k_2\,\mathrm{d}k_3  \delta(\omega + \omega_1 -\omega_2 - \omega_3)\delta(k+k_1-k_2-k_3)[-ff_1(f_2+f_3)+f_2f_3(f_1+f)] \\
			\ =\ & \iiint_{\mathbb{R}_+^{3}}\mathrm{d}|k_1|\,\mathrm{d}|k_2|\,\mathrm{d}|k_3|  |k_1|^2|k_2|^2|k_3|^2\delta(\omega + \omega_1 -\omega_2 - \omega_3)\\
			&\times \iiint_{\left(\mathbb{S}^2\right)^3}\mathrm{d}\mathcal V_1\mathrm{d}\mathcal V_2\mathrm{d}\mathcal V_3[-ff_1(f_2+f_3)+f_2f_3(f_1+f)]\left[\frac1{(2\pi)^3}\int_{\mathbb{R}^3}\mathrm{d}s e^{is\cdot({k}+{k}_1-{k}_2-{k}_3)}\right]\\
			\ =\ & \iiint_{\mathbb{R}_+^{3}}\mathrm{d}|k_1|\,\mathrm{d}|k_2|\,\mathrm{d}|k_3|\frac{32\pi}{|k|} \int_0^\infty\mathrm{d}s\sin(|k_1|s)\sin(|k_2|s)\sin(|k_3|s)\sin(|k|s)\frac{1}{s^2}\\
			&\times \delta(\omega + \omega_1 -\omega_2 - \omega_3)[-ff_1(f_2+f_3)+f_2f_3(f_1+f)]|k_1||k_2||k_3|.
		\end{aligned}
	\end{equation}
	Next, we compute 
	\begin{equation}\label{Lemma:Concave:E2}
		\begin{aligned}
			& \int_0^\infty\mathrm{d}s\sin(|k_1|s)\sin(|k_2|s)\sin(|k_3|s)\sin(|k|s)\frac{1}{s^2}\\
			=\	& \frac1{16}\int_0^\infty\mathrm{d}s (e^{i|k_1| s}-e^{-i|k_1| s})(e^{i|k_2| s}-e^{-i|k_2| s})(e^{i|k_3| s}-e^{-i|k_3| s})(e^{i|k_4| s}-e^{-i|k_4| s})\frac{1}{s^2}\\
			=\	& \frac1{16}\int_0^\infty\mathrm{d}s\frac{1}{s^2}\sum_{x_1,x_2,x_3,x_4=1}^2 (-1)^{x_1}e^{i(-1)^{x_1}|k_1|s}(-1)^{x_2}e^{i(-1)^{x_2}|k_2|s}(-1)^{x_3}e^{i(-1)^{x_3}|k_3|s}(-1)^{x_4}e^{i(-1)^{x_4}|k_4|s}\\
			=\	& \int_0^\infty\mathrm{d}s\frac{1}{s^2}\sum_{x_1,x_2,x_3,x_4=1}^2\frac1{16}(-1)^{x_1+x_2+x_3+x_4}\\
			&\times \exp\Big\{i\Big(\textstyle (-1)^{x_1}|k_1|+(-1)^{x_2}|k_2|+(-1)^{x_3}|k_3|+(-1)^{x_4}|k_4|\Big)s\Big\}\\
			=\	&		\frac\pi{16}\sum_{x_1,x_2,x_3,x_4=1}^2(-1)^{x_1+x_2+x_3+x_4+1}\Big|(-1)^{x_1}|k_1|+(-1)^{x_2}|k_2|+(-1)^{x_3}|k_3|+(-1)^{x_4}|k_4|\Big|\\
					=\ &  \tfrac\pi4\min\{|k_1|,|k_2|,|k_3|,|k|\}.
		\end{aligned}
	\end{equation}
	Inserting \eqref{Lemma:Concave:E2} into \eqref{Lemma:Concave:E1}, we find
	\begin{equation}\label{Lemma:Concave:E3}
		\begin{aligned}
			\mathfrak C  \left[ f\right]  
			\ =\ & 8\pi^2\iiint_{\mathbb{R}_+^{3}}\mathrm{d}|k_1|\,\mathrm{d}|k_2|\,\mathrm{d}|k_3|\frac{|k_1||k_2||k_3|\min\{|k_1|,|k_2|,|k_3|,|k|\}}{|k|} \delta(\omega + \omega_1 -\omega_2 - \omega_3)\\
			&\times [-ff_1(f_2+f_3)+f_2f_3(f_1+f)].
		\end{aligned}
	\end{equation}
 After the change of variables $|k|\to\omega$, $|k_1|\to\omega_1$, $|k_2|\to\omega_2$, $|k_3|\to\omega_3$, we have the following the weak form of \eqref{Lemma:Concave:E3}  
	\begin{equation}\label{Lemma:Concave:E4}
		\begin{aligned}
			\int_{\mathbb{R}_+}\mathrm{d}\omega 	\mathfrak C  \left[ f\right]  \rho(\omega)|k| \mho
			\ =\ &  C_0\iiiint_{\mathbb{R}_+^{4}}\mathrm{d}\omega_1\,\mathrm{d}\omega_2\,\mathrm{d}\omega_3\mathrm{d}\omega \delta(\omega + \omega_1 -\omega_2 - \omega_3)\Lambda(\omega,\omega_1,\omega_2,\omega_3)\\
			&\times [-ff_1(f_2+f_3)+f_2f_3(f_1+f)]\rho,
		\end{aligned}
	\end{equation}
	for some universal constant $ C_0>0$ and a test function $\rho(\omega)$ to be fixed later and 
	$$\Lambda(\omega,\omega_1,\omega_2,\omega_3)={\mho\mho_1\mho_2\mho_3\min\{|k_1|,|k_2|,|k_3|,|k|\}}.$$
 From \eqref{Lemma:Concave:E4}, we continue with  
	\begin{equation}\label{Lemma:Concave:E5}
		\begin{aligned}
			\int_{\mathbb{R}_+}\mathrm{d}\omega \mathfrak C \left[ f\right]  \rho(\omega)|k| \mho
			\ =\ & C_0\iiiint_{\mathbb{R}_+^{4}}\mathrm{d}\omega_1\,\mathrm{d}\omega_2\,\mathrm{d}\omega_3\mathrm{d}\omega \Lambda(\omega,\omega_1,\omega_2,\omega_3)\\
			&\times \delta(\omega + \omega_1 -\omega_2 - \omega_3)f_1f_2f[-\rho(\omega)-\rho(\omega_1)+\rho(\omega_2)+\rho(\omega_3)]\\
			\ =\ & C_0\iiint_{\mathbb{R}_+^{3}}\mathrm{d}\omega_1\,\mathrm{d}\omega_2\,\mathrm{d}\omega \Lambda(\omega,\omega_1,\omega_2,\omega+\omega_1-\omega_2)\\
			&\times f_1f_2f[-\rho(\omega)-\rho(\omega_1)+\rho(\omega_2)+\rho(\omega+\omega_1-\omega_2)]\\
				\ =\ & C_1\iiint_{\mathbb{R}_+^{3}}\mathrm{d}\omega_1\,\mathrm{d}\omega_2\,\mathrm{d}\omega f_1f_2f\mathbf{1}_{\omega+\omega_1-\omega_2\ge 0}\\
			&\times 	 \Big\{ [-\rho(\omega_{Max})-\rho(\omega_{Min})+\rho(\omega_{Mid})+\rho(\omega_{Max}+\omega_{Min}-\omega_{Mid})]\\
			&\ \ \ \ \times\Lambda(\omega_{Max},\omega_{Min},\omega_{Mid},\omega_{Max}+\omega_{Min}-\omega_{Mid})  \\
			&+[-\rho(\omega_{Max})-\rho(\omega_{Mid})+\rho(\omega_{Min})+\rho(\omega_{Max}+\omega_{Mid}-\omega_{Min})]\\
			&\ \ \ \ \times\Lambda(\omega_{Max},\omega_{Mid},\omega_{Min},\omega_{Max}+\omega_{Mid}-\omega_{Min})  \\
			&+[-\rho(\omega_{Min})-\rho(\omega_{Mid})+\rho(\omega_{Max})+\rho(\omega_{Min}+\omega_{Mid}-\omega_{Max})]\\
			&\ \ \ \ \times\Lambda(\omega_{Min},\omega_{Mid},\omega_{Max},\omega_{Min}+\omega_{Mid}-\omega_{Max})\Big\}\\
			\ =\ & C_1\iiint_{\mathbb{R}_+^{3}}\mathrm{d}\omega_1\,\mathrm{d}\omega_2\,\mathrm{d}\omega f_1f_2f\mathbf{1}_{\omega+\omega_1-\omega_2\ge 0}\\
			&\times 	 \Big\{ [-\rho(\omega_{Max})-\rho(\omega_{Min})+\rho(\omega_{Mid})+\rho(\omega_{Max}+\omega_{Min}-\omega_{Mid})]\\
			&\ \ \ \ \times\mho(\omega_{Max})\mho(\omega_{Min})\mho(\omega_{Mid})\mho(\omega_{Max}+\omega_{Min}-\omega_{Mid})|k_{Min}|  \\
			&+[-\rho(\omega_{Max})-\rho(\omega_{Mid})+\rho(\omega_{Min})+\rho(\omega_{Max}+\omega_{Mid}-\omega_{Min})]\\
			&\ \ \ \ \times\mho(\omega_{Max})\mho(\omega_{Min})\mho(\omega_{Mid})\mho(\omega_{Max}+\omega_{Mid}-\omega_{Min})|k_{Min}|  \\
			&+[-\rho(\omega_{Min})-\rho(\omega_{Mid})+\rho(\omega_{Max})+\rho(\omega_{Min}+\omega_{Mid}-\omega_{Max})]\\
			&\ \ \ \ \times\Lambda(\omega_{Min},\omega_{Mid},\omega_{Max},\omega_{Min}+\omega_{Mid}-\omega_{Max})\Big\},
		\end{aligned}
	\end{equation}
 	for some universal constant $ C_1>0$, where we have used the same notations as in  \eqref{Sec:FV:5}. We assume that   $|k_{Min}|$ is associated to $\omega_{Min}$, $|k_{Max}|$ is associated to $\omega_{Max}$, $|k_{Mid}|$ is associated to $\omega_{Mid}$, and when $\omega_{Min}+\omega_{Mid}-\omega_{Max}\le 0$, we assume that $\Lambda(\omega_{Min},\omega_{Mid},\omega_{Max},\omega_{Min}+\omega_{Mid}-\omega_{Max})=0.$

	It is straightforward that 
	\begin{equation}\label{Lemma:Concave:E5a}
		\begin{aligned}
			&[-\rho(\omega_{Min})-\rho(\omega_{Mid})+\rho(\omega_{Max})+\rho(\omega_{Min}+\omega_{Mid}-\omega_{Max})]\\
			&\ \ \ \ \times\Lambda(\omega_{Min},\omega_{Mid},\omega_{Max},\omega_{Min}+\omega_{Mid}-\omega_{Max})\Big\}\\	=\ &\int_{0}^{\omega_{Max}-\omega_{Min}}\mathrm{d}\xi_1\int_{0}^{\omega_{Max}-\omega_{Mid}}\mathrm{d}\xi_2\rho''(\xi_1+\xi_2+\omega_{Min})\\
			&\ \ \ \ \times\Lambda(\omega_{Min},\omega_{Mid},\omega_{Max},\omega_{Min}+\omega_{Mid}-\omega_{Max})\ \le 0,
		\end{aligned}
	\end{equation}
	under the further restriction   $\rho''\le 0$. Note that the form of $\rho$ will be fixed later.

	Next, we compute
	\begin{equation}\label{Lemma:Concave:E6}
		\begin{aligned}
			& 	  [-\rho(\omega_{Max})-\rho(\omega_{Min})+\rho(\omega_{Mid})+\rho(\omega_{Max}+\omega_{Min}-\omega_{Mid})]\mho(\omega_{Max}+\omega_{Min}-\omega_{Mid})  \\
			&+[-\rho(\omega_{Max})-\rho(\omega_{Mid})+\rho(\omega_{Min})+\rho(\omega_{Max}+\omega_{Mid}-\omega_{Min})]\mho(\omega_{Max}+\omega_{Mid}-\omega_{Min})\\
	=\		& 	  -\int_{0}^{\omega_{Mid}-\omega_{Min}}\mathrm{d}s\int_{0}^{\omega_{Max}-\omega_{Mid}}\mathrm{d}s_0\mho(\omega_{Max}+\omega_{Min}-\omega_{Mid}) \rho''(\omega_{Min}+s+s_0) \\
		&+\int_{0}^{\omega_{Mid}-\omega_{Min}}\mathrm{d}s\int_{0}^{\omega_{Max}-\omega_{Min}}\mathrm{d}s_0\mho(\omega_{Max}-\omega_{Min}+\omega_{Mid}) \rho''(\omega_{Min}+s+s_0)\\
		=\		& 	  \int_{0}^{\omega_{Mid}-\omega_{Min}}\mathrm{d}s\int_{0}^{\omega_{Max}-\omega_{Mid}}\mathrm{d}s_0 \rho''(\omega_{Min}+s+s_0) \\
		&\times[\mho(\omega_{Max}-\omega_{Min}+\omega_{Mid})-\mho(\omega_{Max}+\omega_{Min}-\omega_{Mid})]\\
	&+\int_{0}^{\omega_{Mid}-\omega_{Min}}\mathrm{d}s\int_{0}^{\omega_{Mid}-\omega_{Min}}\mathrm{d}s_0\mho(\omega_{Max}-\omega_{Min}+\omega_{Mid}) \rho''(\omega_{Min}+s+s_0)\\
	\le\ &	\int_{0}^{\omega_{Mid}-\omega_{Min}}\mathrm{d}s\int_{0}^{\omega_{Mid}-\omega_{Min}}\mathrm{d}s_0\mho(\omega_{Max}-\omega_{Min}+\omega_{Mid}) \rho''(\omega_{Min}+s+s_0),
		\end{aligned}
	\end{equation}
under the further restriction   $\rho''\le 0$ and note that $\mho$ is a non-decreasing function.

By choosing $\rho(\omega)=\ln(\omega+1)$ for $\omega\in\mathbb{R}_+$, we have $\rho''(\omega)=-\frac{1}{(\omega+1)^2}$. We obtain

		\begin{equation}\label{Lemma:Concave:E7}
		\begin{aligned}
			& 	  [-\rho(\omega_{Max})-\rho(\omega_{Min})+\rho(\omega_{Mid})+\rho(\omega_{Max}+\omega_{Min}-\omega_{Mid})]\mho(\omega_{Max}+\omega_{Min}-\omega_{Mid})  \\
			&+[-\rho(\omega_{Max})-\rho(\omega_{Mid})+\rho(\omega_{Min})+\rho(\omega_{Max}+\omega_{Mid}-\omega_{Min})]\mho(\omega_{Max}+\omega_{Mid}-\omega_{Min})\\
			\le\ &	-\int_{0}^{\omega_{Mid}-\omega_{Min}}\mathrm{d}s\int_{0}^{\omega_{Mid}-\omega_{Min}}\mathrm{d}s_0\mho(\omega_{Max}-\omega_{Min}+\omega_{Mid}) \frac{1}{(\omega_{Min}+s+s_0+1)^2}\\
			\le\ &	-\int_{0}^{\omega_{Mid}-\omega_{Min}}\mathrm{d}s\int_{0}^{\omega_{Mid}-\omega_{Min}}\mathrm{d}s_0\mho(\omega_{Max}-\omega_{Min}+\omega_{Mid}) \frac{1}{(\omega_{Min}+2(\omega_{Mid}-\omega_{Min})+1)^2}\\
			\le\ &	-\mho(\omega_{Max}-\omega_{Min}+\omega_{Mid}) \frac{(\omega_{Mid}-\omega_{Min})^2}{(\omega_{Min}+2(\omega_{Mid}-\omega_{Min})+1)^2}\\\ \le\  &	-\mho(\omega_{Max}-\omega_{Min}+\omega_{Mid}) \frac{(\omega_{Mid}-\omega_{Min})^2}{( 2\omega_{Mid}-\omega_{Min}+1)^2},
		\end{aligned}
	\end{equation}
which, together with \eqref{Lemma:Concave:E5}-\eqref{Lemma:Concave:E5a}, implies 
	\begin{equation}\label{Lemma:Concave:E8}
	\begin{aligned}
		&\int_{\mathbb{R}_+}\mathrm{d}\omega \mathfrak C \left[ f\right]  \ln(\omega+1)|k| \mho
		\ \le\  -C_1'\iiint_{\mathbb{R}_+^{3}}\mathrm{d}\omega_1\,\mathrm{d}\omega_2\,\mathrm{d}\omega f_1f_2f\mathbf{1}_{\omega+\omega_1-\omega_2\ge 0}|k_{Min}| \\
	&\times 	
		\frac{(\omega_{Mid}-\omega_{Min})^2}{( 2\omega_{Mid}-\omega_{Min})^2+1}
		\mho(\omega_{Max})\mho(\omega_{Min})\mho(\omega_{Mid})\mho(\omega_{Max}-\omega_{Min}+\omega_{Mid}),
	\end{aligned}
\end{equation}
for some constant $C_1'$, yielding
\begin{equation}\label{Lemma:Concave:E9}
	\begin{aligned}
	&\partial_t	\int_{\mathbb{R}_+}\mathrm{d}\omega \  f  \ln(\omega+1)|k| \mho
		\ \le\  -C_1'\iiint_{\mathbb{R}_+^{3}}\mathrm{d}\omega_1\,\mathrm{d}\omega_2\,\mathrm{d}\omega f_1f_2f\mathbf{1}_{\omega+\omega_1-\omega_2\ge 0}|k_{Min}| \\
			&\times \frac{(\omega_{Mid}-\omega_{Min})^2}{( 2\omega_{Mid}-\omega_{Min})^2+1}
		\mho(\omega_{Max})\mho(\omega_{Min})\mho(\omega_{Mid})\mho(\omega_{Max}-\omega_{Min}+\omega_{Mid}).
	\end{aligned}
\end{equation}
Integrating both sides of \eqref{Lemma:Concave:E9} in $t$, we find
\begin{equation}\label{Lemma:Concave:E10}
	\begin{aligned}
	& 	\int_{\mathbb{R}}\mathrm{d}k \  f(0,\omega) \ln(\omega+1) - 	\int_{\mathbb{R}}\mathrm{d}k \  f(T,\omega) \ln(\omega+1) \\
		\ \ge\ & C_1'\int_0^T\mathrm d t \iiint_{\mathbb{R}_+^{3}}\mathrm{d}\omega_1\,\mathrm{d}\omega_2\,\mathrm{d}\omega f_1f_2f\mathbf{1}_{\omega+\omega_1-\omega_2\ge 0}|k_{Min}| \frac{(\omega_{Mid}-\omega_{Min})^2}{( 2\omega_{Mid}-\omega_{Min})^2+1}\\
		&\times 	\mho(\omega_{Max})\mho(\omega_{Min})\mho(\omega_{Mid})\mho(\omega_{Max}-\omega_{Min}+\omega_{Mid}).
	\end{aligned}
\end{equation}

The proposition is proved. 	
\end{proof}

\subsection{Estimating non-condensation times}
\begin{proposition}
	\label{Propo:Collision} There exists a universal constant ${C}_0>0$ such that the following estimate holds true, for $0<\nu,h<1/10$,	\begin{equation}\label{Propo:Collision:1}
		\begin{aligned}
			\frac{{C}_0R^{2+\frac2\alpha}(	\mathscr M +\mathscr E)}{h^{2}\mho(h) \nu^4}
		\ \ge\ &\int_{\digamma^2}\mathrm{d}t  \left( \int_{[0,R)}\mathrm{d}\omega F(t,\omega)\right)^3.
		\end{aligned}
	\end{equation}
\end{proposition}
\begin{proof}
	We deduce from \eqref{lemma:Concave:1} that
	
		\begin{equation}\label{Propo:Collision:1}
		\begin{aligned}
	\mathscr M +\mathscr E\ =	\ &	\int_{\mathbb{R}^3}\mathrm{d}k  f(0,k)  (\omega+1)
			\ \ge\   C\int_0^T\mathrm{d}t\iiint_{\mathbb{R}_+^{3}}\mathrm{d}\omega_1\,\mathrm{d}\omega_2\,\mathrm{d}\omega f_1f_2f\mathbf{1}_{\omega+\omega_1-\omega_2\ge 0}\\
			&\times 	\mho(\omega_{Max})\mho(\omega_{Min})\mho(\omega_{Mid})\mho(\omega_{Max}-\omega_{Min}+\omega_{Mid})|k_{Min}| \frac{(\omega_{Mid}-\omega_{Min})^2}{( 2\omega_{Mid}-\omega_{Min})^2+1},
		\end{aligned}
	\end{equation}
for some constant $C>0$, yielding
	\begin{equation}\label{Propo:Collision:2}
	\begin{aligned}
	&	\mathscr M +\mathscr E\ =	\ 	\int_{\mathbb{R}^3}\mathrm{d}k  f(0,k)  (\omega+1)
		\ \ge\   C\int_0^T\mathrm{d}t\iiint_{[0,R)^{3}}\mathrm{d}\omega_1\,\mathrm{d}\omega_2\,\mathrm{d}\omega f_1f_2f\mathbf{1}_{\omega+\omega_1-\omega_2\ge 0}\\
		&\times 	\mho(\omega_{Max})\mho(\omega_{Min})\mho(\omega_{Mid})\mho(\omega_{Max}-\omega_{Min}+\omega_{Mid})|k_{Min}| \frac{(\omega_{Mid}-\omega_{Min})^2}{( 2\omega_{Mid}-\omega_{Min})^2+1}\\
			\ \gtrsim\   & \int_0^T\mathrm{d}t\iiint_{[0,R)^{3}}\mathrm{d}\omega_1\,\mathrm{d}\omega_2\,\mathrm{d}\omega F_1F_2F\mathbf{1}_{\omega+\omega_1-\omega_2\ge 0}  \frac{\mho(\omega_{Max}-\omega_{Min}+\omega_{Mid})}{|k_{Mid}||k_{Max}|} \frac{(\omega_{Mid}-\omega_{Min})^2}{( 2\omega_{Mid}-\omega_{Min})^2+1}\\
			\ \gtrsim\   & \int_0^T\mathrm{d}t\iiint_{[0,R)^{3}}\mathrm{d}\omega_1\,\mathrm{d}\omega_2\,\mathrm{d}\omega F_1F_2F\mathbf{1}_{\omega+\omega_1-\omega_2\ge 0}  \frac{\mho(\omega_{Max}-\omega_{Min}+\omega_{Mid})}{|k_{Mid}||k_{Max}|} \frac{(\omega_{Mid}-\omega_{Min})^2}{( 2\omega_{Mid}-\omega_{Min})^2+1}
			\\
				\ \gtrsim\   & \int_0^T\mathrm{d}t\iiint_{[0,R)^{3}}\mathrm{d}\omega_1\,\mathrm{d}\omega_2\,\mathrm{d}\omega F_1F_2F\mathbf{1}_{\omega+\omega_1-\omega_2\ge 0}  \frac{\mho(\omega_{Max}-\omega_{Min}+\omega_{Mid})}{|\omega_{Mid}|^\frac1\alpha|\omega_{Max}|^\frac1\alpha} \frac{(\omega_{Mid}-\omega_{Min})^2}{( 2\omega_{Mid}-\omega_{Min})^2+1},
	\end{aligned}
\end{equation}
 where we denoted $F=F(\omega)$,  $F_1=F(\omega_1)$, $F_2=F(\omega_2)$.  
 We divide the rest of the proof into smaller steps.

{\it Step 1: Subdomain estimates.}

We will show that
	
		\begin{equation}\label{Propo:Collision:3}
		E_{R,h}\subset \left[
		\bigcup_{\mathscr N_{R,h} > \mathrm{max}\{i,j,l\}\geq \mathrm{mid}\{i,j,l\}>\mathrm{min}\{i,j,l\}+1}\Delta_{i,j,l}^{h,R} \right] \subset   	E_{R,h}',
	\end{equation}
	which can be combined with \eqref{Propo:Collision:2} to get
	\begin{equation}\label{Propo:Collision:4}
		\begin{aligned}
			\mathscr M +\mathscr E 
			\ \gtrsim\ & \sum_{\mathscr N_{R,h} > \mathrm{max}\{i,j,l\}\geq \mathrm{mid}\{i,j,l\}>\mathrm{min}\{i,j,l\}+1}\int_0^T\mathrm{d}t\iiint_{\Delta_{i,j,l}^{h,R} }\mathrm{d}\omega_1\,\mathrm{d}\omega_2\,\mathrm{d}\omega F_1F_2F\mathbf{1}_{\omega+\omega_1-\omega_2\ge 0}\\
			&\times 	 \frac{\mho(\omega_{Max}-\omega_{Min}+\omega_{Mid})}{|\omega_{Mid}|^\frac1\alpha|\omega_{Max}|^\frac1\alpha} \frac{(\omega_{Mid}-\omega_{Min})^2}{( 2\omega_{Mid}-\omega_{Min})^2+1}.
		\end{aligned}
	\end{equation}

	To prove \eqref{Propo:Collision:3}, let us consider $\left( \omega_2,\omega_{1},\omega\right) \in%
	{E}_{R,h}.$ Without loss of generality, we assume    $%
	\omega=\omega_{Min}<\omega_{1}=\omega_{Mid}\leq\omega_{2}=%
	\omega_{Max}.$ From our construction, there exist $%
	i,\ j$ such that $\omega_{2}\in\Delta_{i}^{h,R},\ \omega_{1}\in%
	\Delta_{j}^{h,R}$ with $j\leq i.$ By the definition of ${E}_{R,h}$, we have $\omega\le \omega_1-2h< (j-1)h,$ yielding $\omega\in\bigcup_{l=0}^{j-2}\Delta_{l}^{h,R}$. Therefore:%
	\begin{equation*}
		\left(\omega_2, \omega_{1},\omega\right) \in\bigcup_{\mathscr N_{R,h} > i\geq j>l+1}\Delta_{i,j,l}^{h,R}, 
	\end{equation*}
	yielding
	
	\begin{equation}
		E_{R,h}\subset \left[
		\bigcup_{\mathscr N_{R,h} > \mathrm{max}\{i,j,l\}\geq \mathrm{mid}\{i,j,l\}>\mathrm{min}\{i,j,l\}+1}\Delta_{i,j,l}^{h,R} \right],  
	\end{equation}
	where we have used the notations of \eqref{Mid}.  Next, we choose $$\left(\omega_2, \omega_{1},\omega\right) \in
	\bigcup_{\mathscr N_{R,h} > \mathrm{max}\{i,j,l\}\geq \mathrm{mid}\{i,j,l\}>\mathrm{min}\{i,j,l\}+1}\Delta_{i,j,l}^{h,R}.$$
	Without loss of generality, we suppose
	$\omega_{2}\in\Delta_{i}^{h,R},\ \omega_{1}\in%
	\Delta_{j}^{h,R}$, $\omega\in \Delta_{l}^{h,R}$ with $\mathscr N_{R,h} > i\geq j>l+1$, yielding
	$\left\vert \omega_{Mid}-	\omega_{Min}\right\vert \ge h$. We then have
	\begin{equation}
		\left[
		\bigcup_{\mathscr N_{R,h} > \mathrm{max}\{i,j,l\}\geq \mathrm{mid}\{i,j,l\}>\mathrm{min}\{i,j,l\}+1}\Delta_{i,j,l}^{h,R} \right] \subset E_{R,h}',  
	\end{equation}
that completes the proof of \eqref{Propo:Collision:3}.

	Coming back to the estimate \eqref{Propo:Collision:4}, using the second inclusion in \eqref{Propo:Collision:3}, we obtain
		\begin{equation}\label{Propo:Collision:5}
		\begin{aligned}
			\mathscr M +\mathscr E
			\ \gtrsim\ &  \frac{h^2}{R^{2+\frac2\alpha}}\sum_{\mathscr N_{R,h} > i\geq j>l+1}\int_0^T\mathrm{d}t\iiint_{\Delta_{i,j,l}^{h,R} }\mathrm{d}\omega_1\,\mathrm{d}\omega_2\,\mathrm{d}\omega F_1F_2F\mathbf{1}_{\omega+\omega_1-\omega_2\ge 0}\\
			&\times\mho(\omega_{Max}-\omega_{Min}+\omega_{Mid})\\
			\ \gtrsim\ &  \frac{h^2}{R^{2+\frac2\alpha}}\sum_{(i,j,l)\in\mathfrak{X}}\int_0^T\mathrm{d}t\iiint_{\Delta_{i,j,l}^{h,R} }\mathrm{d}\omega_1\,\mathrm{d}\omega_2\,\mathrm{d}\omega F_1F_2F\mathbf{1}_{\omega+\omega_1-\omega_2\ge 0}\\
			&\times\mho(\omega_{Max}-\omega_{Min}+\omega_{Mid}),		\end{aligned}
	\end{equation}
where $\mathfrak{X} : = \{(i,j,l) \in\mathbb{Z}^3 ~~ |~~ i,j,l\ge0,\  \mathscr N_{R,h} > i\geq j>l+1\}$.

To further investigate  \eqref{Propo:Collision:5}, we choose $t$ to be any time in $\digamma^2$ and follow a domain decomposition argument to construct bigger non-overlapping subdomains.

{\it Step 2: Constructing bigger non-overlapping subdomains.}

 We denote   $l_t^*$ to be an index in $\{0,\cdots,\mathscr N_{R,h} -1\}$ such that
	\begin{equation}\label{Propo:Collision:6}
	\begin{aligned}
\int_{\Delta_{l_t^*}^{h,R} }\mathrm{d}\omega F(t) \ = \ \max_{i\in\{0,\cdots,\mathscr N_{R,h} -1\}}\Big\{	\int_{\Delta_{i}^{h,R} }\mathrm{d}\omega F(t)\, \Big\}.
	\end{aligned}
\end{equation}
We denote $\mathfrak S_t$ to be the set of partitions of $\{0,\cdots,\mathscr N_{R,h} -1\}$ that satisfies the following criteria.
	
	\begin{itemize}
		\item[(a)] If $\mathfrak{I}_t\in \mathfrak{S}_t$, then $l_t^*\in \mathfrak{I}_t$.  
			\item[(b)] Suppose $\mathfrak{I}_t=\{l_1,\cdots,l_m\}$, then
			for all $l\in  \mathfrak{I}_t$
				\begin{equation}\label{Propo:Collision:6a}	 	 \Delta_{l}^{h,R}\cap \bigcup_{i\in\mathfrak{I}_t \backslash\{l\}}\Xi_{i}^{h,R} \ =\ \emptyset.
			\end{equation}
		\item[(b)] There exists a constant $20\ge c_t\ge 3$ such that for all $\mathfrak{I}_t\in \mathfrak{S}_t$
		\begin{equation}\label{Propo:Collision:7}	 	c_t\int_{\cup_{i\in\mathfrak{I}_t }\Delta_{i}^{h,R} }\mathrm{d}\omega F(t) \ \ge \ \int_{\cup_{i\in\mathfrak{I}_t }\Xi_{i}^{h,R} }\mathrm{d}\omega F(t).
		\end{equation}
			\item[(d)] For all $\mathfrak{I}_t\in \mathfrak{S}_t$ 
			\begin{equation}\label{Propo:Collision:8} \int_{\cup_{i\in\mathfrak{I}_t }\Xi_{i}^{h,R}}\mathrm{d}\omega F(t) \ < \ 	(1- \nu)\int_{[0,R)}\mathrm{d}\omega F(t).
			\end{equation}
		
	\end{itemize}
Since  \begin{equation}\label{Propo:Collision:9}	3\int_{\Delta_{l_t^*}^{h,R} }\mathrm{d}\omega F(t) \ \ge \ \int_{ \Xi_{l_t^*}^{h,R} }\mathrm{d}\omega F(t),
\end{equation}
the set of partitions $	\mathfrak S_t$ is non empty. We define  with $\mathfrak{Z}_t=\{l_1,\cdots,l_m\}\in\mathfrak S_t$, in which $l_1=l_t^*$, the partition such that 
\begin{equation}\label{Propo:Collision:10:1}
	\begin{aligned}
\int_{\Delta_{l_j}^{h,R} }\mathrm{d}\omega F(t) \ = \ \max_{i\in\{0,\cdots,\mathscr N_{R,h} -1\}\backslash\Big(I_{l_1}^{h,R} \cup\cdots\cup I_{l_{j-1}}^{h,R} \Big)}\Big\{	\int_{\Delta_{i}^{h,R} }\mathrm{d}\omega F(t)\, \Big\},
	\end{aligned}
\end{equation}
for $j=2,\cdots,m$ and we have used the notations of \eqref{Sec:FV:3:1}. In addition, for all $l\in\{0,\cdots,\mathscr N_{R,h} -1\}\backslash \Big(I_{l_1}^{h,R} \cup\cdots\cup I_{l_{m}}^{h,R} \Big)$, we have 
		\begin{equation}\label{Propo:Collision:10}	\int_{\cup_{i\in\mathfrak{Z}_t \cup\{l\}}\Xi_{i}^{h,R} }\mathrm{d}\omega F(t)  \ \ge \ 	(1- \nu)\int_{[0,R)}\mathrm{d}\omega F(t).
	\end{equation}
Denoting $\mathscr I_t= \{0,\cdots,\mathscr N_{R,h} -1\}\backslash \mathfrak{Z}_t $ and $\mathscr X^1= [0,R)\backslash\cup_{i\in\mathfrak{Z}_t }\Xi_{i}^{h,R}=\cup_{i\in\mathfrak{Y}_t }\Delta_{i}^{h,R}$, 	 $\mathscr X^2=\cup_{i\in\mathfrak{Z}_t }\Delta_{i}^{h,R}$, we have \begin{equation}\label{Propo:Collision:11}	\int_{\mathscr X^1 }\mathrm{d}\omega F(t) \ \ge \ 	 \nu\int_{[0,R)}\mathrm{d}\omega F(t).
\end{equation}
	 We denote   $l_t^{**}$ to be an index in $\mathscr I_t$ such that $\Delta_{l_t^{**}}^{h,R} \cap\bigcup_{i\in\mathfrak{Z}_t}\Xi_{i}^{h,R}=\emptyset$ and
	\begin{equation}\label{Propo:Collision:12}
		\begin{aligned}
			\int_{\Delta_{l_t^{**}}^{h,R} }\mathrm{d}\omega F(t) \ = \ \max_{i\in \mathscr I_t}\Big\{	\int_{\Delta_{i}^{h,R} }\mathrm{d}\omega F(t)\, \Big\}.
		\end{aligned}
	\end{equation}
Let $\mathfrak A$ and $\mathfrak B$ be the two sets that satisfies 
$$\mathfrak A\cup \mathfrak B = \Xi_{l_t^{**}}^{h,R}, \ \ \ \ \mathfrak A\cap \mathfrak B=\emptyset, \ \ \ \   \mathfrak A\subset \bigcup_{i\in\mathfrak{Z}_t}\Xi_{i}^{h,R}, \ \ \ \    \mathfrak B\cap \bigcup_{i\in\mathfrak{Z}_t}\Xi_{i}^{h,R}=\emptyset.$$
By the definition of the set $\mathfrak{Z}_t$,
	\begin{equation}\label{Propo:Collision:13}	 	c_t\int_{\cup_{i\in\mathfrak{Z}_t }\Delta_{i}^{h,R} }\mathrm{d}\omega F(t) \ \ge \ \int_{\cup_{i\in\mathfrak{Z}_t}\Xi_{i}^{h,R} }\mathrm{d}\omega F(t),
\end{equation}
	and by the definition of the point $l_t^{**}$
	\begin{equation}\label{Propo:Collision:14}	 	c_t\int_{ \Delta_{l_t^{**}}^{h,R} }\mathrm{d}\omega F(t) \ \ge \ 3\int_{ \Delta_{l_t^{**}}^{h,R} }\mathrm{d}\omega F(t) \ \ge \ \int_{\mathfrak B }\mathrm{d}\omega F(t).
	\end{equation}
Combining \eqref{Propo:Collision:13}	 and \eqref{Propo:Collision:14} then using \eqref{Propo:Collision:10} we obtain

	\begin{equation}\label{Propo:Collision:15}	 	c_t\int_{\cup_{i\in\mathfrak{Z}_t \cup\{l_t^{**}\}}\Delta_{i}^{h,R} }\mathrm{d}\omega F(t) \ \ge \ \int_{\cup_{i\in\mathfrak{Z}_t\cup\{l_t^{**}\}}\Xi_{i}^{h,R} }\mathrm{d}\omega F(t) \ \ge \ 	(1- \nu)\int_{[0,R)}\mathrm{d}\omega F(t).
\end{equation}
	Let us consider the two cases. We suppose in the first case that 
		\begin{equation}\label{Propo:Collision:16}
		\begin{aligned}
			\int_{\Delta_{l_t^{**}}^{h,R} }\mathrm{d}\omega F(t) \ \ge \frac{ \nu}{20}\int_{[0,R)}\mathrm{d}\omega F(t).
		\end{aligned}
	\end{equation}
	Then by \eqref{Propo:Collision:6}-\eqref{Propo:Collision:16}, we find
		\begin{equation}\label{Propo:Collision:17}
		\begin{aligned}
	\int_{\mathscr X^2 } \mathrm{d}\omega G(t) \ = \	\int_{\cup_{i\in\mathfrak{Z}_t }\Delta_{i}^{h,R} }\mathrm{d}\omega F(t) \ \ge \	\int_{\Delta_{l_t^{*}}^{h,R} }\mathrm{d}\omega F(t) \ \ge \frac{ \nu}{20}\int_{[0,R)}\mathrm{d}\omega F(t).
		\end{aligned}
	\end{equation}
In the second case, we suppose the opposite of \eqref{Propo:Collision:16}
		\begin{equation}\label{Propo:Collision:18}
		\begin{aligned}
			\int_{\Delta_{l_t^{**}}^{h,R} }\mathrm{d}\omega F(t) \ < \frac{ \nu}{20}\int_{[0,R)}\mathrm{d}\omega F(t),
		\end{aligned}
	\end{equation}
	from which, in combination with \eqref{Propo:Collision:15}, we deduce 
	
		\begin{equation}\label{Propo:Collision:19}	 c_t\int_{\cup_{i\in\mathfrak{Z}_t }\Delta_{i}^{h,R} }\mathrm{d}\omega F(t)  + 	 \frac{\nu}{20}\int_{[0,R)}\mathrm{d}\omega F(t)  \ \ge \ 	(1- \nu)\int_{[0,R)}\mathrm{d}\omega F(t),
	\end{equation}
	yielding
		\begin{equation}\label{Propo:Collision:20}	  c_t\int_{\cup_{i\in\mathfrak{Z}_t }\Delta_{i}^{h,R} }\mathrm{d}\omega F(t)    \ \ge \ 	(1- \nu- \frac{\nu}{20})\int_{[0,R)}\mathrm{d}\omega F(t).
	\end{equation}
As a consequence,  we obtain 

		\begin{equation}\label{Propo:Collision:21}	 \int_{\mathscr X^2 } \mathrm{d}\omega F(t)    \ \ge \ 	\frac{(1-2\nu)}{c_t}\int_{[0,R)}\mathrm{d}\omega F(t)\ \ge \ 	\frac{1-2\nu}{20}\int_{[0,R)}\mathrm{d}\omega F(t)\ \ge \  \frac{\nu}{20}\int_{[0,R)}\mathrm{d}\omega F(t).
	\end{equation}
Combining \eqref{Propo:Collision:11}, \eqref{Propo:Collision:17}, \eqref{Propo:Collision:21}, we have in both cases
	\begin{equation}\label{Propo:Collision:22}	 \int_{\mathscr X^1 } \mathrm{d}\omega F(t)    \ \ge \ 	   \frac{ \nu}{20}\int_{[0,R)}\mathrm{d}\omega F(t), \ \ \   \int_{\mathscr X^2 } \mathrm{d}\omega F(t)    \ \ge \ 	   \frac{ \nu}{20}\int_{[0,R)}\mathrm{d}\omega F(t). 
\end{equation}
The two new sets $\mathscr X^1=\cup_{i\in\mathfrak{Y}_t }\Delta_{i}^{h,R} , \mathscr X^2=\cup_{i\in\mathfrak{Z}_t }\Delta_{i}^{h,R}$ are the two new non-overlapping subdomains that we wanted to construct. 


{\it Step 3: Distributing the small subdomains into the two big non-overlapping subdomains.} 

Let us consider    $l_1\ge l_2\in  \mathfrak{Z}_t $, $l_3\in \mathfrak{Y}_t$. There are several possibilities.

{\it Possibility 1:} If $l_1> l_2 >l_3$, then since $\mathscr X^1\cap\cup_{i\in\mathfrak{Z}_t }\Xi_{i}^{h,R}=\emptyset$, we have $\mathscr N_{R,h} > l_1> l_2 >l_3+1.$  Moreover, by \eqref{Propo:Collision:6a}, we deduce $l_1\ge l_2+2$. Therefore, $(l_1,l_2,l_3)\in\{(i,j,l) \in\mathbb{Z}^3 ~~ |~~ i,j,l\ge0,\  \mathscr N_{R,h} > i\geq j>l+1, i\ge j+2\}$. We estimate
\begin{equation}\label{Propo:Collision:5:2}
		\begin{aligned}
			&  \sum_{(i,j,l)\in\mathfrak{X}}\iiint_{\Delta_{i,j,l}^{h,R} }\mathrm{d}\omega_1\,\mathrm{d}\omega_2\,\mathrm{d}\omega F_1F_2F\mathbf{1}_{\omega+\omega_1-\omega_2\ge 0}\mho(\omega_{Max}-\omega_{Min}+\omega_{Mid})\\
		\ge\	  &  \mho(h)\sum_{(i,j,l)\in\{(i,j,l) \in\mathbb{Z}^3 ~~ |~~ i,j,l\ge0,\  \mathscr N_{R,h} > i\geq j>l+1, i\ge j+2\}}\iiint_{\Delta_{i,j,l}^{h,R} }\mathrm{d}\omega_1\,\mathrm{d}\omega_2\,\mathrm{d}\omega F_1F_2F\mathbf{1}_{\omega+\omega_1-\omega_2\ge 0}\\
	\ge\	  & 	\mho(h)\sum_{l_1, l_2\in  \mathfrak{Z}_t; l_3\in \mathfrak{Y}_t;l_1> l_2>l_3}\int_{\Delta_{l_2}^{h,R} }\mathrm{d}\omega F \int_{\Delta_{l_1}^{h,R} }\mathrm{d}\omega F\int_{\Delta_{l_3}^{h,R} }\mathrm{d}\omega F.	\end{aligned}
	\end{equation}

{\it Possibility 2:} If $l_1= l_2 >l_3$, then since $\mathscr X^1\cap\cup_{i\in\mathfrak{Z}_t }\Xi_{i}^{h,R}=\emptyset$, we have $\mathscr N_{R,h} > l_1= l_2 >l_3+1.$  Therefore, $(l_1,l_2,l_3)\in\{(i,j,l) \in\mathbb{Z}^3 ~~ |~~ i,j,l\ge0,\  \mathscr N_{R,h} > i\geq j>l+1\}$. We bound
\begin{equation}\label{Propo:Collision:5:3}
		\begin{aligned}
			&  \sum_{(i,j,l)\in\mathfrak{X}}\iiint_{\Delta_{i,j,l}^{h,R} }\mathrm{d}\omega_1\,\mathrm{d}\omega_2\,\mathrm{d}\omega F_1F_2F\mathbf{1}_{\omega+\omega_1-\omega_2\ge 0}\mho(\omega_{Max}-\omega_{Min}+\omega_{Mid})\\
		\ge\	  &  \mho(h)\sum_{(i,j,l)\in\{(i,j,l) \in\mathbb{Z}^3 ~~ |~~ i,j,l\ge0,\  \mathscr N_{R,h} > i\geq j>l+1\}}\iiint_{\Delta_{i,j,l}^{h,R} }\mathrm{d}\omega_1\,\mathrm{d}\omega_2\,\mathrm{d}\omega F_1F_2F\mathbf{1}_{\omega+\omega_1-\omega_2\ge 0}\\
	\ge\	  &	\mho(h)\sum_{l_1, l_2\in  \mathfrak{Z}_t; l_3\in \mathfrak{Y}_t;l_1= l_2 >l_3}\int_{\Delta_{l_2}^{h,R} }\mathrm{d}\omega F \int_{\Delta_{l_1}^{h,R} }\mathrm{d}\omega F\int_{\Delta_{l_3}^{h,R} }\mathrm{d}\omega F.	\end{aligned}
	\end{equation}

{\it Possibility 3:} If $l_3>l_1> l_2$, then by \eqref{Propo:Collision:6a}, we have $\mathscr N_{R,h} > l_3> l_1 >l_2+1.$ Therefore, $(l_3,l_1,l_2)\in\mathfrak{X}$.

We bound	\begin{equation}\label{Propo:Collision:23}
	\begin{aligned}
	&  \sum_{(i,j,l)\in\mathfrak{X}}\iiint_{\Delta_{i,j,l}^{h,R} }\mathrm{d}\omega_1\,\mathrm{d}\omega_2\,\mathrm{d}\omega F_1F_2F\mathbf{1}_{\omega+\omega_1-\omega_2\ge 0}\mho(\omega_{Max}-\omega_{Min}+\omega_{Mid})\\
		\ge\	  &  \mho(h)\    \sum_{l_1,l_2 \in  \mathfrak{Z}_t; l_3\in \mathfrak{Y}_t; l_3>l_1> l_2}\int_{\Delta_{l_1}^{h,R} }\mathrm{d}\omega F\,\int_{\Delta_{l_2}^{h,R} }\mathrm{d}\omega F\int_{\Delta_{l_3}^{h,R} }\mathrm{d}\omega F.
	\end{aligned}
\end{equation}

{\it Possibility 4:} If $l_3>l_1= l_2$, we consider the quantity
 	\begin{equation}\label{Propo:Collision:24}
 	\begin{aligned}
 		 &     \sum_{l_1,l_2 \in  \mathfrak{Z}_t, l_3\in \mathfrak{Y}_t,  l_3>l_1= l_2}\int_{\Delta_{l_1}^{h,R} }\mathrm{d}\omega F\,\int_{\Delta_{l_2}^{h,R} }\mathrm{d}\omega F\int_{\Delta_{l_3}^{h,R} }\mathrm{d}\omega F\\
 	=\	 &     \sum_{l_1  \in  \mathfrak{Z}_t, l_3\in \mathfrak{Y}_t,  l_3>l_1 }\left(\int_{\Delta_{l_1}^{h,R} }\mathrm{d}\omega F\right)^2\int_{\Delta_{l_3}^{h,R} }\mathrm{d}\omega F\\
 	\le\	 &     \sum_{l_3\in \mathfrak{Y}_t  ,  l_3>l_t^*}\left(\int_{\Delta_{l_t^*}^{h,R} }\mathrm{d}\omega F\right) ^2\int_{\Delta_{l_3}^{h,R} }\mathrm{d}\omega F  + \sum_{l_1  \in  \mathfrak{Z}_t, l_3\in \mathfrak{Y}_t,  l_3>l_1\ne l_t^*}\int_{\Delta_{l_t^*}^{h,R} }\mathrm{d}\omega F \int_{\Delta_{l_1}^{h,R} }\mathrm{d}\omega F\int_{\Delta_{l_3}^{h,R} }\mathrm{d}\omega F,
 	\end{aligned}
 \end{equation}
where $l^*_t$ is defined in \eqref{Propo:Collision:6}.
The second quantity on the right hand side of the above inequality can be bounded, by a similar argument used in Possibility 3
	\begin{equation}\label{Propo:Collision:25}
	\begin{aligned}
			 &   \mho(h) \sum_{l_1  \in  \mathfrak{Z}_t, l_3\in \mathfrak{Y}_t,  l_3>l_1\ne l_t^*}\int_{\Delta_{l_t^*}^{h,R} }\mathrm{d}\omega F \int_{\Delta_{l_1}^{h,R} }\mathrm{d}\omega F\int_{\Delta_{l_3}^{h,R} }\mathrm{d}\omega F\\
\le \			 &   \sum_{(i,j,l)\in\mathfrak{X}}\iiint_{\Delta_{i,j,l}^{h,R} }\mathrm{d}\omega_1\,\mathrm{d}\omega_2\,\mathrm{d}\omega F_1F_2F\mathbf{1}_{\omega+\omega_1-\omega_2\ge 0}\mho(\omega_{Max}-\omega_{Min}+\omega_{Mid}).
	\end{aligned}
\end{equation}
Now, we suppose that there exists a constant $\mathfrak{C}>\frac13$ such that 
	\begin{equation}\label{Propo:Collision:26}
	\begin{aligned}
 \sum_{l_3\in \mathfrak{Y}_t ,   l_3>l_t^*} \int_{\Delta_{l_3}^{h,R} }\mathrm{d}\omega F\ > \ \mathfrak{C} \sum_{l_3\in \mathfrak{Y}_t ,   l_3<l_t^*} \int_{\Delta_{l_3}^{h,R} }\mathrm{d}\omega F.
		\end{aligned}
\end{equation}
Inequality \eqref{Propo:Collision:26} yields

		\begin{equation}\label{Propo:Collision:27}
		\begin{aligned}
			\sum_{l_3\in \mathfrak{Y}_t ,   l_3>l_t^*} \int_{\Delta_{l_3}^{h,R} }\mathrm{d}\omega F\ > \ \frac{\mathfrak{C}}{\mathfrak{C}+1} \sum_{l_3\in \mathfrak{Y}_t } \int_{\Delta_{l_3}^{h,R} }\mathrm{d}\omega F \ \ge \ 	   \frac{\mathfrak{C}}{\mathfrak{C}+1}\frac{ \nu}{20}\int_{[0,R)}\mathrm{d}\omega F,
		\end{aligned}
	\end{equation}
where we have used \eqref{Propo:Collision:22}. Combining \eqref{Propo:Collision:8} 
 and \eqref{Propo:Collision:27}, we obtain
	\begin{equation}\label{Propo:Collision:28}
	\begin{aligned}
		\sum_{l_3\in \mathfrak{Y}_t ,   l_3>l_t^*} \int_{\Delta_{l_3}^{h,R} }\mathrm{d}\omega F\ > \ 	   \frac{\mathfrak{C}}{\mathfrak{C}+1}\frac{ \nu}{20}\int_{[0,R)}\mathrm{d}\omega F\ > \ 	   \frac{\mathfrak{C}}{\mathfrak{C}+1}\frac{ \nu}{20(1- \nu)}\int_{\Delta_{l_t^*}^{h,R} }\mathrm{d}\omega F,
	\end{aligned}
\end{equation}
	yielding
		\begin{equation}\label{Propo:Collision:29}
		\begin{aligned}
		    & \sum_{l_3\in \mathfrak{Y}_t  ,  l_3>l_t^*}\left(\int_{\Delta_{l_t^*}^{h,R} }\mathrm{d}\omega F\right) ^2\int_{\Delta_{l_3}^{h,R} }\mathrm{d}\omega F\\  \ \le \ &  \frac{\mathfrak{C}+1}{\mathfrak{C}}\frac{20(1- \nu)}{ \nu}	\left(\int_{\Delta_{l_t^*}^{h,R} }\mathrm{d}\omega F\right)\left[\sum_{l_3\in \mathfrak{Y}_t ,   l_3>l_t^*} \int_{\Delta_{l_3}^{h,R} }\mathrm{d}\omega F\right]^2\\
		    \ \le \ &  \frac{\mathfrak{C}_0}{ \nu}	\left(\int_{\Delta_{l_t^*}^{h,R} }\mathrm{d}\omega F\right)\left[\sum_{l_3\in \mathfrak{Y}_t ,   l_3>l_t^*} \int_{\Delta_{l_3}^{h,R} }\mathrm{d}\omega F\right]^2\\
		    \ \le \ &\sum_{l_3,l_3'\in \mathfrak{Y}_t ,   l_3,l_3'>l_t^*}  \frac{\mathfrak{C}_0}{ \nu}	\left(\int_{\Delta_{l_t^*}^{h,R} }\mathrm{d}\omega F \int_{\Delta_{l_3}^{h,R} }\mathrm{d}\omega F\int_{\Delta_{l_3'}^{h,R} }\mathrm{d}\omega F\right),
		\end{aligned}
	\end{equation}
for some universal constant $\mathfrak{C}_0>0$ that varies from estimates to estimates. Since $\mathscr X^1\cap \Xi_{l_t^*}^{h,R}=\emptyset$, we have $\mathscr N_{R,h} > l_3,l_3'>l_t^*+1.$ The argument of Possibility 1 can be reapplied to \eqref{Propo:Collision:29}, leading to
	\begin{equation}\label{Propo:Collision:30}
	\begin{aligned}
		&\mho(h) \sum_{l_3\in \mathfrak{Y}_t ,  l_3>l_t^*}\left(\int_{\Delta_{l_t^*}^{h,R} }\mathrm{d}\omega F\right) ^2\int_{\Delta_{l_3}^{h,R} }\mathrm{d}\omega F\\  
		\ \le \ &  \frac{\mathfrak{C}_0}{ \nu}\sum_{(i,j,l)\in\mathfrak{X}}\iiint_{\Delta_{i,j,l}^{h,R} }\mathrm{d}\omega_1\,\mathrm{d}\omega_2\,\mathrm{d}\omega F_1F_2F\mathbf{1}_{\omega+\omega_1-\omega_2\ge 0}\mho(\omega_{Max}-\omega_{Min}+\omega_{Mid}),
	\end{aligned}
\end{equation}
for some universal constant $\mathfrak{C}_0>0$ that varies from estimates to estimates.

%

	Next, we suppose that \eqref{Propo:Collision:26} does not happen, then
		\begin{equation}\label{Propo:Collision:31}
		\begin{aligned}
			\sum_{l_3\in \mathfrak{Y}_t ,   l_3>l_t^*} \int_{\Delta_{l_3}^{h,R} }\mathrm{d}\omega F\ \le \ \mathfrak{C} \sum_{l_3\in \mathfrak{Y}_t ,   l_3<l_t^*} \int_{\Delta_{l_3}^{h,R} }\mathrm{d}\omega F
		\end{aligned}
	\end{equation}
	yielding
\begin{equation}\label{Propo:Collision:31}
	\begin{aligned}
		& \sum_{l_3\in \mathfrak{Y}_t ,   l_3>l_t^*}\left(\int_{\Delta_{l_t^*}^{h,R} }\mathrm{d}\omega F\right) ^2\int_{\Delta_{l_3}^{h,R} }\mathrm{d}\omega F\\  \ \le \ &  \mathfrak{C}	\left(\int_{\Delta_{l_t^*}^{h,R} }\mathrm{d}\omega F\right)^2\left[\sum_{l_3\in \mathfrak{Y}_t ,   l_3<l_t^*} \int_{\Delta_{l_3}^{h,R} }\mathrm{d}\omega F\right],
	\end{aligned}
\end{equation}
and
\begin{equation}\label{Propo:Collision:27:aa}
		\begin{aligned}
			\sum_{l_3\in \mathfrak{Y}_t ,   l_3<l_t^*} \int_{\Delta_{l_3}^{h,R} }\mathrm{d}\omega F\ > \ \frac{\mathfrak{C}}{\mathfrak{C}+1} \sum_{l_3\in \mathfrak{Y}_t } \int_{\Delta_{l_3}^{h,R} }\mathrm{d}\omega F \ \ge \ 	   \frac{\mathfrak{C}}{\mathfrak{C}+1}\frac{ \nu}{20}\int_{[0,R)}\mathrm{d}\omega F\ \ge \ 	   \frac{ \nu}{80}\int_{[0,R)}\mathrm{d}\omega F.
		\end{aligned}
	\end{equation}
The argument of \eqref{Propo:Collision:30} can be reapplied to \eqref{Propo:Collision:31}, leading to
	\begin{equation}\label{Propo:Collision:32:1}
	\begin{aligned}
		& \mathfrak{C}	\left(\int_{\Delta_{l_t^*}^{h,R} }\mathrm{d}\omega F\right)^2\left[\sum_{l_3\in \mathfrak{Y}_t ,   l_3<l_t^*} \int_{\Delta_{l_3}^{h,R} }\mathrm{d}\omega F\right]\\  
		\ \le \ &   \frac{\mathfrak{C}_0}{ \mho(h)\nu}\sum_{(i,j,l)\in\mathfrak{X}}\iiint_{\Delta_{i,j,l}^{h,R} }\mathrm{d}\omega_1\,\mathrm{d}\omega_2\,\mathrm{d}\omega F_1F_2F\mathbf{1}_{\omega+\omega_1-\omega_2\ge 0}\mho(\omega_{Max}-\omega_{Min}+\omega_{Mid}),
	\end{aligned}
\end{equation}
for some universal constant $\mathfrak{C}_0>0$ that varies from estimates to estimates.

Finally, we bound
\begin{equation}\label{Propo:Collision:32}
	\begin{aligned}
 		 &     \sum_{l_1,l_2 \in  \mathfrak{Z}_t, l_3\in \mathfrak{Y}_t,  l_3>l_1= l_2}\int_{\Delta_{l_1}^{h,R} }\mathrm{d}\omega F\,\int_{\Delta_{l_2}^{h,R} }\mathrm{d}\omega F\int_{\Delta_{l_3}^{h,R} }\mathrm{d}\omega F\\
		\ \le \ &   \frac{\mathfrak{C}_0}{ \mho(h)\nu}\sum_{(i,j,l)\in\mathfrak{X}}\iiint_{\Delta_{i,j,l}^{h,R} }\mathrm{d}\omega_1\,\mathrm{d}\omega_2\,\mathrm{d}\omega F_1F_2F\mathbf{1}_{\omega+\omega_1-\omega_2\ge 0}\mho(\omega_{Max}-\omega_{Min}+\omega_{Mid}),
	\end{aligned}
\end{equation}
for some universal constant $\mathfrak{C}_0>0$ that varies from estimates to estimates.

Combining \eqref{Propo:Collision:5:2}, \eqref{Propo:Collision:5:3}, \eqref{Propo:Collision:23}, \eqref{Propo:Collision:30}, \eqref{Propo:Collision:32},   we obtain 
	\begin{equation}\label{Propo:Collision:33}\begin{aligned}	& \left(\int_{\mathscr X^1 } \mathrm{d}\omega F(t)  \right)\left(\int_{\mathscr X^2 } \mathrm{d}\omega F(t)  \right)^2\\ \ \le \ &	 \frac{\mathfrak{C}_0}{ \mho(h)\nu}\sum_{(i,j,l)\in\mathfrak{X}}\iiint_{\Delta_{i,j,l}^{h,R} }\mathrm{d}\omega_1\,\mathrm{d}\omega_2\,\mathrm{d}\omega F_1F_2F\mathbf{1}_{\omega+\omega_1-\omega_2\ge 0}\mho(\omega_{Max}-\omega_{Min}+\omega_{Mid}), \end{aligned}
\end{equation}
which, together with \eqref{Propo:Collision:22}, implies
	\begin{equation}\label{Propo:Collision:33}	 \left( \int_{[0,R)}\mathrm{d}\omega F(t)\right)^3 \ \le \ 	\frac{\mathfrak{C}_0}{\nu^4 \mho(h)}\sum_{(i,j,l)\in\mathfrak{X}}\iiint_{\Delta_{i,j,l}^{h,R} }\mathrm{d}\omega_1\,\mathrm{d}\omega_2\,\mathrm{d}\omega F_1F_2F\mathbf{1}_{\omega+\omega_1-\omega_2\ge 0}\mho(\omega_{Max}-\omega_{Min}+\omega_{Mid}), 
\end{equation}
for some universal constant $\mathfrak{C}_0>0$ that varies from estimates to estimates. 
The two inequalities \eqref{Propo:Collision:5} and \eqref{Propo:Collision:33}	can be combined and we then find

	\begin{equation}\label{Propo:Collision:34}
	\begin{aligned}
	\frac{\mathfrak{C}_0R^{2+\frac2\alpha}(	\mathscr M +\mathscr E)}{h^{2}\mho(h) \nu^4}
		\ \ge\ &\int_{\digamma^2}\mathrm{d}t  \left( \int_{[0,R)}\mathrm{d}\omega F(t)\right)^3.
	\end{aligned}
\end{equation}
for some universal constants $\mathfrak{C}_0$ that varies from estimates to estimates, yielding the conclusion of our proposition. 
\end{proof}
\section{Condensation times}
\subsection{Supersolution for an integro-differential equation}
In the following lemma, we establish a construction for a supersolution of an  integro-differential equations, that could be considered as an extension of the previous sub/supersolution constructions done in \cite{caffarelli2009regularity}[Lemma 8.1] and \cite{schwab2016regularity}[Lemma 4.2]. The main improvement in  our consideration, in comparison with the previous ones (i.e. \cite{caffarelli2009regularity}[Lemma 8.1] and \cite{schwab2016regularity}[Lemma 4.2]),  is that the kernel $\mu$ is time-dependent. 
\begin{lemma}\label{Lemma:Supersol}
	Let $z\mu(t,z),z^2\mu(t,z)\in L^\infty([0,\infty),L^1([0,\infty)))$ and suppose that $\mu(t,z)\ge 0$ for a.e. $(t,z)\in [0,\infty)^2$. We consider the equation
	\begin{equation}
		\label{Lemma:Supersol:1}\partial_t\rho(t,x)  \ + \ \int_{0}^\infty\mathrm{d}z \mu(t,z)[\rho(t,x+z)-\rho(t,x)] \ = \ 0.
	\end{equation}
	Let $\Upsilon(z): [0,\infty)\to [0,\infty) $ be a function such that $\Upsilon'(z)$ and $\Upsilon''(z)$ exists for a.e. $z$ in $[0,\infty) $. Moreover,
	\begin{equation}
		\label{Lemma:Supersol:2}\Upsilon''(z_1)\ge  C_\Upsilon\Upsilon(z_2),
	\end{equation}
for a.e. $z_1,z_2\in [0,\infty),$ for some constant $C_\Upsilon>0$. 
We set, for a fixed time $T_1>0$
	\begin{equation}\begin{aligned}
	\label{Lemma:Supersol:3}A_\mu(t) \ = \  &  \int_{t}^{T_1}\mathrm{d}s\int_{0}^\infty\mathrm{d}z \mu(s,z)z,\\ \ \ \ B_\mu(t) \ = \ &  C_\Upsilon \int_{t}^{T_1}\mathrm{d}s\int_{0}^\infty\mathrm{d}z \mu(s,z)z^2/2.\end{aligned}
\end{equation}
Then 
	\begin{equation}\begin{aligned}
		\label{Lemma:Supersol:4}  &  \rho_{super}(t,x)\ : = \ e^{B_\mu}\Upsilon(A_\mu+x),\end{aligned}
\end{equation}
is a supersolution of \eqref{Lemma:Supersol:1}, in the sense that
	\begin{equation}
	\label{Lemma:Supersol:5}\partial_t\rho_{super}(t,x)  \ + \ \int_{0}^\infty\mathrm{d}z \mu(t,z)[\rho_{super}(t,x+z)-\rho_{super}(t,x)] \ \ge \ 0,
\end{equation}
for a.e. $(t,x)\in [0,\infty)^2$. 
\end{lemma}
\begin{proof}
	Using Taylor's expansion, we compute
		\begin{equation}
		\label{Lemma:Supersol:E1}\begin{aligned}&   \partial_t\rho_{super}(t,x)  \ + \ \int_{0}^\infty\mathrm{d}z \mu(t,z)[\rho_{super}(t,x+z)-\rho_{super}(t,x)]\\
			 \ = \ & \partial_t\rho_{super}(t,x)  \ + \ \partial_x\rho_{super}(t,x)  \int_{0}^\infty\mathrm{d}z \mu(t,z)z\\
			 & \ +\ \int_{0}^\infty\mathrm{d}z \mu(t,z)\int_0^z\mathrm{d}\zeta \partial^2_{xx}\rho_{super}(x+\zeta)(z-\zeta).\end{aligned}
	\end{equation}
	Plugging the form \eqref{Lemma:Supersol:4} into \eqref{Lemma:Supersol:1}, we get
		\begin{equation}\label{Lemma:Supersol:E2}\begin{aligned}
&    \partial_t\rho_{super}(t,x)   \ + \ \int_{0}^\infty\mathrm{d}z \mu(t,z)[\rho_{super}(t,x+z)-\rho_{super}(t,z)]\\
 \ = \ & -C_\Upsilon e^{B_\mu(t)} \Upsilon(A_\mu(t)+x)\int_{0}^\infty\mathrm{d}z \mu(t,z)z^2/2 \ - \ e^{B_\mu(t)} \Upsilon'(A_\mu(t)+x) \int_{0}^\infty\mathrm{d}z \mu(t,z)z\\
 & \ + \ e^{B_\mu(t)} \Upsilon'(A_\mu(t)+x) \int_{0}^\infty\mathrm{d}z \mu(t,z)z\\
 &   \ +\ \int_{0}^\infty\mathrm{d}z \mu(t,z)\int_0^z\mathrm{d}\zeta \Upsilon''(A_\mu(t)+x+\zeta)(z-\zeta)e^{B_\mu(t)},
\end{aligned}
	\end{equation}
 which, in combination with \eqref{Lemma:Supersol:2}, yields
	\begin{equation}\label{Lemma:Supersol:E3}\begin{aligned}
		&   \partial_t\rho_{super}(t,x)   \ + \ \int_{0}^\infty\mathrm{d}z \mu(t,z)[\rho_{super}(t,x+z)-\rho_{super}(t,x)]\\
		\ \ge \ & -C_\Upsilon e^{B_\mu(t)} \Upsilon(A_\mu(t)+x)\int_{0}^\infty\mathrm{d}z \mu(t,z)z^2/2\\
		&  \ +\ C_\Upsilon e^{B_\mu(t)}\Upsilon(A_\mu(t)+x)\int_{0}^\infty\mathrm{d}z \mu(t,z)\int_0^z\mathrm{d}\zeta (z-\zeta)\\
		\ \ge \ & -C_\Upsilon e^{B_\mu(t)} \Upsilon(A_\mu(t)+x)\int_{0}^\infty\mathrm{d}z \mu(t,z)z^2/2    \ +\ C_\Upsilon e^{B_\mu(t)}\Upsilon(A_\mu(t)+x)\int_{0}^\infty\mathrm{d}z \mu(t,z)z^2/2 \ \ge\ 0.
	\end{aligned}
\end{equation}
That concludes the proof of our lemma. 
\end{proof}
\subsection{Growth lemmas}
The following lemmas play the roles of the growth lemmas studied in the theory of nonlocal integro-differential equations, that provide bounds on sets where the solutions exhibit some  growths (see, for instance \cite{schwab2016regularity}[Section 4]).

		In addition, we set $\frac{2}{\alpha}\frac{\beta+1}{\beta+2}+\beta\frac{\beta+3}{\beta+2}<\varepsilon< \frac{2}{\alpha}+\beta$.

Let $n\in\mathbb{N}, n>100$. We set  $\mathfrak{L}_{n}=2^{-n}$.   We also set $\mathfrak{h}_n=\frac{2^{-n}}{2^{\mathfrak{M}_0(n)}}$, $\mathfrak{M}_0(n)=\left\lfloor	{n(-\varepsilon+2/\alpha+\beta)}\right\rfloor	$. 
Employing the definition \eqref{Sec:FV:2}, we obtain
 
 \begin{equation}\label{Sec:Growthlemmas:1}
 	\begin{aligned}
 		&\mbox{(A) Number of Subdomains:  } 	\mathscr N_{\mathfrak{h}_n, \mathfrak{L}_{n}}  = 2^{\mathfrak{M}_0(n)}; \\ \ 
 		& \mbox{(B) Nonoverlapping Subdomains:  } 	\\
 		&\Delta_{i}^{\mathfrak{h}, \mathfrak{L}_{n}} =\left[ i\mathfrak{h}_n ,(1+i)\mathfrak{h}_n \right) \ \ ,\ \
 		i=0, \cdots, 2^{\mathfrak{M}_0(n)}-2,\ \
   	\Delta_{2^{\mathfrak{M}_0(n)}-1}^{\mathfrak{h}_n, \mathfrak{L}_{n}} =\left[ (2^{\mathfrak{M}_0(n)}-1)\mathfrak h_n ,\mathfrak{L}_{n}\right);  \\
 		& \mbox{(C) Overlapping Subdomains:  } \\
 		&	\Xi_{i}^{\mathfrak{h}_n, \mathfrak{L}_{n}} =\left[ (i-1)\mathfrak{h}_n ,(2+i)\mathfrak{h}_n \right) \ \ ,\ \
 		i=1, \cdots, 2^{\mathfrak{M}_0(n)}-3, \\
 		&	\Xi_{0}^{\mathfrak{h}_n, \mathfrak{L}_{n}} =\left[0, 2\mathfrak{h}_n \right), \ \	\Xi_{1}^{\mathfrak{h}_n, \mathfrak{L}_{n}} =\left[0, 3\mathfrak h_n\right),\ \
 		 \Xi_{2^{\mathfrak{M}_0(n)}-2}^{\mathfrak{h}_n, \mathfrak{L}_{n}} =\left[ (2^{\mathfrak{M}_0(n)}-3)\mathfrak{h}_n,\mathfrak{L}_{n}\right), \\ \ &	\Xi_{2^{\mathfrak{M}_0(n)}-1}^{\mathfrak{h}_n, \mathfrak{L}_{n}} =\left[ (2^{\mathfrak{M}_0(n)}-2)\mathfrak{h}_n ,\mathfrak{L}_{n}\right).\end{aligned}
 \end{equation}

  We set      \begin{equation}\label{varsigma}
  0\le\varsigma\le\frac12\min\Big\{\frac{\beta+2}{3}\Big[\varepsilon-\frac{2}{\alpha}\frac{\beta+1}{\beta+2}-\beta\frac{\beta+3}{\beta+2}\Big],\frac{\varepsilon-2\beta}{3}\Big\},
  \end{equation}  we define using the notation of \eqref{FDefinition} 

\begin{equation}\label{Sec:Growthlemmas:2}
	\mathfrak{S}_n^T=\left\{ t\in\left[ 0,T\right] :\int_{\left[ 0,\mathfrak{L}_n\right)
	}\mathrm{d}\omega F\left( t,\omega\right)  \geq \mathscr{C}_F \mathfrak L_{n} 
^\varsigma\right\},
\end{equation} 
\begin{equation}\label{Sec:Growthlemmas:3}
		\mathfrak{S}_{n,i}^T=\left\{ t\in\left[ 0,T\right] :\text{ such that }\int_{
	\Xi_{i}^{\mathfrak{h_n},\mathfrak{L}_{n}}} \mathrm{d}\omega F\left( t,\omega\right) \geq  \mathscr{C}_F\mathfrak L_{n+1} 
^\varsigma\right\} , 
\end{equation}
for $	i=0, \cdots, 2^{\mathfrak{M}_0(n)}-1$;
%
 We will also need the following definitions
\begin{equation}\label{Sec:Growthlemmas:4}
	\mathscr{U}_{n}^T= \mathfrak{S}_{n}^T\setminus\bigcup_{i=0}^{2^{\mathfrak{M}_0(n)}-1}\mathfrak{S}^T_{n,i},
\end{equation}
\begin{equation}\label{Sec:Growthlemmas:5}
	\mathscr{V}_{n}^T= \bigcup_{i=2^{{\mathfrak{M}_0(n)}-1}-1}^{2^{\mathfrak{M}_0(n)}-1}\mathfrak{S}^T_{n,i},
\end{equation}
\begin{equation}\label{Sec:Growthlemmas:6}
	\mathscr{W}_{n}^T= \bigcup_{i=0}^{2^{\mathfrak{M}_0(n)-1}-2}\mathfrak{S}^T_{n,i}.
\end{equation}


We define $T_{\mathfrak{h}_n,\mathfrak{L}_{n}}<T$ to be  a time that satisfies (such a time may not exist)	\begin{equation}	\label{Lemma:Growthlemmas:9}
 	\int_{0}^{{T}_{\mathfrak{h}_n,\mathfrak{L}_{n}}} \mathrm d t  \left[\int_{\mathbb{R}_+ }\mathrm d\omega
 	\Big[ \mathcal G_{\mathfrak{h}_n,\mathfrak{L}_{n}}(t,\omega)\Big]\right]
 	^{2}  = \Big[{{C_1} C_\omega^{\frac2\alpha} }\Big]^{-1}\frac{\mathfrak{L}_{n}^{\varepsilon-\beta}}{4000\mho(\mathfrak{L}_n/2)},
 \end{equation}
where   the constants are defined in \eqref{Settings2}-\eqref{Settings3}-\eqref{Lemma:Concave:E5}   
and 
$$\mathcal G_{\mathfrak{h}_n,\mathfrak{L}_{n}}(t,\omega) \ :=\ F(t,\omega)\chi_{\omega\in \Xi_{\mathfrak{i}\left( t\right)}^{\mathfrak{h}_n,\mathfrak{L}_{n}}}\chi_{\mathscr{V}_{n}^{T_{\mathfrak{h}_n,\mathfrak{L}_{n}}}}(t).$$ 
Before going into the details of the ``growth lemmas'', we will need to construct an upper solution for an integro-differential equation
	following Lemma \ref{Lemma:Supersol}. To this end, we define, using the notation of Lemma \ref{Lemma:Supersol}
\begin{equation}\label{Lemma:Growthlemmas:7}\begin{aligned}
		\mu\left( t,z\right) \ := & \ \frac{\mho(\mathfrak{L}_n/2){C_1} C_\omega^{\frac2\alpha} }{\mathfrak L_n^{\frac2\alpha}}\int_{\mathbb{R}_+ }\mathrm d\omega
	    \Big[ \mathcal G_{\mathfrak{h}_n,\mathfrak{L}_{n}}(t,\omega)\Big]\Big[\mathcal G_{\mathfrak{h}_n,\mathfrak{L}_{n}}(t,\omega+z)\Big],\end{aligned}
\end{equation}
for $(t,z)\in [0,\infty)^2$ and $C_1$ is defined in \eqref{Lemma:Concave:E5}.
In the next lemma, we will construct a super solution for \eqref{Lemma:Supersol:1} with  $\mu$ defined in  \eqref{Lemma:Growthlemmas:7}
	\begin{equation}
	\label{Lemma:Growthlemmas:8}\partial_t\rho_{super}(t,x)  \ + \ \int_{0}^\infty\mathrm{d}z \mu(t,z)[\rho_{super}(t,x+z)-\rho_{super}(t,x)] \ \ge \ 0,
\end{equation}
for a.e. $(t,x)\in [0,\infty)^2$.

\begin{lemma}
	\label{Lemma:GrowthFinal}  Suppose that there exists a sufficiently large natural number $N_0$ such that 
		\begin{equation}	\label{Lemma:GrowthFinal:1}
		\begin{aligned}
			  &    \int_{0}^{\mathfrak{L}_{N_0+1}}\mathrm{d}\omega F(0,\omega)\ \ge \  \mathscr{C}_F^*=2\mathscr{C}_F|\mathfrak{L}_{N_0}|^\varsigma, 
	\end{aligned}\end{equation}
then the set $
	\mathfrak{S}_{N_0}^T$ is identical to $[0,T]$, 
for all $T\in[0,T_0)$.
\end{lemma}

\begin{proof}
Using  the test function $
		\varphi\left( \omega\right) =\left( \mathfrak{L}_{N_0}-\omega\right)
		_{+} $
in the place of $\rho_{super}$ in \eqref{Lemma:Growth1:E12}, we find
	
	\begin{equation}	\label{Lemma:GrowthFinal:E1}
	\begin{aligned}
		&	\partial_{t}\left( \int_{\mathbb{R}^{+}}\mathrm{d}\omega F\varphi\right)\  =\ \  C_1\iiint_{\mathbb{R}_+^{3}}\mathrm{d}\omega_1\,\mathrm{d}\omega_2\,\mathrm{d}\omega f_1f_2f\mathbf{1}_{\omega+\omega_1-\omega_2\ge 0}\\
		&\times 	 \Big\{ [-\varphi(\omega_{Max})-\varphi(\omega_{Min})+\varphi(\omega_{Mid})+\varphi(\omega_{Max}+\omega_{Min}-\omega_{Mid})]\\
		&\ \ \ \ \times\mho(\omega_{Max})\mho(\omega_{Min})\mho(\omega_{Mid})\mho(\omega_{Max}+\omega_{Min}-\omega_{Mid})|k_{Min}|  \\
		&+[-\varphi(\omega_{Max})-\varphi(\omega_{Mid})+\varphi(\omega_{Min})+\varphi(\omega_{Max}+\omega_{Mid}-\omega_{Min})]\\
		&\ \ \ \ \times\mho(\omega_{Max})\mho(\omega_{Min})\mho(\omega_{Mid})\mho(\omega_{Max}+\omega_{Mid}-\omega_{Min})|k_{Min}|  \\
		&+[-\varphi(\omega_{Min})-\varphi(\omega_{Mid})+\varphi(\omega_{Max})+\varphi(\omega_{Min}+\omega_{Mid}-\omega_{Max})]\\
		&\ \ \ \ \times\Xi(\omega_{Max},\omega_{Min},\omega_{Mid},\omega_{Min}+\omega_{Mid}-\omega_{Max})\Big\} \ \ge 0,
\end{aligned}\end{equation}
where the same arguments as in  \eqref{Lemma:Growth1:E13}-\eqref{Lemma:Growth1:E14a} have been used. As a consequence, we obtain
	\begin{equation}	\label{Lemma:GrowthFinal:E2}
	\begin{aligned}
		\mathfrak{L}_{N_0} \int_{[0,\mathfrak{L}_{N_0})}\mathrm{d}\omega F(t,\omega)\  \ge \  &   \int_{\mathbb{R}^{+}}\mathrm{d}\omega F(t,\omega)\left[ \left( \mathfrak{L}_{N_0}-\omega\right)_+\right] 
		\  \ge \     \  \int_{\mathbb{R}^{+}}\mathrm{d}\omega  F(0,\omega)\left[ \left( \mathfrak{L}_{N_0}-\omega\right)_+ \right]\\
		\  \ge \  &  \  \int_{0}^{\mathfrak{L}_{N_0+1}}\mathrm{d}\omega F(0,\omega)\left[  \left( \mathfrak{L}_{N_0}-\omega\right)_+ \right] 
		\  \ge \     \  \mathfrak{L}_{N_0+1}\int_{0}^{\mathfrak{L}_{N_0+1}}\mathrm{d}\omega F(0,\omega),
\end{aligned}\end{equation}
yielding
	\begin{equation}	\label{Lemma:GrowthFinal:E3}
	\begin{aligned}
	  2\int_{[0,\mathfrak{L}_{N_0})}\mathrm{d}\omega F(t,\omega)\  \ge \  &    \int_{0}^{\mathfrak{L}_{N_0+1}}\mathrm{d}\omega F(0,\omega)\ \ge \   \mathscr{C}_F^*=2\mathscr{C}_F|\mathfrak{L}_{N_0}|^\varsigma, 
\end{aligned}\end{equation}
which, in combination with the definition \eqref{Sec:Growthlemmas:2} gives
\begin{equation}\label{Lemma:GrowthFinal:E4} 
	\mathfrak{S}_{N_0}^T= [0,T], 
\end{equation} 
for all $T\in[0,T_0)$.
\end{proof}

\begin{lemma}[Supersolution of an integro-differential equation]
	\label{Lemma:SuperSolu} Suppose that the time  $T_{\mathfrak{h}_n,\mathfrak{L}_{n}}$ defined in \eqref{Lemma:Growthlemmas:9} exists, then there exists a super solution  $\rho_{super}(t,x)$ for \eqref{Lemma:Supersol:1} with  $\mu$ defined in  \eqref{Lemma:Growthlemmas:7} such that $\rho_{super}(t,x) \in L^{\infty}\left( \left[ 0,T_{\mathfrak{h}_n,\mathfrak{L}_{n}}\right]\times\mathbb{R}^{+}\right)$ and  \begin{itemize}
		\item[(a)] $0\le \rho_{super}(t,x)\le 0.2$  for a.e. $(t,x)\in [0,\infty)^2$. 
				\item[(b)] $\rho_{super}(t,x)\ge 0.05$  for a.e.  $x\in\Big[0,\frac{\mathfrak{L}_n}{40}\Big]$ and $\mathrm{supp}\rho_{super}(t,x)\subset\Big[0,\frac{\mathfrak{L}_n}{10}\Big]$ for all $t\in[0,T_{\mathfrak{h}_n,\mathfrak{L}_{n}}]$.
				
					\item[(c)] $\partial_x\rho_{super}(t,x)\le 0$ and $\partial_{xx}\rho_{super}(t,x)\ge 0$
				a.e. $(t,x)\in [0,\infty)^2$. 
	\end{itemize}

\end{lemma}

\begin{proof}
	The proof  is an application of  Lemma \ref{Lemma:Supersol}. To this end, we define  $\Upsilon(z): [0,\infty)\to [0,\infty) $, using the same notations of Lemma \ref{Lemma:Supersol}, as follow
	\begin{equation}\label{Lemma:SuperSolu:E1}
		\Upsilon\left( z\right)\ := \ \exp\left[ \frac{1}{\mathfrak{L}_{n}}\left( 
		\frac{\mathfrak{L}_{n}}{10}-z\right) _{+}\right] -1 \  := \ \exp\left[ 	\frac{1}{\mathfrak{L}_{n}}\max\left\{\frac{\mathfrak{L}_{n}}{10}-z,0\right\} \right] -1 .
	\end{equation}
	Following \eqref{Lemma:Supersol:2}, we have for any $z_1, \, z_2$
\begin{equation}
	\label{Lemma:Supersolu:E2}\Upsilon''(z_1)\ge \frac{1}{20}\mathfrak{L}_{n}^{-2}\Upsilon(z_2) =: C_\Upsilon\Upsilon(z_2),
\end{equation}
and  
	\begin{equation}\begin{aligned}
			\label{Lemma:Supersolu:E3}A_\mu(t) \ = \  &  \int_{t}^{T_{\mathfrak{h}_n,\mathfrak{L}_{n}}}\mathrm{d}s\int_{0}^\infty\mathrm{d}z \mu(s,z),\\ \ \ \ B_\mu(t) \ = \ &  C_\Upsilon \int_{t}^{T_{\mathfrak{h}_n,\mathfrak{L}_{n}}}\mathrm{d}s\int_{0}^\infty\mathrm{d}z \mu(s,z)z^2/2.\end{aligned}
	\end{equation}
	Then by Lemma \ref{Lemma:Supersol},
	\begin{equation}\begin{aligned}
			\label{Lemma:Supersolu:E4}  &  \rho_{super}(t,x)\ : = \ e^{B_\mu}\Upsilon(A_\mu+x),\end{aligned}
	\end{equation}
	satisfies \eqref{Lemma:Growthlemmas:8}, which is a supersolution. It is straightforward that $\partial_x\rho_{super}(t,x)\le 0$ and $\partial_{xx}\rho_{super}(t,x)\ge 0$
	a.e. $(t,x)\in [0,\infty)^2$. 
	
	Next, we bound
		\begin{equation}\begin{aligned}\label{Lemma:Supersolu:E5}    
A_\mu(t) \ = \ 		&	 \int_{t}^{T_{\mathfrak{h}_n,\mathfrak{L}_{n}}}\mathrm{d}s\int_{0}^\infty\mathrm{d}z \mu(s,z)z \ \le \  \int_{0}^{T_{\mathfrak{h}_n,\mathfrak{L}_{n}}}\mathrm{d}s\int_{0}^\infty\mathrm{d}z \mu(s,z)z\\
		 \ \le \ & 27 \frac{{C_1}C_\omega^{\frac2\alpha} \mho(\mathfrak{L}_n/2)}{\mathfrak L_n^{\frac2\alpha}}\mathfrak{h}_n\int_{0}^{T_{\mathfrak{h}_n,\mathfrak{L}_{n}}}\mathrm{d}s\int_{0}^\infty\mathrm{d}z  \int_{\mathbb{R}_+ }\mathrm d\omega
		 \Big[ \mathcal G_{\mathfrak{h}_n,\mathfrak{L}_{n}}(s,\omega)\Big]\Big[\mathcal G_{\mathfrak{h}_n,\mathfrak{L}_{n}}(s,\omega+z)\Big]\\
		 \ \le \ &27 \frac{{C_1}C_\omega^{\frac2\alpha} \mho(\mathfrak{L}_n/2)}{\mathfrak L_n^{\frac2\alpha}}\mathfrak{h}_n \int_{0}^{T_{\mathfrak{h}_n,\mathfrak{L}_{n}}}\mathrm{d}s  	\left[\int_{\mathbb{R}_+ }\mathrm d\omega\Big[ \mathcal G_{\mathfrak{h}_n,\mathfrak{L}_{n}}(s,\omega)\Big]\right]
		 ^{2}   \ \le \   27  \frac{1}{\mathfrak L_n^{\frac2\alpha}}\mathfrak{h}_n\frac{\mathfrak{L}_n^{\varepsilon-\beta}}{4000}\\
		 \ \le \ &27 \frac{\mathfrak{L}_n^{\varepsilon-\frac2\alpha-\beta}\mathfrak{h}_n}{4000}  \ \le \   27 \frac{2^{-n(\varepsilon-\frac2\alpha-\beta)}}{4000}\frac{\mathfrak L_n}{2^{\mathfrak{M}_0(n)}}\ \le \   27 \frac{\mathfrak{L}_n}{4000},\end{aligned}
	\end{equation}
	where we have used \eqref{Lemma:Growthlemmas:9} and the definitions of $\mathfrak{h}_n,{\mathfrak{M}_0(n)}$.
	  We then bound
		\begin{equation}\begin{aligned}\label{Lemma:Supersolu:E6}    
		B_\mu(t) \ = \ 		&	 C_\Upsilon\int_{t}^{T_{\mathfrak{h}_n,\mathfrak{L}_{n}}}\mathrm{d}s\int_{0}^\infty\mathrm{d}z \mu(s,z)z^2/2 \ \le \  9C_\Upsilon\int_{0}^{T_{\mathfrak{h}_n,\mathfrak{L}_{n}}}\mathrm{d}s\int_{0}^\infty\mathrm{d}z \mu(s,z)z\mathfrak{h}_n/2\\
			\ \le \ & 9C_\Upsilon A_\mu(t) \mathfrak{h}_n/2  \ \le \    \frac{243}{160000} \mathfrak{L}_n^{-1}\mathfrak{h}_n \ \le \ \frac{243}{2^{\mathfrak{M}_0(n)}160000}   \ \le \  \frac{243}{160000}.\end{aligned}
	\end{equation}
 
Combining \eqref{Lemma:Supersolu:E5} and \eqref{Lemma:Supersolu:E6}, we find
	\begin{equation}\begin{aligned}\label{Lemma:Supersolu:E7} 
	\rho_{super}(t,x) \ = \ 		&	    e^{B_\mu}\Upsilon(A_\mu+x) \ \le  \ e^{\frac{243}{160000}}\left[\ \exp\left[ 	\frac{1}{\mathfrak{L}_{n}}\max\left\{\frac{\mathfrak{L}_{n}}{10}-x-A_\mu,0\right\} \right] -1\right]\\ \le  \ &  e^{\frac{243}{160000}}\Big[e^\frac{1}{10}-1\Big]\le 0.2.\end{aligned}
\end{equation}
	When $x\in\Big[0,\frac{\mathfrak{L}_n}{40}\Big]$, we find
		\begin{equation}\begin{aligned}\label{Lemma:Supersolu:E8} 
			\rho_{super}(t,x) \ = \ 		&	    e^{B_\mu}\Upsilon(A_\mu+x) \ \ge  \  \left[\ \exp\left[ \frac{1}{\mathfrak{L}_{n}}\max\left\{\frac{\mathfrak{L}_{n}}{10}-x-A_\mu,0\right\} \right] -1\right]\\ \ge  \ &   e^\frac{1}{20}-1\ \ge\ 0.05.\end{aligned}
	\end{equation}
	The proof of the Lemma is then completed.

\end{proof}
\begin{lemma}[Growth lemma]
	\label{Lemma:Growth1} Let $T_0$ be as in \eqref{T0}. By the definitions of \eqref{Sec:Growthlemmas:2},   \eqref{Sec:Growthlemmas:3},   \eqref{Sec:Growthlemmas:4},   \eqref{Sec:Growthlemmas:5},  \eqref{Sec:Growthlemmas:6}, there exists a universal constant $\mathfrak C_{1}>0$ such that,  
	\begin{equation}\label{Lemma:Growth1:1}
	\begin{aligned}
		{\mathfrak{C}_1 \mathfrak{L}_n^{(\beta+2)\Big[\varepsilon-\frac{2}{\alpha}\frac{\beta+1}{\beta+2}-\beta\frac{\beta+3}{\beta+2}-\frac{3 \varsigma}{\beta+2}\Big]}(	\mathscr M +\mathscr E)}  
		\ > \ 	  & | {	\mathscr{U}_{n}^T}|,
	\end{aligned}
\end{equation}
for   $n\in\mathbb{N},n>100$, $\forall T\in[0,T_0)$.
In addition,   suppose that there exist arbitrary constants $C_{ini},c_{ini}>0$ and a small constant $r_0>0$ such that
\begin{equation}\label{Lemma:Growth1:5}
	\begin{aligned}
		\int_{0}^r\mathrm d\omega F(0,\omega)     
		\geq  C_{ini}  {r^{c_{ini}}}, \end{aligned}
\end{equation}
for all $r_0>r>0$.
We also have the following bound on the sets $\mathfrak{S}_n^T$, $\mathfrak{S}_{n+1}^T$ and 	$\mathscr{V}_{n}^T$
	\begin{equation}\label{Lemma:Growth1:3}
	\begin{aligned}
		\Big|	\Big(\mathfrak{S}_n^T\setminus  \mathfrak{S}_{n+1}^T\Big)\setminus 
		\Big(\mathscr{V}_{n}^T\backslash \mathscr{W}_{n}^T\Big)\Big| \ < \ {\mathfrak{C}_1\mathfrak{L}_n^{(\beta+2)\Big[\varepsilon-\frac{2}{\alpha}\frac{\beta+1}{\beta+2}-\beta\frac{\beta+3}{\beta+2}-\frac{3 \varsigma}{\beta+2}\Big]}(	\mathscr M +\mathscr E)},
	\end{aligned}
\end{equation}
and there exists $\mathscr{M}_*$ such that $\forall n>\mathscr{M}_*, n\in\mathbb{N}$
	\begin{equation}\label{Lemma:Growth1:4}
	\begin{aligned}
		\Big| 
		\mathscr{V}_{n}^T\Big|  \ \le \ \mathfrak{C}_V{\mathfrak{L}_{n}^{\varepsilon-2\beta-2\varsigma}},
	\end{aligned}
\end{equation}
for some universal constant $\mathfrak{C}_V>0$, independent of $\varepsilon, \alpha,\beta,\varsigma$.

We define $\gamma(\varsigma)=\min\left\{ \varepsilon-2\beta-2\varsigma,(\beta+2)\Big[\varepsilon-\frac{2}{\alpha}\frac{\beta+1}{\beta+2}-\beta\frac{\beta+3}{\beta+2}-\frac{3 \varsigma}{\beta+2}\Big]\right\} $,   then there exists $\mathscr{M}^*>0$ such that
\begin{equation}\label{Lemma:Growth2:1}
	\left\vert\bigcup_{i=n}^\infty	\big(\mathfrak{S}_i^T\backslash
			\mathscr{W}_{i}^T\big) \right\vert   \leq{\mathfrak{C}_\mathfrak{S}}\left(
	\mathfrak{L}_n\right) ^{\gamma(\varsigma)},
\end{equation}
for a universal constant $\mathfrak{C}_\mathfrak{S}>0$,  independent of $\varepsilon, \alpha,\beta,\varsigma$,
 for   $n>\mathscr{M}^*, n\in\mathbb{N}$, $\forall T\in[0,T_0)$.
 
 There exist $\varsigma_*>0$ and $\mathscr{M}^*(\varsigma_*)>0$ such that for $n>\mathscr{M}^*(\varsigma_*)$, we have
 \begin{equation}\label{Lemma:Growth:3}
	\left\vert \mathfrak{S}_n^T \right\vert   \leq{\mathfrak{C}_\mathfrak{S}}\left(
	\mathfrak{L}_n\right) ^{\gamma(\varsigma_*)}.
\end{equation}
\end{lemma}

\begin{proof}
	
	{\it (i) Proof of \eqref{Lemma:Growth1:1}}.

We first observe that the same argument used in   the proof of Proposition \ref{Propo:Collision} can be reiterated, with $\nu= 1-\frac{1}{2^{\varsigma}}$, with $\varsigma$ as in \eqref{varsigma}, yielding the same inequality as  \eqref{Propo:Collision:34}
	\begin{equation}\label{Lemma:Growth1:E1}
	\begin{aligned}
		\frac{\mathfrak{C}_0\mathfrak{L}_n^{2+\frac2\alpha}(	\mathscr M +\mathscr E)}{\mathfrak h^{2}_n\mho(\mathfrak h_n)\Big( {\nu}\Big)^4}
		\ \ge\ &\int_{0}^T\mathrm{d}t \chi_{	\mathscr{U}_{n}^T}(t) \left( \int_{[0,\mathfrak{L}_n)}\mathrm{d}\omega F(t)\right)^3,
	\end{aligned}
\end{equation}
which, in combination with \eqref{Sec:Growthlemmas:2}, implies
	\begin{equation}\label{Lemma:Growth1:E2}
	\begin{aligned}
		\frac{\mathfrak{C}_1\mathfrak{L}_n^{2+\frac2\alpha}(	\mathscr M +\mathscr E)}{\mathfrak h^{2+\beta}_n}
		\ \ge\ &\int_{0}^T\mathrm{d}t \chi_{	\mathscr{U}_{n}^T}(t) \mathfrak L_{n } 
		^{3\varsigma} \ = \ \mathscr{C}_F| {	\mathscr{U}_{n}^T}| \mathfrak L_{n} 
		^{3\varsigma},
	\end{aligned}
\end{equation}
for a universal constant $\mathfrak{C}_1>0$. Since $$(\beta+2)\Big[\varepsilon-\frac{2}{\alpha}\frac{\beta+1}{\beta+2}-\beta\frac{\beta+3}{\beta+2}-\frac{3 \varsigma}{\beta+2}\Big]>0$$ then $$\frac2\alpha-\beta-(2+\beta)(\frac2\alpha+\beta-\varepsilon)>3 \varsigma.$$
Therefore,
	\begin{equation}\label{Lemma:Growth1:E2}
	\begin{aligned}
\mathfrak{C}_1\mathfrak{L}_n^{(\beta+2)\Big[\varepsilon-\frac{2}{\alpha}\frac{\beta+1}{\beta+2}-\beta\frac{\beta+3}{\beta+2}-\frac{3 \varsigma}{\beta+2}\Big]}\ \ge\ 2^{(2+\beta)\mathfrak{M}_0}{\mathfrak{C}_1\mathfrak{L}_n^{\frac2\alpha-\beta-3\varsigma}(	\mathscr M +\mathscr E)}			\ \ge\ & | {	\mathscr{U}_{n}^T}|,
	\end{aligned}
\end{equation}
leading to the proof of \eqref{Lemma:Growth1:1}.

	{\it (ii) Proof of \eqref{Lemma:Growth1:3}}.

Next, we will prove that 
	\begin{equation}\label{Lemma:Growth1:2}
	\begin{aligned}
		\Big(\mathfrak{S}_n^T\setminus  \mathfrak{S}_{n,*}^T\Big)\bigcap 
		\mathscr{W}_{n}^T	\ = \ 	  \emptyset,
	\end{aligned}
\end{equation} whose proof is quite straightforward. If $t\in \mathscr{W}_{n}^T$, then there exists $i\in \{0,\cdots,2^{\mathfrak{M}_0-1}-2\}$ such that 
	\begin{equation*}
	\int_{
		\Xi_{i}^{\mathfrak{h}_n,\mathfrak{L}_{n}}} \mathrm{d}\omega F\left( t,\omega\right) \geq \mathscr{C}_F \mathfrak L_{n+1} 
	^\varsigma,
	\end{equation*}
yielding 
	\begin{equation}\label{Lemma:Growth1:E3}
 \int_{\left[ 0,\mathfrak{L}_{n+1}\right)
	}\mathrm{d}\omega F\left( t,\omega\right) \geq \mathscr{C}_F   \mathfrak L_{n+1} 
	^\varsigma,
\end{equation}
since $\Xi_{i}^{\mathfrak{h}_{n},\mathfrak{L}_{n}}\subset \left[ 0,\mathfrak{L}_{n+1}\right)$. 
	On the other hand, when $t\in	\Big(\mathfrak{S}_n^T\backslash  \mathfrak{S}_{n,*}^T\Big) $ we
	have:%
	\begin{equation}\label{Lemma:Growth1:E4}
		\int_{\left[ 0,\mathfrak{L}_{n}\right)
		}\mathrm{d}\omega F\left( t,\omega\right)\geq   \mathscr{C}_F \mathfrak L_{n} 
		^\varsigma\ \ ,\ \ \int_{\left[ 0,\mathfrak{L}_{n+1}\right)
		}\mathrm{d}\omega F\left( t,\omega\right)  <   \mathscr{C}_F\mathfrak L_{n+1} 
		^\varsigma.
	\end{equation}
Combining \eqref{Lemma:Growth1:E3}-\eqref{Lemma:Growth1:E4}, we obtain \eqref{Lemma:Growth1:2}. 

 From \eqref{Lemma:Growth1:2}, we deduce
 
 	\begin{equation}\label{Lemma:Growth1:E5}
 	\begin{aligned}
 		\Big(\mathfrak{S}_n^T\setminus  \mathfrak{S}_{n,*}^T\Big)\setminus 
 		\Big(\mathscr{V}_{n}^T\backslash \mathscr{W}_{n}^T\Big)	\ = \ 	  	\Big(\mathfrak{S}_n^T\setminus  \mathfrak{S}_{n,*}^T\Big)\setminus 
 		\bigcup_{i=0}^{2^\mathfrak{M}_0-1}\mathfrak{S}^T_{n,i} \ \subset  \mathfrak{S}_n^T  \setminus 
 		\bigcup_{i=0}^{2^\mathfrak{M}_0-1}\mathfrak{S}^T_{n,i} \ = \ \mathscr{U}_{n}^T,
 	\end{aligned}
 \end{equation}
 which, in comparison with \eqref{Lemma:Growth1:E2}, implies
 	\begin{equation}\label{Lemma:Growth1:E6}
 	\begin{aligned}
 	\Big|	\Big(\mathfrak{S}_n^T\setminus  \mathfrak{S}_{n,*}^T\Big)\setminus 
 		\Big(\mathscr{V}_{n}^T\backslash \mathscr{W}_{n}^T\Big)\Big| \ \le \ & | {	\mathscr{U}_{n}^T}|   		   \ \le \ &  \mathfrak{C}_1\mathfrak{L}_n^{(\beta+2)\Big[\varepsilon-\frac{2}{\alpha}\frac{\beta+1}{\beta+2}-\beta\frac{\beta+3}{\beta+2}-\frac{3 \varsigma}{\beta+2}\Big]},
	\end{aligned}
 \end{equation}
leading to \eqref{Lemma:Growth1:3}. 

	{\it (iii) Proof of \eqref{Lemma:Growth1:4}}.

We will now prove \eqref{Lemma:Growth1:4}. To this end, we first consider the case that  there exists $\mathscr{M}_*$ such that $\forall n>\mathscr{M}_*$	\begin{equation}	\label{Lemma:Growth1:E7}
	\int_{0}^{{T}} \mathrm d t  \left[\int_{\mathbb{R}_+ }\mathrm d\omega
	\Big[ F(t,\omega)\chi_{\omega\in \Xi_{\mathfrak{i}\left( t\right)}^{\mathfrak{h}_n,\mathfrak{L}_{n}}}\chi_{\mathscr{V}_{n}^{T}}(t)\Big]\right]
	^{2}  \le \Big[{{C_1} C_\omega^{\frac2\alpha} }\Big]^{-1}\frac{\mathfrak{L}_{n}^{\varepsilon-\beta}}{4000\mho(\mathfrak{L}_n/2)},
\end{equation}
which, in combination with \eqref{Sec:Growthlemmas:3}, \eqref{Sec:Growthlemmas:5}, implies
\begin{equation}	\label{Lemma:Growth1:E7:aa}
	\int_{0}^{{T}} \mathrm d t  \chi_{\mathscr{V}_{n}^T}\Big(\Big| \mathfrak L_{n+1} 
	^{\varsigma}\Big|\Big)	^{2}  \le C' {\mathfrak{L}_{n}^{\varepsilon-2\beta}},
\end{equation}
for some constant $C'>0$, then
	\begin{equation}	\label{Lemma:Growth1:E8}
		 \mathscr{C}_F\Big| 
	\mathscr{V}_{n}^T\Big| \mathfrak L_{n+1} 
	^{2\varsigma}\ \le\ {\mathfrak{L}_{n}^{\varepsilon-2\beta}}C',
\end{equation}
yielding
	\begin{equation}\label{Lemma:Growth1:E9}
	\begin{aligned}
	\Big| 
	\mathscr{V}_{n}^T\Big| \ \le \ C' {\mathfrak{L}_{n}^{\varepsilon-2\beta-2\varsigma}}2^{2\varsigma} \mathscr{C}_F^{-1}.
	\end{aligned}
\end{equation}
We then obtain \eqref{Lemma:Growth1:4}.   We now consider the opposite case of \eqref{Lemma:Growth1:E7}, that there exists an infinite set $\mathcal{S}_*=\{n_1,n_2,\cdots\}$ and $\forall n_j\in \mathcal{S}_*$
		\begin{equation}	\label{Lemma:Growth1:E10}
		\int_{0}^{{T}} \mathrm d t  \left[\int_{\mathbb{R}_+ }\mathrm d\omega
		\Big[ F(t,\omega)\chi_{\omega\in \Xi_{\mathfrak{i}\left( t\right)}^{\mathfrak{h}_{n_j},\mathfrak{L}_{n_j}}}\chi_{\mathscr{V}_{n_j}^{T}}(t)\Big]\right]
		^{2}  > \Big[{{C_1}\ C_\omega^{\frac2\alpha} }\Big]^{-1}\frac{\mathfrak{L}_{n_j}^{\varepsilon-\beta}}{4000\mho(\mathfrak{L}_{n_j}/2)}.
	\end{equation}
There must  exist a time  $\bar T_{n_j}=T_{\mathfrak{h}_{n_j},\mathfrak{L}_{{n_j}}}\in\left[ 0, T
	\right]$ such that: 
	\begin{equation}	\label{Lemma:Growth1:E11}
	\int_{0}^{{\bar T}_{n_j}} \mathrm d t  \left[\int_{\mathbb{R}_+ }\mathrm d\omega
	\Big[ \mathcal G_{\mathfrak{h}_{n_j},\mathfrak{L}_{{n_j}}}(t,\omega)\Big]\right]
	^{2}  = \Big[{{C_1}C_\omega^{\frac2\alpha} }\Big]^{-1}\frac{\mathfrak{L}_{{n_j}}^{\varepsilon-\beta}}{4000\mho(\mathfrak{L}_{n_j}/2)}.
\end{equation}
	Using $\rho_{super}(t,\omega)$ constructed in Lemma  \ref{Lemma:SuperSolu} as a test function, with $\bar T_{n_j}=T_{\mathfrak{h}_{n_j},\mathfrak{L}_{{n_j}}}$, we obtain: 	\begin{equation}	\label{Lemma:Growth1:E12}
	\begin{aligned}
	&	\partial_{t}\left( \int_{\mathbb{R}^{+}}\mathrm{d}\omega F\rho_{super}\right)\  =\
		\int_{\mathbb{R}^{+}}\mathrm{d}\omega F\partial_{t}\rho_{super} \\
		 & + \  C_1\iiint_{\mathbb{R}_+^{3}}\mathrm{d}\omega_1\,\mathrm{d}\omega_2\,\mathrm{d}\omega f_1f_2f\mathbf{1}_{\omega+\omega_1-\omega_2\ge 0}\\
		&\times 	 \Big\{ [-\rho_{super}(\omega_{Max})-\rho_{super}(\omega_{Min})+\rho_{super}(\omega_{Mid})+\rho_{super}(\omega_{Max}+\omega_{Min}-\omega_{Mid})]\\
		&\ \ \ \ \times\mho(\omega_{Max})\mho(\omega_{Min})\mho(\omega_{Mid})\mho(\omega_{Max}+\omega_{Min}-\omega_{Mid})|k_{Min}|  \\
		&+[-\rho_{super}(\omega_{Max})-\rho_{super}(\omega_{Mid})+\rho_{super}(\omega_{Min})+\rho_{super}(\omega_{Max}+\omega_{Mid}-\omega_{Min})]\\
		&\ \ \ \ \times\mho(\omega_{Max})\mho(\omega_{Min})\mho(\omega_{Mid})\mho(\omega_{Max}+\omega_{Mid}-\omega_{Min})|k_{Min}|  \\
		&+[-\rho_{super}(\omega_{Min})-\rho_{super}(\omega_{Mid})+\rho_{super}(\omega_{Max})+\rho_{super}(\omega_{Min}+\omega_{Mid}-\omega_{Max})]\\
		&\ \ \ \ \times\Xi(\omega_{Max},\omega_{Min},\omega_{Mid},\omega_{Min}+\omega_{Mid}-\omega_{Max})\Big\},
	\end{aligned}\end{equation}
where we have used \eqref{Lemma:Concave:E5}. Next, we will prove that 
	\begin{equation}	\label{Lemma:Growth1:E12:a}
	\begin{aligned}
		& [-\rho_{super}(\omega_{Max})-\rho_{super}(\omega_{Min})+\rho_{super}(\omega_{Mid})+\rho_{super}(\omega_{Max}+\omega_{Min}-\omega_{Mid})]\\
		&\ \ \ \ \times\mho(\omega_{Max})\mho(\omega_{Min})\mho(\omega_{Mid})\mho(\omega_{Max}+\omega_{Min}-\omega_{Mid})|k_{Min}|  \\
		&+[-\rho_{super}(\omega_{Max})-\rho_{super}(\omega_{Mid})+\rho_{super}(\omega_{Min})+\rho_{super}(\omega_{Max}+\omega_{Mid}-\omega_{Min})]\\
		&\ \ \ \ \times\mho(\omega_{Max})\mho(\omega_{Min})\mho(\omega_{Mid})\mho(\omega_{Max}+\omega_{Mid}-\omega_{Min})|k_{Min}|  \\
		&+[-\rho_{super}(\omega_{Min})-\rho_{super}(\omega_{Mid})+\rho_{super}(\omega_{Max})+\rho_{super}(\omega_{Min}+\omega_{Mid}-\omega_{Max})]\\
		&\ \ \ \ \times\Xi(\omega_{Max},\omega_{Min},\omega_{Mid},\omega_{Min}+\omega_{Mid}-\omega_{Max}) \ \ge \ 0.
\end{aligned}\end{equation}

To prove \eqref{Lemma:Growth1:E12:a},   we note  that \begin{equation}\label{Lemma:Growth1:E13}-\rho_{super}(\omega_{Min})-\rho_{super}(\omega_{Mid})+\rho_{super}(\omega_{Max})+\rho_{super}(\omega_{Min}+\omega_{Mid}-\omega_{Max})\ge 0,\end{equation}
which can be easily proved using the convexity of $\rho_{super}$ (see \eqref{Lemma:Concave:E5a}).

	It follows from Lemma \ref{Lemma:SuperSolu}(c) that
\begin{equation}
	\rho_{super}(\omega_{Max}+\omega_{Min}-\omega_{Mid}) \ +\ \rho_{super}(\omega_{Max}+\omega_{Mid}-\omega_{Min}) 	 \geq 2 \rho_{super}(\omega_{Max}),  
\end{equation}
which leads to
\begin{equation}\begin{aligned}
		&	-\rho_{super}(\omega_{Max})-\rho_{super}(\omega_{Min})+\rho_{super}(\omega_{Mid})+\rho_{super}(\omega_{Max}+\omega_{Min}-\omega_{Mid})\\
		\geq \ & -[-\rho_{super}(\omega_{Max})-\rho_{super}(\omega_{Mid})+\rho_{super}(\omega_{Min})+\rho_{super}(\omega_{Max}+\omega_{Mid}-\omega_{Min})].  \end{aligned}
\end{equation}
Next, we multiply by $\mho(\omega_{Max}+\omega_{Min}-\omega_{Mid})$  both sides of the inequality and obtain \begin{equation}\begin{aligned}
		&	[-\rho_{super}(\omega_{Max})-\rho_{super}(\omega_{Min})+\rho_{super}(\omega_{Mid})+\rho_{super}(\omega_{Max}+\omega_{Min}-\omega_{Mid})]\\
		&\times\mho(\omega_{Max}+\omega_{Min}-\omega_{Mid})\\
		\geq \ & -[-\rho_{super}(\omega_{Max})-\rho_{super}(\omega_{Mid})+\rho_{super}(\omega_{Min})+\rho_{super}(\omega_{Max}+\omega_{Mid}-\omega_{Min})]\\
		&\times\mho(\omega_{Max}+\omega_{Min}-\omega_{Mid}),  \end{aligned}
\end{equation}
which implies 
\begin{equation}\label{Lemma:Growth1:E14a}\begin{aligned}
		&	[-\rho_{super}(\omega_{Max})-\rho_{super}(\omega_{Min})+\rho_{super}(\omega_{Mid})+\rho_{super}(\omega_{Max}+\omega_{Min}-\omega_{Mid})]\\
		&\times\mho(\omega_{Max}+\omega_{Min}-\omega_{Mid})
		\\
		& +  [-\rho_{super}(\omega_{Max})-\rho_{super}(\omega_{Mid})+\rho_{super}(\omega_{Min})+\rho_{super}(\omega_{Max}+\omega_{Mid}-\omega_{Min})]\\
		&\times\mho(\omega_{Max}+\omega_{Mid}-\omega_{Min})\\
		\geq \ & [-\rho_{super}(\omega_{Max})-\rho_{super}(\omega_{Mid})+\rho_{super}(\omega_{Min})+\rho_{super}(\omega_{Max}+\omega_{Mid}-\omega_{Min})]\\
		&\times[\mho(\omega_{Max}+\omega_{Mid}-\omega_{Min})-\mho(\omega_{Max}+\omega_{Min}-\omega_{Mid})]\\
		= \ & \Big[\int_{0}^{\omega_{Mid}-\omega_{Min}}\mathrm d\xi_1\rho'_{super}(\omega_{Max}+\xi_1)-\int_{0}^{\omega_{Mid}-\omega_{Min}}\mathrm d\xi_1\rho'_{super}(\omega_{Min}+\xi_1)\Big]\\
		&\times[\mho(\omega_{Max}+\omega_{Mid}-\omega_{Min})-\mho(\omega_{Max}+\omega_{Min}-\omega_{Mid})]\ \ge\ 0. \end{aligned}
\end{equation}
Therefore \eqref{Lemma:Growth1:E12:a} is proved.


Restricting to the case $\omega_2\le \omega_1\le\omega$, we deduce from \eqref{Lemma:Growth1:E12:a}  

 that
\begin{equation}	\label{Lemma:Growth1:E15}
	\begin{aligned}
		&	\partial_{t}\left( \int_{\mathbb{R}^{+}}\mathrm{d}\omega F\rho_{super}\right)\\
		\  \ge \ &
		\int_{\mathbb{R}^{+}}\mathrm{d}\omega F\partial_{t}\rho_{super} \
	 + \  C_1\iiint_{\mathbb{R}_+^{3}}\mathrm{d}\omega_1\,\mathrm{d}\omega_2\,\mathrm{d}\omega f_1f_2f\mathbf{1}_{\omega_2\le \omega_1\le\omega}\\
		&\times 	 \Big\{[-\rho_{super}(\omega_2)-\rho_{super}(\omega_1)+\rho_{super}(\omega)+\rho_{super}(\omega_2+\omega_1-\omega)]\\
		&\times\mho(\omega)\mho(\omega_{2})\mho(\omega_{1})\mho(\omega_2+\omega_1-\omega)\min\{|k_0|,|k|\}\\ &+[-\rho_{super}(\omega)-\rho_{super}(\omega_{2})+\rho_{super}(\omega_{1})+\rho_{super}(\omega+\omega_{2}-\omega_{1})]\\
		&\ \ \ \ \times\mho(\omega)\mho(\omega_{2})\mho(\omega_{1})\mho(\omega+\omega_{2}-\omega_{1})|k_{2}|  \\
		&+[-\rho_{super}(\omega)-\rho_{super}(\omega_{1})+\rho_{super}(\omega_{2})+\rho_{super}(\omega+\omega_{1}-\omega_{2})]\\
		&\ \ \ \ \times\mho(\omega)\mho(\omega_{2})\mho(\omega_{1})\mho(\omega+\omega_{1}-\omega_{2})|k_{2}|\Big\}\\
			\  \ge \	&
		\int_{\mathbb{R}^{+}}\mathrm{d}\omega F\partial_{t}\rho_{super} \
		+ \  C_1\iiint_{\mathbb{R}_+^{3}}\mathrm{d}\omega_1\,\mathrm{d}\omega_2\,\mathrm{d}\omega \frac{F(\omega)}{|k|\mho(\omega)}\frac{F(\omega_2)}{|k_2|\mho(\omega_2)}\frac{F(\omega_1)}{|k_1|\mho(\omega_1)}\mathbf{1}_{\omega_2\le \omega_1\le\omega}\\
			&\times 	 \Big\{[-\rho_{super}(\omega_2)-\rho_{super}(\omega_1)+\rho_{super}(\omega)+\rho_{super}(\omega_2+\omega_1-\omega)]\\
	&\times\mho(\omega)\mho(\omega_{2})\mho(\omega_{1})\mho(\omega_2+\omega_1-\omega)\min\{|k_0|,|k|\}\\
		& +[-\rho_{super}(\omega)-\rho_{super}(\omega_{2})+\rho_{super}(\omega_{1})+\rho_{super}(\omega+\omega_{2}-\omega_{1})]\\
		&\ \ \ \ \times\mho(\omega)\mho(\omega_{2})\mho(\omega_{1})\mho(\omega+\omega_{2}-\omega_{1})|k_{2}|  \\
		&+[-\rho_{super}(\omega)-\rho_{super}(\omega_{1})+\rho_{super}(\omega_{2})+\rho_{super}(\omega+\omega_{1}-\omega_{2})]\\
		&\ \ \ \ \times\mho(\omega)\mho(\omega_{2})\mho(\omega_{1})\mho(\omega+\omega_{1}-\omega_{2})|k_{2}|\Big\},
\end{aligned}\end{equation}
which can be simplified as
\begin{equation}	\label{Lemma:Growth1:E16}
	\begin{aligned}
		&	\partial_{t}\left( \int_{\mathbb{R}^{+}}\mathrm{d}\omega F\rho_{super}\right)\\
		\  \ge \	&
		\int_{\mathbb{R}^{+}}\mathrm{d}\omega F\partial_{t}\rho_{super} \
		+ \  C_1\iiint_{\mathbb{R}_+^{3}}\mathrm{d}\omega_1\,\mathrm{d}\omega_2\,\mathrm{d}\omega \frac{FF_1F_2}{|k||k_1|}\mathbf{1}_{\omega_2\le \omega_1\le\omega}\\
			&\times 	 \Big\{[-\rho_{super}(\omega_2)-\rho_{super}(\omega_1)+\rho_{super}(\omega)+\rho_{super}(\omega_2+\omega_1-\omega)]
		\\	&\times \mho(\omega_2+\omega_1-\omega)\min\{|k_0|,|k|\}/|k_2|\\
		& +[-\rho_{super}(\omega)-\rho_{super}(\omega_{2})+\rho_{super}(\omega_{1})+\rho_{super}(\omega+\omega_{2}-\omega_{1})] \mho(\omega+\omega_{2}-\omega_{1})   \\
		&+[-\rho_{super}(\omega)-\rho_{super}(\omega_{1})+\rho_{super}(\omega_{2})+\rho_{super}(\omega+\omega_{1}-\omega_{2})] \mho(\omega+\omega_{1}-\omega_{2})\Big\},
\end{aligned}\end{equation}
where $F=F(\omega), F_1=F(\omega_1), F_2=F(\omega_2)$. 
Using \eqref{Lemma:Growth1:E13}, we bound
\begin{equation}	\label{Lemma:Growth1:E17}
	\begin{aligned}
		&	\partial_{t}\left( \int_{\mathbb{R}^{+}}\mathrm{d}\omega F\rho_{super}\right)\\
		\  \ge \	&
		\int_{\mathbb{R}^{+}}\mathrm{d}\omega F\partial_{t}\rho_{super} \
		+ \  C_1C_\omega^{\frac2\alpha}\iiint_{\mathbb{R}_+^{3}}\mathrm{d}\omega_1\,\mathrm{d}\omega_2\,\mathrm{d}\omega \frac{FF_1F_2}{\mathfrak L_{n_j}^{\frac2\alpha}}\chi_{\left\{ \omega_{2}\leq\frac{\mathfrak L_{{n_j}} }{10}\right\}}  \chi_{\left\{\omega_1\le\omega;\omega_1,\omega \in \Xi_{\mathfrak{i}\left( t\right)}^{\mathfrak{h}_{n_j},\mathfrak{L}_{{n_j}}}\right\}}\chi_{\mathscr{V}_{{n_j}}^{\bar T_{n_j}}}(t) \\
		&\times 	 \Big\{ [-\rho_{super}(\omega)-\rho_{super}(\omega_{2})+\rho_{super}(\omega_{1})+\rho_{super}(\omega+\omega_{2}-\omega_{1})] \mho(\omega+\omega_{2}-\omega_{1})   \\
		&+[-\rho_{super}(\omega)-\rho_{super}(\omega_{1})+\rho_{super}(\omega_{2})+\rho_{super}(\omega+\omega_{1}-\omega_{2})] \mho(\omega+\omega_{1}-\omega_{2})\Big\}.
\end{aligned}\end{equation} 
By \eqref{Sec:Growthlemmas:5}, then for $\omega_1,\omega \in \Xi_{\mathfrak{i}\left( t\right)}^{\mathfrak{h}_{n_j},\mathfrak{L}_{{n_j}}}$, we have $\omega_1,\omega>\frac{\mathfrak L_{{n_j}} }{10}$. Therefore $\rho_{super}(\omega)=\rho_{super}(\omega_1)=\rho_{super}(\omega+\omega_{1}-\omega_{2})=0$, we deduce 
\begin{equation}	\label{Lemma:Growth1:E18}
	\begin{aligned}
		&	\partial_{t}\left( \int_{\mathbb{R}^{+}}\mathrm{d}\omega F\rho_{super}\right)\\
		\  \ge \	&
		\int_{\mathbb{R}^{+}}\mathrm{d}\omega F\partial_{t}\rho_{super} \
		+ \  C_1C_\omega^{\frac2\alpha}\iiint_{\mathbb{R}_+^{3}}\mathrm{d}\omega_1\,\mathrm{d}\omega_2\,\mathrm{d}\omega \frac{FF_1F_2}{\mathfrak L_{n_j}^{\frac2\alpha}}\chi_{\left\{ \omega_{2}\leq\frac{\mathfrak L_{{n_j}} }{10}\right\}}  \chi_{\left\{\omega_1\le\omega;\omega_1,\omega \in \Xi_{\mathfrak{i}\left( t\right)}^{\mathfrak{h}_{n_j},\mathfrak{L}_{n_j}}\right\}}\chi_{\mathscr{V}_{{n_j}}^{\bar T_{n_j}}}(t) \\
		&\times 	 \Big\{ [-\rho_{super}(\omega_{2})+\rho_{super}(\omega+\omega_{2}-\omega_{1})] \mho(\omega+\omega_{2}-\omega_{1})  \ + \ \rho_{super}(\omega_{2})  \mho(\omega+\omega_{1}-\omega_{2})\Big\}.
\end{aligned}\end{equation}

Observing that $\omega+\omega_{1}-\omega_{2}\ge \mathfrak{L}_{n_j}/2$ and $0\le\omega+\omega_{2}-\omega_{1}\le \omega_2+3\mathfrak{h}_{n_j}\le \mathfrak{L}_{n_j}/10+3\mathfrak{h}_{n_j}\le \mathfrak{L}_{n_j}/2$, we bound
\begin{equation}	\label{Lemma:Growth1:E18:1}
	\begin{aligned}& [-\rho_{super}(\omega_{2})+\rho_{super}(\omega+\omega_{2}-\omega_{1})] \mho(\omega+\omega_{2}-\omega_{1})  \ + \ \rho_{super}(\omega_{2})  \mho(\omega+\omega_{1}-\omega_{2})\\
			\ge\ &	[-\rho_{super}(\omega_{2})+\rho_{super}(\omega+\omega_{2}-\omega_{1})] \mho(\mathfrak{L}_{n_j}/2) \ + \ \rho_{super}(\omega_{2})  \mho(\mathfrak{L}_{n_j}/2),
\end{aligned}\end{equation}
yielding
\begin{equation}	\label{Lemma:Growth1:E18:2}
	\begin{aligned}
		&	\partial_{t}\left( \int_{\mathbb{R}^{+}}\mathrm{d}\omega F\rho_{super}\right) \\
		\  \ge \	&
		\int_{\mathbb{R}^{+}}\mathrm{d}\omega F\partial_{t}\rho_{super} \
		+ \  C_1C_\omega^{\frac2\alpha}\iiint_{\mathbb{R}_+^{3}}\mathrm{d}\omega_1\,\mathrm{d}\omega_2\,\mathrm{d}\omega \frac{FF_1F_2}{\mathfrak L_{n_j}^{\frac2\alpha}} \chi_{\left\{\omega_1\le\omega;\omega_1,\omega \in \Xi_{\mathfrak{i}\left( t\right)}^{\mathfrak{h}_{n_j},\mathfrak{L}_{{n_j}}}\right\}}\chi_{\mathscr{V}_{{n_j}}^{\bar T_{n_j}}}(t) \\
		&\times 	\chi_{\left\{ \omega_{2}\leq\frac{\mathfrak L_{{n_j}} }{10}\right\}}    \Big\{ [-\rho_{super}(\omega_{2})+\rho_{super}(\omega+\omega_{2}-\omega_{1})] \mho(\mathfrak{L}_{n_j}/2) \ + \rho_{super}(\omega_{2}) \mho(\mathfrak{L}_{n_j}/2)\Big\}\\
		\  \ge \	&
		\int_{\mathbb{R}^{+}}\mathrm{d}\omega F\partial_{t}\rho_{super} \
		+ \  C_1C_\omega^{\frac2\alpha}\iiint_{\mathbb{R}_+^{3}}\mathrm{d}\omega_1\,\mathrm{d}\omega_2\,\mathrm{d}\omega \frac{FF_1F_2\mho(\mathfrak{L}_{n_j}/2)}{\mathfrak L_{n_j}^{\frac2\alpha}}\chi_{\left\{ \omega_{2}\leq\frac{\mathfrak L_{{n_j}} }{10}\right\}} \\
		&\times 	   \chi_{\left\{\omega_1\le\omega;\omega_1,\omega \in \Xi_{\mathfrak{i}\left( t\right)}^{\mathfrak{h}_{n_j},\mathfrak{L}_{{n_j}}}\right\}}\chi_{\mathscr{V}_{{n_j}}^{\bar T_{n_j}}}(t)  \Big\{ [-\rho_{super}(\omega_{2})+\rho_{super}(\omega+\omega_{2}-\omega_{1})]   \Big\}\\
		& +\   C_1C_\omega^{\frac2\alpha}\iiint_{\mathbb{R}_+^{3}}\mathrm{d}\omega_1\,\mathrm{d}\omega_2\,\mathrm{d}\omega \frac{FF_1F_2}{\mathfrak L_{n_j}^{\frac2\alpha}}\chi_{\left\{ \omega_{2}\leq\frac{\mathfrak L_{{n_j}} }{10}\right\}}  \chi_{\left\{\omega_1\le\omega;\omega_1,\omega \in \Xi_{\mathfrak{i}\left( t\right)}^{\mathfrak{h}_{n_j},\mathfrak{L}_{{n_j}}}\right\}}\chi_{\mathscr{V}_{{n_j}}^{\bar T_{n_j}}}(t)  \rho_{super}(\omega_{2}) \mho(\mathfrak{L}_{n_j}/2).
\end{aligned}\end{equation}

Using the definition of $\rho_{super}$ as a supersolution, we obtain
\begin{equation}	\label{Lemma:Growth1:E20}
	\begin{aligned}
		&	\partial_{t}\left( \int_{\mathbb{R}^{+}}\mathrm{d}\omega F\rho_{super}\right)\\
		\  \ge \	&   \frac{\mathfrak C_2\mho(\mathfrak{L}_{n_j}/2)}{\mathfrak L_{n_j}^{\frac2\alpha}}\int_{\mathbb{R}_+}\mathrm{d}\omega_2\,F_2\rho_{super}(\omega_{2}) \iint_{\mathbb{R}_+^{2}}\mathrm{d}\omega_1\,\mathrm{d}\omega {FF_1} \chi_{\left\{\omega_1\le\omega;\omega_1,\omega \in \Xi_{\mathfrak{i}\left( t\right)}^{\mathfrak{h}_{n_j},\mathfrak{L}_{{n_j}}}\right\}}\chi_{\mathscr{V}_{n_j}^{\bar T_{n_j}}}(t)  ,
\end{aligned}\end{equation}
for a universal constant $\mathfrak C_2>0$ that varies from line to line. Therefore
\begin{equation}	\label{Lemma:Growth1:E21}
	\begin{aligned}
		&	\partial_{t}\left( \int_{\mathbb{R}^{+}}\mathrm{d}\omega F\rho_{super}\right)\\
		\  \ge \	&   \frac{\mathfrak C_2\mho(\mathfrak{L}_{n_j}/2)}{2\mathfrak L_{n_j}^{\frac2\alpha}}\left[\int_{\mathbb{R}_+}\mathrm{d}\omega\,F\rho_{super}(\omega)\right]\left[ \int_{\mathbb{R}_+}\mathrm{d}\omega {F} \chi_{\left\{\omega \in \Xi_{\mathfrak{i}\left( t\right)}^{\mathfrak{h}_{n_j},\mathfrak{L}_{{n_j}}}\right\}}\chi_{\mathscr{V}_{{n_j}}^{\bar T_{n_j}}}(t)\right]^2.
\end{aligned}\end{equation}
	
Let us recall that  
	$\rho_{super}(0,\omega)\ge 0.05$  for a.e.  $\omega\in\Big[0,\frac{\mathfrak{L}_{n_j}}{40}\Big]$, we bound
	\begin{equation}\label{Lemma:Growth1:E22}
		\begin{aligned}
		\int_{\mathbb{R}_+}\mathrm d\omega F(0,\omega)  \rho_{super}(0,\omega)   
	  \geq  0.05C_{ini}  {\Big(\min\Big\{\frac{\mathfrak L_{n_j}}{40},r_0/2\Big\}\Big)^{c_{ini}}}. \end{aligned}
	\end{equation}
Solving   \eqref{Lemma:Growth1:E21}, we   obtain 
	\begin{equation}\label{Lemma:Growth1:E23}
	\begin{aligned}
		& \int_{\mathbb{R}_+}\mathrm d\omega F(\bar T_{n_j},\omega)  \rho_{super}(\bar T_{n_j},\omega)\\ \ \geq \  &  0.05C_{ini}  {\Big(\min\Big\{\frac{\mathfrak L_{n_j}}{40},r_0/2\Big\}\Big)^{c_{ini} }} \exp\left\{  \frac{\mathfrak C_2\mho(\mathfrak{L}_{n_j}/2)}{2\mathfrak L_{n_j}^{\frac2\alpha}}\int_{0}^{\bar T_{n_j}}\mathrm d t \left[ \int_{\mathbb{R}_+}\mathrm{d}\omega {F} \chi_{\left\{\omega \in \Xi_{\mathfrak{i}\left( t\right)}^{\mathfrak{h}_{n_j},\mathfrak{L}_{n_j}}\right\}}\chi_{\mathscr{V}_{{n_j}}^{\bar T_{n_j}}}(t)\right]^2\right\}\\
		\ \geq \ &  0.05C_{ini}  {\Big(\min\Big\{\frac{\mathfrak L_{n_j}}{40},r_0/2\Big\}\Big)^{c_{ini} }}  \exp\left\{\frac{\mathfrak C_2\mho(\mathfrak{L}_{n_j}/2)}{2\mathfrak L_{n_j}^{\frac2\alpha}}\int_{0}^{{\bar T_{n_j}}} \mathrm d t  \left[\int_{\mathbb{R}_+ }\mathrm d\omega
		\Big[   \mathcal G_{\mathfrak{h}_{n_j},\mathfrak{L}_{{n_j}}}(t,\omega)\Big]\right]
		^{2}\right\}.
\end{aligned}
\end{equation}
Plugging \eqref{Lemma:Growth1:E11} into \eqref{Lemma:Growth1:E23}, we get 
	\begin{equation}\label{Lemma:Growth1:E24}
	\begin{aligned}
	&	\int_{\mathbb{R}_+}\mathrm d\omega F(\bar T_{n_j},\omega)  \rho_{super}(\bar T_{n_j},\omega)\\ \ \geq \ &    0.05C_{ini}  {\Big(\min\Big\{\frac{\mathfrak L_{n_j}}{40},r_0/2\Big\}\Big)^{c_{ini} }}  \exp\left\{\frac{\mathfrak C_2\mho(\mathfrak{L}_{n_j}/2)}{2\mathfrak L_{n_j}^{\frac2\alpha}}\frac{\mathfrak{L}_{{n_j}}^{\varepsilon-\beta}}{4000\mho(\mathfrak{L}_{n_j}/2)}\Big[{{C_1}C_\omega^{\frac2\alpha} }\Big]^{-1}\right\}\\
		 \ \geq \ &    0.05C_{ini}  {\Big(\min\Big\{\frac{\mathfrak L_{n_j}}{40},r_0/2\Big\}\Big)^{c_{ini} }} \exp\left\{\frac{\mathfrak C_2\mathfrak{L}_{{n_j}}^{\Big(\varepsilon-\frac2\alpha-\beta\Big)}}{8000}\Big[{{C_1}C_\omega^{\frac2\alpha} }\Big]^{-1}\right\},
	\end{aligned}
\end{equation}
 yielding 
	 	\begin{equation}\label{Lemma:Growth1:E25}
	 	\begin{aligned}
	 		\mathscr M +\mathscr E  
	 		\ \geq \ &  \lim_{j\to\infty}  0.05C_{ini}  {\Big(\min\Big\{\frac{\mathfrak L_{{n_j}}}{40},r_0/2\Big\}\Big)^{c_{ini} }}  \exp\left\{\frac{\mathfrak C_2\mathfrak{L}_{{n_j}}^{\Big(\varepsilon-\frac2\alpha-\beta\Big)}}{4000}\Big[{{C_1}C_\omega^{\frac2\alpha} }\Big]^{-1}\right\}\ = \ \infty.
	 	\end{aligned}
	 \end{equation}
We obtain a contradiction and hence \eqref{Lemma:Growth1:4} follows. 
 
 	{\it (iv) Proof of \eqref{Lemma:Growth2:1}}.
 
	Using \eqref{Lemma:Growth1:1} and \eqref{Lemma:Growth1:4}, we obtain

		\begin{equation}\label{Lemma:Growth2:E1}
		\begin{aligned}
		\Big|	
		\mathscr{U}_n^T\cup
			\mathscr{V}_{n}^T\Big|  \ < \ & {\mathfrak{C}_1\mathfrak{L}_n^{(\beta+2)\Big[\varepsilon-\frac{2}{\alpha}\frac{\beta+1}{\beta+2}-\beta\frac{\beta+3}{\beta+2}-\frac{3 \varsigma}{\beta+2}\Big]}(	\mathscr M +\mathscr E)} + \mathfrak{C}_V {\mathfrak{L}_{n}^{\varepsilon-2\beta-2\varsigma}} \le \mathfrak{C}_3\mathfrak{L}_n^{\gamma},
		\end{aligned}
	\end{equation}
where $\mathfrak C_3>0$ is a universal constant that varies from line to line. Next, we
 write
	\begin{equation*}
\left\vert\bigcup_{i=n}^\infty	\big(\mathfrak{S}_i^T\backslash
			\mathscr{W}_{i}^T\big) \right\vert =		\left\vert \bigcup_{i=n}^\infty	\big(\mathscr{U}_n^T\cup
			\mathscr{V}_{n}^T\big)\right\vert =\sum_{i=n}^{\infty}	\Big|	
		\mathscr{U}_n^T\cup
			\mathscr{V}_{n}^T\Big| \leq \mathfrak{C}_3\sum_{i=n}^{\infty
		}\mathfrak{L}_i^{\gamma}=\frac{\mathfrak{C}_3}{1-2^{-\gamma}}\left(
		\mathfrak{L}_n\right) ^{\gamma}. 
	\end{equation*}

	 	{\it (v) Proof of \eqref{Lemma:Growth:3}}.

	Now, we will prove the existence of $\varsigma_*>0$ and $\mathscr{M}^*(\varsigma_*)>0$ such that for $n>\mathscr{M}^*(\varsigma_*)$  \begin{equation}\label{Lemma:Growth2:E1:2}\mathfrak{S}_n^T\subset\bigcup_{i=n}^\infty	\big(\mathfrak{S}_i^T\backslash
			\mathscr{W}_{i}^T\big).\end{equation}
		By \eqref{Lemma:Growth1:2}, it follows
		$$\bigcup_{i=n}^\infty\big(\mathfrak{S}_i^T\backslash\mathfrak{S}_{i+1}^T\big)\subset \bigcup_{i=n}^\infty	\big(\mathfrak{S}_i^T\backslash
			\mathscr{W}_{i}^T\big).$$ To prove \eqref{Lemma:Growth2:E1:2}, we  will show the existence of   $\varsigma_*>0$ and $\mathscr{M}^*(\varsigma_*)>0$ such that for $n>\mathscr{M}^*(\varsigma_*)$
			 \begin{equation}\label{Lemma:Growth2:E1:3}\mathfrak{S}_n^T\backslash \Big(\bigcup_{i=n}^\infty\big(\mathfrak{S}_i^T\backslash\mathfrak{S}_{i+1}^T\big)\Big)=\emptyset.\end{equation}

			 We perform a proof by contradiction. Suppose that for all $\varsigma>0$ there is a sequence $m(\varsigma)=m_0(\varsigma)< m_1(\varsigma)<\cdots$ such that 
			 $ \mathfrak{S}_{m_j(\varsigma)}^T\backslash \Big(\bigcup_{i=m_j(\varsigma)}^\infty\big(\mathfrak{S}_i^T\backslash\mathfrak{S}_{i+1}^T\big)\Big)\ne \emptyset$ for $j\in\{0,1,2,\cdots\}$. Choosing $t(\varsigma)\in \mathfrak{S}_{m(\varsigma)}^T\backslash \Big(\bigcup_{i=m(\varsigma)}^\infty\big(\mathfrak{S}_i^T\backslash\mathfrak{S}_{i+1}^T\big)\Big)\subset [0,T]$  we will prove that  \begin{equation}\label{Lemma:Growth2:E1:3}\int_{\left[ 0,\mathfrak{L}_{i}\right)
	}\mathrm{d}\omega F\left( t(\varsigma),\omega\right)  \geq \mathscr{C}_F|\mathfrak{L}_{i}|^\varsigma,\end{equation}
	for all $i\ge m(\varsigma)$. Suppose the contrary, that there exists $i_0\ge m(\varsigma)$ such that \eqref{Lemma:Growth2:E1:3} does not hold. Since \eqref{Lemma:Growth2:E1:3} holds for $i=m(\varsigma)$, we can suppose that \eqref{Lemma:Growth2:E1:3}  holds for $ i, i+1,\cdots, i_0-1$ and \ does not hold for $ i_0$. Hence, $t(\varsigma)\in \big(\mathfrak{S}_{i_0-1}^T\backslash\mathfrak{S}_{i_0}^T\big)$, yielding a contradiction. Therefore \eqref{Lemma:Growth2:E1:3} holds for all $i\ge m(\varsigma)$. 
	
	The same argument with \eqref{Lemma:GrowthFinal:E3} gives
	 \begin{equation}\label{Lemma:Growth2:E1:4}\int_{\left[ 0,2\mathfrak{L}_{i}\right)
	}\mathrm{d}\omega F\left( t,\omega\right)  \geq \mathscr{C}_F|\mathfrak{L}_{i}|^\varsigma/2,\forall t>t(\varsigma),\end{equation}
		for all $i\ge m(\varsigma)$.
Since $0<t(\varsigma)\le T$, we can set $$t_*=\sup_{\varsigma\in\Big(0,\frac12\min\Big\{\frac{\beta+2}{3}\Big[\varepsilon-\frac{2}{\alpha}\frac{\beta+1}{\beta+2}-\beta\frac{\beta+3}{\beta+2}\Big],\frac{\varepsilon-2\beta}{3}\Big\}\Big)}t(\varsigma),$$ then $0\le t_*\le T<T_0.$ Moreover,
 \begin{equation}\label{Lemma:Growth2:E1:5}\int_{\left[ 0,2\mathfrak{L}_{i}\right)
	}\mathrm{d}\omega F\left( t,\omega\right)  \geq \mathscr{C}_F|2\mathfrak{L}_{i}|^\varsigma/(2^{\varsigma+1}),\end{equation}
$\forall t>t_*, \forall	i\ge m(\varsigma), \forall \varsigma\in\Big(0,\frac12\min\Big\{\frac{\beta+2}{3}\Big[\varepsilon-\frac{2}{\alpha}\frac{\beta+1}{\beta+2}-\beta\frac{\beta+3}{\beta+2}\Big],\frac{\varepsilon-2\beta}{3}\Big\}\Big)$.
We choose $ \varepsilon$ such that $ -\varepsilon+2/\alpha+\beta=\varkappa$  is sufficiently small.
 We fix a sufficiently small $\varsigma$, $\exists m^*= m^*(\varsigma)$ such that 
\begin{equation}\label{Lemma:Growth2:E1:7}\int_{\left[ 0,\mathfrak{L}_{j}\right)
	}\mathrm{d}\omega F\left( t,\omega\right)\ > \ \mathscr{C}_F'|\mathfrak{L}_{j}|^\varsigma.\end{equation}
 $\forall  j> m^*(\varsigma)$ and for all  $t_*<t<T_0$, $\mathscr{C}_F'$ is a positive constant independent of $j,t$.

Given $t_*<T_1<T_0$, $j> m^*(\varsigma)$ and a time $t_*<t<T_1$, we will prove the existence of  an unbounded set $A(t)$ such that for all $j\in A(t)$, 
\begin{equation}\label{Lemma:Growth2:E1:8}\int_{\left[\mathfrak{L}_{j+1},\mathfrak{L}_{j}\right)
	}\mathrm{d}\omega F\left(t,\omega\right)\ > \ \frac{2^\varsigma-1}{2^\varsigma}\frac{\mathscr{C}_F'}{2}|\mathfrak{L}_{j}|^\varsigma,\end{equation}
	 while for $j\notin A(t)$, \eqref{Lemma:Growth2:E1:8} does not hold.
	We suppose there exists $M_0>0$ such that for all $j\ge M_0$,
	 \begin{equation}\label{Lemma:Growth2:E1:9}\int_{\left[\mathfrak{L}_{j+1},\mathfrak{L}_{j}\right)
	}\mathrm{d}\omega F\left( t,\omega\right)\ \le \  \frac{2^\varsigma-1}{2^\varsigma} \frac{\mathscr{C}_F'}{2}|\mathfrak{L}_{j}|^\varsigma.\end{equation}
	We deduce 
	 \begin{equation}\label{Lemma:Growth2:E1:10}\int_{\cup_{j=M_0}^\infty\left[\mathfrak{L}_{j+1},\mathfrak{L}_{j}\right)
	}\mathrm{d}\omega F\left( t,\omega\right)\ \le \ \frac{\mathscr{C}_F'}{2}\sum_{j=M_0}^\infty|\mathfrak{L}_{i}|^\varsigma=|\mathfrak{L}_{M_0}|^\varsigma\frac{\mathscr{C}_F'}{2}\frac{2^\varsigma}{2^\varsigma-1} \frac{2^\varsigma-1}{2^\varsigma}=|\mathfrak{L}_{M_0}|^\varsigma\frac{\mathscr{C}_F'}{2},\end{equation}
	yielding
 \begin{equation}\label{Lemma:Growth2:E1:11}\int_{\{0\}
	}\mathrm{d}\omega F\left( t,\omega\right)\ \ge \ |\mathfrak{L}_{M_0}|^\varsigma \mathscr{C}_F'\frac{1}{2}.\end{equation}
	This contradicts our original assumption and as thus this confirms the existence of $A(t)$. For each $n\in\mathbb{N}$, we define $A_n$ to be the set of all $t\in(t_*,T_1)$ such that $n$ belongs to $A(t)$. This implies $\cup_{n\in\mathbb{N}}A_n=(t_*,T_1)$. There must exist $q_0\in\mathbb{N}$ such that $|A_{q_0}|\ge |T_1-t_*|2^{-{q_0}\varkappa}\frac12\frac{ 2^\varkappa -1}{ 2^\varkappa}$. This can be seen by a contradiction argument. Suppose the contrary that $|A_{n}|\le |T_1-t_*|2^{-n\varkappa}\frac12\frac{ 2^\varkappa -1}{ 2^\varkappa}$ for all $n\in\mathbb{N}$, it immediately yields $|T_1-t_*|=\Big|\cup_{n\in\mathbb{N}}A_n\Big|\le \frac12\frac{ 2^\varkappa -1}{ 2^\varkappa} |T_1-t_*|\frac{ 2^\varkappa}{ 2^\varkappa -1}< |T_1-t_*|$. Therefore, the existence of $q_0$ is guaranteed. Since for each $t\in(t_*,T_1)$,  the set $A(t)$ is unbounded, the set $\{n\in A(t) | n>q_0\}$ is also unbounded. Therefore, $\cup_{n\in\mathbb{N},n>q_0}A_n=(t_*,T_1)$. Similarly as above,  there must exist $q_1\in\mathbb{N}, q_1>q_0$ such that $|A_{q_1}|\ge |T_1-t_*|2^{-{q_1}\varkappa}\frac12\frac{ 2^\varkappa -1}{ 2^\varkappa}$. By repeating this argument, we can construct an unbounded set $\{q_0,q_1,\cdots\}$ such that for all $q$ belongs to this set,  $|A_{q}|\ge |T_1-t_*|2^{-{q}\varkappa}\frac12\frac{ 2^\varkappa -1}{ 2^\varkappa}$.

	 Let $q\in \{q_0,q_1,\cdots\}$, $q$ large enough, there also exists $\Xi_{{i}}^{\mathfrak{h}_{q},\mathfrak{L}_{q}}\subset \left[\mathfrak{L}_{q+1},\mathfrak{L}_{q}\right)
$
	\begin{equation}\label{Lemma:Growth2:E1:12}\int_{\Xi_{{i}}^{\mathfrak{h}_{q},\mathfrak{L}_{q}}
	}\mathrm{d}\omega F\left(t,\omega\right)\ > \ \frac{\mathscr{C}_F'}{20}\frac{|\mathfrak{L}_{q}|^\varsigma}{\mathfrak{M}(q)}\ge \mathscr{C}_F''\mathfrak{L}_q^{\varsigma -\varepsilon+2/\alpha+\beta},\end{equation}
for some constant $\mathscr{C}_F''>0$  that varies from lines to lines and $t\in A_q$, yielding
\begin{equation}\label{Lemma:Growth2:E1:13}\int_{t_*}^{T_1}\mathrm{d}t \chi_{A_q}(t)\left[\int_{\Xi_{{i}}^{\mathfrak{h}_{q},\mathfrak{L}_{q}}
	}\mathrm{d}\omega F\left(t,\omega\right)\right]^2\ > \ \mathscr{C}_F''\mathfrak{L}_q^{2(\varsigma -\varepsilon+2/\alpha+\beta)+\varkappa}\ = \ \ \mathscr{C}_F''\mathfrak{L}_q^{3\varkappa+2\varsigma},\end{equation}
for some constant $\mathscr{C}_F''>0$  that varies from lines to lines. For sufficiently small $\varkappa, \varsigma$, 	we have $3\varkappa+2\varsigma<\varepsilon-\beta.$
Therefore, there must  exist a time  $ t_*<T_2<T_1$ such that: 
	\begin{equation}	\label{Lemma:Growth2:E1:14}
	\int_{t_*}^{T_2}\mathrm{d}t \chi_{A_q}(t)\left[\int_{\Xi_{{i}}^{\mathfrak{h}_{q},\mathfrak{L}_{q}}
	}\mathrm{d}\omega F\left(t,\omega\right)\right]^2   = \Big[{{C_1}C_\omega^{\frac2\alpha} }\Big]^{-1}\frac{\mathfrak{L}_{{q}}^{\varepsilon-\beta}}{4000\mho(\mathfrak{L}_{q}/2)}.
\end{equation}

Similar with \eqref{Lemma:Growth1:E18:2}, we bound 

\begin{equation}	\label{Lemma:Growth2:E1:15}
	\begin{aligned}
		&	\partial_{t}\left( \int_{\mathbb{R}^{+}}\mathrm{d}\omega F\rho\right) \\
		\  \ge \	&
		\int_{\mathbb{R}^{+}}\mathrm{d}\omega F\partial_{t}\rho \
		+ \  C_1C_\omega^{\frac2\alpha}\iiint_{\mathbb{R}_+^{3}}\mathrm{d}\omega_1\,\mathrm{d}\omega_2\,\mathrm{d}\omega \frac{FF_1F_2}{\mathfrak L_{q}^{\frac2\alpha}} \chi_{\left\{\omega_1\le\omega;\omega_1,\omega \in \Xi_{{i}}^{\mathfrak{h}_{q},\mathfrak{L}_{{q}}}\right\}}\chi_{{A}_{{q}}}(t) \\
		&\times 	\chi_{\left\{ \omega_{2}\leq\frac{\mathfrak L_{{q}} }{10}\right\}}    \Big\{ [-\rho(\omega_{2})+\rho(\omega+\omega_{2}-\omega_{1})] \mho(\mathfrak{L}_{q}/2) \ + \rho(\omega_{2}) \mho(\mathfrak{L}_{q}/2)\Big\}\\
		\  \ge \	&
		\int_{\mathbb{R}^{+}}\mathrm{d}\omega F\partial_{t}\rho \
		+ \  C_1C_\omega^{\frac2\alpha}\iiint_{\mathbb{R}_+^{3}}\mathrm{d}\omega_1\,\mathrm{d}\omega_2\,\mathrm{d}\omega \frac{FF_1F_2\mho(\mathfrak{L}_{q}/2)}{\mathfrak L_{q}^{\frac2\alpha}}\chi_{\left\{ \omega_{2}\leq\frac{\mathfrak L_{q} }{10}\right\}} \\
		&\times 	   \chi_{\left\{\omega_1\le\omega;\omega_1,\omega \in \Xi_{{i}}^{\mathfrak{h}_{q},\mathfrak{L}_{{q}}}\right\}}\chi_{{A}_{{q}}}(t)  \Big\{ [-\rho(\omega_{2})+\rho(\omega+\omega_{2}-\omega_{1})]   \Big\}\\
		& +\   C_1C_\omega^{\frac2\alpha}\iiint_{\mathbb{R}_+^{3}}\mathrm{d}\omega_1\,\mathrm{d}\omega_2\,\mathrm{d}\omega \frac{FF_1F_2}{\mathfrak L_{q}^{\frac2\alpha}}\chi_{\left\{ \omega_{2}\leq\frac{\mathfrak L_{{q}} }{10}\right\}}  \chi_{\left\{\omega_1\le\omega;\omega_1,\omega \in \Xi_{{i}}^{\mathfrak{h}_{q},\mathfrak{L}_{{q}}}\right\}}\chi_{{A}_{{q}}}(t)  \rho(\omega_{2}) \mho(\mathfrak{L}_{q}/2),
\end{aligned}\end{equation}
for a test function $\rho$. By the same argument with \eqref{Lemma:Growth1:E25}, we deduce a contradiction when $q\to\infty$. Therefore, \eqref{Lemma:Growth2:E1:3} holds  and \eqref{Lemma:Growth:3} is a consequence of  \eqref{Lemma:Growth2:1}.

\end{proof}

\subsection{Finite time condensation: The emergence of a delta function at the origin in finite time}

\begin{proposition}
	\label{Propo:FiniteTimeCondensation} Suppose that there exist arbitrary constants $C_{ini}>0,c_{ini}\ge 0$ and a small constant $r_0>0$ such that
	\begin{equation}\label{Lemma:FiniteTimeCondensation:1}
		\begin{aligned}
			\int_{0}^r\mathrm d\omega F(0,\omega)    
			\geq  C_{ini}  {r^{c_{ini}}}, \end{aligned}
	\end{equation}
	for all $r_0>r>0$ and there exists a sufficiently large natural number $N_0$ such that 
	\begin{equation}	\label{Lemma:FiniteTimeCondensation:2}
		\begin{aligned}
			&    \int_{0}^{\mathfrak{L}_{N_0+1}}\mathrm{d}\omega F(0,\omega)\ \ge \   \mathscr{C}_F^*=2\mathscr{C}_F|\mathfrak{L}_{N_0}|^{\varsigma_*}. 
	\end{aligned},\end{equation} 
	where $\varsigma_*$ is defined in  \eqref{Lemma:Growth:3}.
Then the first condensation time $T_0$ defined in \eqref{T0} is finite and can be bounded as
	\begin{equation}\label{Sec:FiniteTimeCondensation:3}
	T_0\ \leq\  {\mathfrak{C}_\mathfrak{S}}2 ^{-N_0\gamma(\varsigma_*)}.
\end{equation} 
\end{proposition}

\begin{proof}
	By Lemma \ref{Lemma:GrowthFinal}, we have
	\begin{equation}\label{Sec:FiniteTimeCondensation:E1}
		\mathfrak{S}_{N_0}^T= [0,T], 
	\end{equation} 
	for all $T\in[0,T_0)$. By \eqref{Lemma:Growth:3}, we bound
	\begin{equation}\label{Sec:FiniteTimeCondensation:E2}
		\left\vert 	\mathfrak{S}_{N_0}^T\right\vert  \leq{\mathfrak{C}_\mathfrak{S}}\left(
		\mathfrak{L}_{N_0}\right) ^{\gamma(\varsigma_*)},
	\end{equation} 	for all $T\in[0,T_0)$. Therefore
	\begin{equation}\label{Sec:FiniteTimeCondensation:E2}
	T_0\ \leq\ {\mathfrak{C}_\mathfrak{S}}2 ^{-N_0\gamma(\varsigma_*)}.
\end{equation} 
\end{proof}
\section{Proof of Theorem \ref{Theorem2}}
The proof of (i) follows from Proposition \ref{Propo:FiniteTimeCondensation}. 

Next, we will prove (ii). We use the same argument as \eqref{Lemma:GrowthFinal:E2} and bound
	\begin{equation}	\label{Theorem2:E1}
	\begin{aligned}
	&	m\Re \int_{[0,m\Re	)}\mathrm{d}\omega F(t,\omega)\  \ge \     \int_{\mathbb{R}^{+}}\mathrm{d}\omega F(t,\omega)\left[ \left(m\Re	-\omega\right)_+\right] \\
		\  \ge \  &   \  \int_{\mathbb{R}^{+}}\mathrm{d}\omega F(0,\omega)\left[ \left( m\Re	-\omega\right)_+ \right]\
		\  \ge \    \  \int_{0}^{m\Re	}\mathrm{d}\omega F(0,\omega)\left[  \left( m\Re	-\omega\right)_+ \right] \\
		\  \ge \   &  \  (m-1)\Re	\int_{0}^{\Re	}\mathrm{d}\omega F(0,\omega)	\ = \    (m-1) \Re	 \mathscr{M}_o,
\end{aligned}\end{equation}
yielding
\begin{equation}	\label{Theorem2:E2}
	\begin{aligned}
	 \int_{[0,m\Re	)}\mathrm{d}\omega F(t,\omega)\  \ge \  &       \frac{(m-1)   \mathscr{M}_o}{m}.
\end{aligned}\end{equation}

Applying \eqref{Propo:Collision} for $R=m\Re$, $\nu=\theta$, $h=\Re m/\mathscr N$, we find
\begin{equation}\label{Theorem2:E3}
	\begin{aligned}
	&	\frac{{C}_0(m\Re	)^{2+\frac2\alpha}(	\mathscr M +\mathscr E)}{(\Re m/\mathscr N)^{2} \mho(\Re m/\mathscr N)\theta^4}
		\ \ge\ &\int_{\digamma^*_{\mathscr N,\theta, m\Re}}\mathrm{d}t  \left( \int_{[0,m\Re)}\mathrm{d}\omega F(t,\omega)\right)^3	\\ \ge\ &   \mathcal M(\digamma^*_{\mathscr N,\theta, m\Re})  \left( \frac{(m-1)   \mathscr{M}_o}{m}\right)^3,
	\end{aligned}
\end{equation}
yielding
\begin{equation}\label{Theorem2:E4}
	\begin{aligned}
		\frac{{C}_0(m\Re	)^{2+\frac2\alpha}(	\mathscr M +\mathscr E)}{(\Re m/\mathscr N)^{2} \mho(\Re m/\mathscr N) \theta^4\left( \frac{(m-1)   \mathscr{M}_o}{m}\right)^3}
		\ \ge\ &   \mathcal M(\digamma^*_{\mathscr N,\theta, m\Re}).
	\end{aligned}
\end{equation}
\bibliographystyle{plain}

\bibliography{WaveTurbulence}

\def\cprime{$'$} \def\cprime{$'$} \def\cprime{$'$} \def\cprime{$'$}
  \def\cprime{$'$} \def\cprime{$'$}
\begin{thebibliography}{10}

\bibitem{AlonsoGambaBinh}
R.~Alonso, I.~M. Gamba, and M.-B. Tran.
\newblock The {C}auchy problem and {BEC} stability for the quantum
  {B}oltzmann-{G}ross-{P}itaevskii system for bosons at very low temperature.
\newblock {\em arXiv preprint arXiv:1609.07467}, 2016.

\bibitem{benney1969random}
D.~J. Benney and A.~C. Newell.
\newblock Random wave closures.
\newblock {\em Studies in Applied Mathematics}, 48(1):29--53, 1969.

\bibitem{benney1966nonlinear}
D.~J. Benney and P.~G. Saffman.
\newblock Nonlinear interactions of random waves in a dispersive medium.
\newblock {\em Proc. R. Soc. Lond. A}, 289(1418):301--320, 1966.

\bibitem{brout1956statistical}
R.~Brout and I.~Prigogine.
\newblock Statistical mechanics of irreversible processes part viii: general
  theory of weakly coupled systems.
\newblock {\em Physica}, 22(6-12):621--636, 1956.

\bibitem{buckmaster2019onthe}
T.~Buckmaster, P.~Germain, Z.~Hani, and J.~Shatah.
\newblock On the kinetic wave turbulence description for {NLS}.
\newblock {\em Quart. Appl. Math.}, 78(2):261--275, 2020.

\bibitem{buckmaster2019onset}
T.~Buckmaster, P.~Germain, Z.~Hani, and J.~Shatah.
\newblock Onset of the wave turbulence description of the longtime behavior of
  the nonlinear {Schr\"odinger} equation.
\newblock {\em Inventiones mathematicae}, 2021.

\bibitem{caffarelli2009regularity}
L.~Caffarelli and L.~Silvestre.
\newblock Regularity theory for fully nonlinear integro-differential equations.
\newblock {\em Communications on Pure and Applied Mathematics: A Journal Issued
  by the Courant Institute of Mathematical Sciences}, 62(5):597--638, 2009.

\bibitem{collot2024stability}
C.~Collot, H.~Dietert, and P.~Germain.
\newblock Stability and cascades for the kolmogorov--zakharov spectrum of wave
  turbulence.
\newblock {\em Archive for Rational Mechanics and Analysis}, 248(1):7, 2024.

\bibitem{collot2019derivation}
C.~Collot and P.~Germain.
\newblock On the derivation of the homogeneous kinetic wave equation.
\newblock {\em arXiv preprint arXiv:1912.10368}, 2019.

\bibitem{collot2020derivation}
C.~Collot and P.~Germain.
\newblock Derivation of the homogeneous kinetic wave equation: longer time
  scales.
\newblock {\em arXiv preprint arXiv:2007.03508}, 2020.

\bibitem{cortes2020system}
E.~Cort{\'e}s and M.~Escobedo.
\newblock On a system of equations for the normal fluid-condensate interaction
  in a bose gas.
\newblock {\em Journal of Functional Analysis}, 278(2):108315, 2020.

\bibitem{CraciunBinh}
G.~Craciun and M.-B. Tran.
\newblock A reaction network approach to the convergence to equilibrium of
  quantum {B}oltzmann equations for {B}ose gases.
\newblock {\em ESAIM: Control, Optimisation and Calculus of Variations}, 2021.

\bibitem{deng2019derivation}
Y.~Deng and Z.~Hani.
\newblock On the derivation of the wave kinetic equation for nls.
\newblock {\em arXiv preprint arXiv:1912.09518}, 2019.

\bibitem{deng2021propagation}
Y.~Deng and Z.~Hani.
\newblock Derivation of the wave kinetic equation: full range of scaling laws.
\newblock {\em arXiv preprint arXiv:2110.04565}, 2021.

\bibitem{deng2023long}
Y.~Deng and Z.~Hani.
\newblock Long time justification of wave turbulence theory.
\newblock {\em arXiv preprint arXiv:2311.10082}, 2023.

\bibitem{deng2022wave}
Y.~Deng, A.~D. Ionescu, and F.~Pusateri.
\newblock On the wave turbulence theory of 2d gravity waves, i: deterministic
  energy estimates.
\newblock {\em arXiv preprint arXiv:2211.10826}, 2022.

\bibitem{dolce2024convergence}
M.~Dolce and R.~Grande.
\newblock On the convergence rates of discrete solutions to the wave kinetic
  equation.
\newblock {\em arXiv preprint arXiv:2404.14400}, 2024.

\bibitem{dymov2021large}
A.~Dymov, S.~Kuksin, A.~Maiocchi, and S.~Vladuts.
\newblock The large-period limit for equations of discrete turbulence.
\newblock {\em arXiv preprint arXiv:2104.11967}, 2021.

\bibitem{dymov2020zakharov}
Andrey Dymov and Sergei. Kuksin.
\newblock On the {Z}akharov-{L'vov} stochastic model for wave turbulence.
\newblock In {\em Doklady Mathematics}, volume 101, pages 102--109. Springer,
  2020.

\bibitem{dymov2019formal}
Andrey Dymov and Sergei Kuksin.
\newblock Formal expansions in stochastic model for wave turbulence 1:
  {K}inetic limit.
\newblock {\em Comm. Math. Phys.}, 382(2):951--1014, 2021.

\bibitem{dymov2019formal2}
Andrey Dymov and Sergei Kuksin.
\newblock Formal {E}xpansions in {S}tochastic {M}odel for {W}ave {T}urbulence
  2: {M}ethod of {D}iagram {D}ecomposition.
\newblock {\em J. Stat. Phys.}, 190(1):Paper No. 3, 2023.

\bibitem{escobedo2023linearized1}
M.~Escobedo.
\newblock On the linearized system of equations for the condensate--normal
  fluid interaction at very low temperature.
\newblock {\em Studies in Applied Mathematics}, 150(2):448--456, 2023.

\bibitem{escobedo2023linearized}
M.~Escobedo.
\newblock On the linearized system of equations for the condensate-normal fluid
  interaction near the critical temperature.
\newblock {\em Archive for Rational Mechanics and Analysis}, 247(5):92, 2023.

\bibitem{escobedo2024instability}
M.~Escobedo and A.~Menegaki.
\newblock Instability of singular equilibria of a wave kinetic equation.
\newblock {\em arXiv preprint arXiv:2406.05280}, 2024.

\bibitem{EscobedoBinh}
M.~Escobedo and M.-B. Tran.
\newblock Convergence to equilibrium of a linearized quantum {B}oltzmann
  equation for bosons at very low temperature.
\newblock {\em Kinetic and Related Models}, 8(3):493--531, 2015.

\bibitem{EscobedoVelazquez:2015:FTB}
M.~Escobedo and J.~J.~L. Vel{\'a}zquez.
\newblock Finite time blow-up and condensation for the bosonic {N}ordheim
  equation.
\newblock {\em Invent. Math.}, 200(3):761--847, 2015.

\bibitem{EscobedoVelazquez:2015:OTT}
M.~Escobedo and J.~J.~L. Vel{\'a}zquez.
\newblock On the theory of weak turbulence for the nonlinear {S}chr\"odinger
  equation.
\newblock {\em Mem. Amer. Math. Soc.}, 238(1124):v+107, 2015.

\bibitem{EPV}
Miguel Escobedo, Federica Pezzotti, and Manuel Valle.
\newblock Analytical approach to relaxation dynamics of condensed {B}ose gases.
\newblock {\em Ann. Physics}, 326(4):808--827, 2011.

\bibitem{GambaSmithBinh}
I.~M. Gamba, L.~M. Smith, and M.-B. Tran.
\newblock On the wave turbulence theory for stratified flows in the ocean.
\newblock {\em M3AS: Mathematical Models and Methods in Applied Sciences. Vol.
  30, No. 1 105-137}, 2020.

\bibitem{germain2017optimal}
P.~Germain, A.~D. Ionescu, and M.-B. Tran.
\newblock Optimal local well-posedness theory for the kinetic wave equation.
\newblock {\em Journal of Functional Analysis}, 279(4):108570, 2020.

\bibitem{germain2023local}
P.~Germain, J.~La, and K.~Z. Zhang.
\newblock Local well-posedness for the kinetic mmt model.
\newblock {\em arXiv preprint arXiv:2310.11893}, 2023.

\bibitem{germain2024universality}
P.~Germain and H.~Zhu.
\newblock On universality for the kinetic wave equation.
\newblock {\em arXiv preprint arXiv:2402.14773}, 2024.

\bibitem{grande2024rigorous}
R.~Grande and Z.~Hani.
\newblock Rigorous derivation of damped-driven wave turbulence theory.
\newblock {\em arXiv preprint arXiv:2407.10711}, 2024.

\bibitem{halpern2009nonlinear}
L.~Halpern and J.~Szeftel.
\newblock Nonlinear nonoverlapping schwarz waveform relaxation for semilinear
  wave propagation.
\newblock {\em Mathematics of Computation}, pages 865--889, 2009.

\bibitem{hani2023inhomogeneous}
Z.~Hani, J.~Shatah, and H.~Zhu.
\newblock Inhomogeneous turbulence for wick nls.
\newblock {\em arXiv preprint arXiv:2309.12037}, 2023.

\bibitem{hannani2022wave}
A.~Hannani, M.~Rosenzweig, G.~Staffilani, and M.-B. Tran.
\newblock On the wave turbulence theory for a stochastic kdv type
  equation--generalization for the inhomogeneous kinetic limit.
\newblock {\em arXiv preprint arXiv:2210.17445}, 2022.

\bibitem{hasselmann1962non}
K.~Hasselmann.
\newblock On the non-linear energy transfer in a gravity-wave spectrum part 1.
  general theory.
\newblock {\em Journal of Fluid Mechanics}, 12(04):481--500, 1962.

\bibitem{hasselmann1974spectral}
K.~Hasselmann.
\newblock On the spectral dissipation of ocean waves due to white capping.
\newblock {\em Boundary-Layer Meteorology}, 6(1-2):107--127, 1974.

\bibitem{Lions:1989:OSA}
P.-L. Lions.
\newblock On the {S}chwarz alternating method. {II}. {S}tochastic
  interpretation and order properties.
\newblock In {\em Domain decomposition methods ({L}os {A}ngeles, {CA}, 1988)},
  pages 47--70. SIAM, Philadelphia, PA, 1989.

\bibitem{LukkarinenSpohn:WNS:2011}
J.~Lukkarinen and H.~Spohn.
\newblock Weakly nonlinear {S}chr\"odinger equation with random initial data.
\newblock {\em Invent. Math.}, 183(1):79--188, 2011.

\bibitem{ma2022almost}
X.~Ma.
\newblock Almost sharp wave kinetic theory of multidimensional kdv type
  equations with d >= 3.
\newblock {\em arXiv preprint arXiv:2204.06148}, 2022.

\bibitem{menegaki20222}
A.~Menegaki.
\newblock L2-stability near equilibrium for the 4 waves kinetic equation.
\newblock {\em arXiv preprint arXiv:2210.11189}, 2022.

\bibitem{Nazarenko:2011:WT}
S.~Nazarenko.
\newblock {\em Wave turbulence}, volume 825 of {\em Lecture Notes in Physics}.
\newblock Springer, Heidelberg, 2011.

\bibitem{newell2011wave}
A.~C. Newell and B.~Rumpf.
\newblock Wave turbulence.
\newblock {\em Annual review of fluid mechanics}, 43:59--78, 2011.

\bibitem{nguyen2017quantum}
T.~T. Nguyen and M.-B. Tran.
\newblock On the {K}inetic {E}quation in {Z}akharov's {W}ave {T}urbulence
  {T}heory for {C}apillary {W}aves.
\newblock {\em SIAM J. Math. Anal.}, 50(2):2020--2047, 2018.

\bibitem{ToanBinh}
Toan~T Nguyen and Minh-Binh Tran.
\newblock Uniform in time lower bound for solutions to a quantum boltzmann
  equation of bosons.
\newblock {\em Archive for Rational Mechanics and Analysis}, 231(1):63--89,
  2019.

\bibitem{Peierls:1993:BRK}
R.~Peierls.
\newblock Zur kinetischen theorie der warmeleitung in kristallen.
\newblock {\em Annalen der Physik}, 395(8):1055--1101, 1929.

\bibitem{Peierls:1960:QTS}
R.~E. Peierls.
\newblock Quantum theory of solids.
\newblock In {\em Theoretical physics in the twentieth century ({P}auli
  memorial volume)}, pages 140--160. Interscience, New York, 1960.

\bibitem{PomeauBinh}
Y.~Pomeau and M.-B. Tran.
\newblock Statistical physics of non equilibrium quantum phenomena.
\newblock {\em Lecture Notes in Physics, Springer}, 2019.

\bibitem{rumpf2021wave}
B.~Rumpf, A.~Soffer, and M.-B. Tran.
\newblock On the wave turbulence theory: ergodicity for the elastic beam wave
  equation.
\newblock {\em arXiv preprint arXiv:2108.13223}, 2021.

\bibitem{schwab2016regularity}
R.~W. Schwab and L.~Silvestre.
\newblock Regularity for parabolic integro-differential equations with very
  irregular kernels.
\newblock {\em Analysis \& PDE}, 9(3):727--772, 2016.

\bibitem{soffer2018dynamics}
A.~Soffer and M.-B. Tran.
\newblock On the dynamics of finite temperature trapped bose gases.
\newblock {\em Advances in Mathematics}, 325:533--607, 2018.

\bibitem{soffer2020energy}
A.~Soffer and M.-B. Tran.
\newblock On the energy cascade of 3-wave kinetic equations: beyond
  kolmogorov--zakharov solutions.
\newblock {\em Communications in Mathematical Physics}, 376(3):2229--2276,
  2020.

\bibitem{StaffilaniTranCascade1}
G.~Staffilani and M.-B. Tran.
\newblock On the energy transfer towards large values of wavenumbers for
  solutions of 4-wave kinetic equations.
\newblock {\em Submitted}.

\bibitem{staffilani2021wave}
G.~Staffilani and M.-B. Tran.
\newblock On the wave turbulence theory for a stochastic {KdV} type equation.
\newblock {\em arXiv preprint arXiv:2106.09819}, 2021.

\bibitem{toselli2004domain}
A.~Toselli and O.~Widlund.
\newblock {\em Domain decomposition methods-algorithms and theory}, volume~34.
\newblock Springer Science \& Business Media, 2004.

\bibitem{tran2020reaction}
M.-B. Tran, G.~Craciun, L.~M. Smith, and S.~Boldyrev.
\newblock A reaction network approach to the theory of acoustic wave
  turbulence.
\newblock {\em Journal of Differential Equations}, 269(5):4332--4352, 2020.

\bibitem{walton2023numerical}
S.~Walton and M.-B. Tran.
\newblock A numerical scheme for wave turbulence: 3-wave kinetic equations.
\newblock {\em SIAM Journal on Scientific Computing}, 45(4):B467--B492, 2023.

\bibitem{walton2024numerical}
S.~Walton and M.-B. Tran.
\newblock Numerical schemes for 3-wave kinetic equations: A complete treatment
  of the collision operator.
\newblock {\em arXiv preprint arXiv:2402.17481}, 2024.

\bibitem{walton2022deep}
S.~Walton, M.-B. Tran, and A.~Bensoussan.
\newblock A deep learning approximation of non-stationary solutions to wave
  kinetic equations.
\newblock {\em Applied Numerical Mathematics}, 2022.

\bibitem{zakharov2012kolmogorov}
V.~E. Zakharov, V.~S. L'vov, and G.~Falkovich.
\newblock {\em Kolmogorov spectra of turbulence I: Wave turbulence}.
\newblock Springer Science \& Business Media, 2012.

\bibitem{zaslavskii1967limits}
G.~M. Zaslavskii and R.~Z. Sagdeev.
\newblock Limits of statistical description of a nonlinear wave field.
\newblock {\em Soviet physics JETP}, 25:718--724, 1967.

\end{thebibliography}

\end{document}